\documentclass[11pt]{article}
\usepackage{tikz,tikz-cd}
\usepackage{subcaption}
\usepackage[ruled]{algorithm2e}
\usepackage{pgfplots}\pgfplotsset{compat=1.18}

\usepackage[utf8]{inputenc}
\usepackage[T1]{fontenc}
\usepackage[english]{babel}
\usepackage{amsmath,amssymb,amsthm}
\usepackage{graphicx} 
\usepackage[margin=1.in]{geometry}
\usepackage[shortlabels]{enumitem}
\usepackage{tikz,tikz-cd}
\usepackage{mathabx}
\usepackage{apptools}\AtAppendix{\counterwithin{theorem}{section}\counterwithin{equation}{section}}
\usepackage{todonotes}
\usepackage[bibstyle=numeric, backend=biber, sorting=none, citestyle=numeric-comp, giveninits,url=false, maxbibnames=5]{biblatex} 
\usepackage{csquotes}\MakeOuterQuote"
\usepackage{soul}\setstcolor{red}
\usepackage[colorlinks,citecolor=blue,linkcolor=blue]{hyperref}
\usepackage{cleveref}
\usepackage{comment}
\usepackage{booktabs}
\usepackage{multirow}

\newtheorem{theorem}{Theorem}
\newtheorem{lemma}[theorem]{Lemma}
\newtheorem{corollary}[theorem]{Corollary}
\newtheorem{proposition}[theorem]{Proposition}

\newtheorem{definition}[theorem]{Definition}
\newtheorem{introtheorem}{Theorem} 

\theoremstyle{definition}

\newtheorem*{problem*}{Problem}
\newtheorem*{assumption*}{Assumption}
\newtheorem{example}[theorem]{Example}
\newtheorem{remark}[theorem]{Remark}
\newtheorem*{warning*}{Warning}

\newcommand{\ip}[2]{\langle #1,#2\rangle}

\newcommand{\ketbra}[2]{|#1\rangle\langle#2|}
\newcommand{\ket}[1]{|#1\rangle}
\newcommand{\kettbra}[1]{\ketbra{#1}{#1}}
\newcommand{\bra}[1]{\langle#1|}
\DeclareMathOperator{\supp}{supp}
\newcommand{\norm}[1]{\lVert #1\rVert}
\newcommand{\oo}{\infty}
\newcommand{\ox}{\otimes}
\newcommand{\mc}{\mathcal}
\newcommand{\eps}{\varepsilon}
\newcommand{\III}{{\mathrm{III}}}
\newcommand{\II}{{\mathrm{II}}}
\newcommand{\I}{{\mathrm{I}}}

\DeclareMathOperator{\Aut}{Aut}
\DeclareMathOperator{\Out}{Out}
\DeclareMathOperator{\Inn}{Inn}
\newcommand{\abs}[1]{\lvert #1 \rvert}
\newcommand{\up}[1]{^{(#1)}}

\DeclareMathOperator{\tr}{Tr}
\DeclareMathOperator{\Tr}{Tr} 

\renewcommand{\tilde}{\widetilde}
\renewcommand{\hat}{\widehat}

\newcommand{\hide}[1]{}

\def\A{{\mc A}}

\def\B{{\mc B}}
\def\CC{{\mathbb C}}

\def\H{{\mc H}}
\def\K{{\mathcal K}}
\def\M{{\mc M}}
\def\N{{\mc N}}
\def\O{{\mc O}}

\def\U{{\mc U}}

\def\RR{{\mathbb R}}

\def\NN{{\mathbb N}}
\renewcommand\P{\mc P}
\def\ZZ{{\mathbb Z}}

\newcommand{\R}{\mc R}
\newcommand{\proj}{\mathrm{Proj}}

\def\locc{\xrightarrow{\LOCC}}
\def\barlocc{\xrightarrow{\overline\LOCC}}
\def\slocc{\xrightarrow{\SLOCC}}
\def\barslocc{\xrightarrow{\overline\SLOCC}}

\DeclareMathOperator{\lin}{span}
\DeclareMathOperator{\id}{id}

\DeclareMathOperator{\Ad}{Ad}

\DeclareMathOperator{\Sp}{Sp}
\DeclareMathOperator{\conv}{conv}

\newcommand{\LOCC}{\mathrm{LOCC}}
\newcommand{\SLOCC}{\mathrm{SLOCC}}

\newcommand{\placeholder}[0]{{\,\cdot\,}}
\newcommand{\piso}{\overline\U}

\newcommand{\DS}{\mathrm{DS}}
\newcommand{\DSS}{\mathrm{DSS}}

\newcommand{\qandq}{\quad\text{and}\quad}

\usepackage{tikz}
\newcommand{\no}{%
\tikz[scale=0.23] {
    \draw[line width=0.7,line cap=round] (0.0,0.05) to [bend left=4] (.9,1);
    \draw[line width=0.7,line cap=round] (0.1,0.95) to [bend right=2] (0.8,0.05);
}}
\newcommand{\yes}{%
\tikz[scale=0.23] {
    \draw[line width=0.7,line cap=round] (0.25,0) to [bend left=10] (1,1);
    \draw[line width=0.8,line cap=round] (0,0.35) to [bend right=1] (0.23,0);
}}

\title{Pure state entanglement and von Neumann algebras}
\author{Lauritz van Luijk, Alexander Stottmeister, Reinhard F.\ Werner, and Henrik Wilming}
\date{\small Institut f\"ur Theoretische Physik, Leibniz Universit\"at Hannover, \\ Appelstraße 2, 30167 Hannover, Germany\\[2ex]\today}

\begin{document}
\newgeometry{top=0.in} 
\maketitle
 \vspace{-.5cm}
\begin{abstract}
We develop the theory of local operations and classical communication (LOCC) for bipartite quantum systems represented by commuting von Neumann algebras. Our central result is the extension of Nielsen's Theorem, stating that the LOCC ordering of bipartite pure states is equivalent to the majorization of their restrictions, to arbitrary factors. As a consequence, we find that in bipartite system modeled by commuting factors in Haag duality, a) all states have infinite single-shot entanglement if and only if the local factors are not of type I, b) type III factors are characterized by LOCC transitions of arbitrary precision between any two pure states, and c) the latter holds even without classical communication for type III$_{1}$ factors. In the case of semifinite factors, the usual construction of pure state entanglement monotones carries over. Together with recent work on embezzlement of entanglement, this gives a one-to-one correspondence between the classification of factors into types and subtypes and operational entanglement properties. In the appendix, we provide a self-contained treatment of majorization on semifinite von Neumann algebras and $\sigma$-finite measure spaces.
\end{abstract}

\setcounter{tocdepth}{2}
\tableofcontents
\restoregeometry

\section{Introduction and Overview}
Entanglement is at the core of quantum information theory.
Its study has traditionally been restricted, for the most part, to systems of finitely many degrees of freedom, such as spin-1/2 particles or polarization degrees of freedom of photons (both modeled by finite-dimensional Hilbert spaces), or a finite number of continuous variable degrees of freedom, i.e., a finite number of bosonic modes. 
Already, in the latter case, technical difficulties appear, and some fundamental questions regarding the distillation or formation of entanglement have only very recently been generalized to the case of infinite dimensional Hilbert spaces \cite{yamasaki_entanglement_2024}.
Here, we study entanglement in systems with infinitely many degrees of freedom. This is motivated from several points of view: First, various quantum information theoretic protocols presuppose the existence of an unbounded number of entangled states (e.g., ebits) shared between two agents (henceforth Alice and Bob). While the actual number of ebits used will be finite in any given run of the protocol, one cannot give an a priori bound on the number. It is desirable to idealize such a scenario by one where Alice and Bob share a (countable) infinity of Bell pairs and to be able to use the mathematical tools that follow from such an idealization. Importantly, the mathematical framework should allow a given protocol in the idealized limit to be approximable in sufficiently large systems.
Second, tools and concepts from quantum information theory have by now diffused into other branches of physics, such as quantum many-body physics and quantum field theory. Understanding the entanglement structure of the latter is now a core topic at the frontier of research in both fields. It is desirable to have at one's disposal a mathematical framework that allows one to pose and answer questions regarding the entanglement content of quantum fields or many-body systems without explicit regularisation (in the case of quantum fields) or directly in the thermodynamic limit (in the case of quantum many-body systems).

The operational definition of bipartite entanglement in quantum information theory proceeds via the separated-labs paradigm: Alice and Bob are situated in separated labs and are allowed to implement arbitrary quantum operations locally but only to communicate via classical communication. 
This setup is commonly called the setting of \emph{local operations and classical communication} (LOCC).
By definition, entanglement is the property of bipartite quantum systems that cannot be created via LOCC. 
Studying entanglement in systems of infinitely many degrees of freedom requires an appropriate mathematical formulation of LOCC. Here, we provide such a formulation in a von Neumann algebraic framework and explore its consequences, focussing on entanglement theory for pure states. We refer to \cite{verch_distillability_2005} for a discussion of LOCC in a C*-algebraic framework, with a focus on the distillability of entanglement in quantum field theories.
\begin{figure}[h]
    \centering
    \includegraphics[width=.8\linewidth]{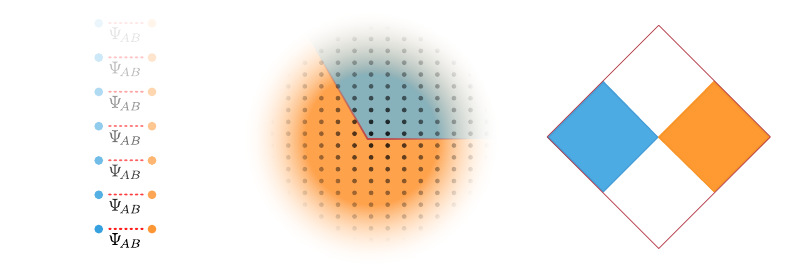}
    \caption{Examples for the applicability of our results, discussed in \cref{sec:examples}: \emph{Left:} Alice (blue) and Bob (orange) share an infinite supply of entangled qubits. \emph{Middle:} Alice and Bob act on different, infinite parts of a many-body system in the thermodynamic limit. \emph{Right:} Alice and Bob act on complementary wedges of Minkowski spacetime, depicted as a Penrose diagram.}
    \label{fig:enter-label}
\end{figure}
\subsection*{Set-up}
Although we discuss the general framework of LOCC in a multipartite setting in the main text, we restrict attention to the bipartite case here (see \cref{sec:general LOCC}).
We model a bipartite quantum system  by a pair of  commuting von Neumann algebras $(\M_A,\M_B)$ jointly acting on a Hilbert space $\H$. In systems with finitely many degrees of freedom, we have $\H= \H_A\ox\H_B$ and $\M_{A/B} \subseteq \B(\H_{A/B})$. 
A bipartite quantum system is purely quantum if the von Neumann algebra $\M_{AB}=\M_A\vee \M_B$, generated by $\M_A$ and $\M_B$, is a factor, i.e., if all elements of that commute with all other elements of $\M_{AB}$ are proportional to the identity. This is equivalent to $\M_A$ and $\M_B$ being factors. 
Factors can be classified into different types ($\I$, $\II$ and $\III$) and subtypes ($\I_n$ with $n\in\NN\cup\{\infty\}$, $\II_1,\II_\infty$ and $\III_\lambda$ with $\lambda\in[0,1]$). 
Type $\I_n$ factors are isomorphic to $\B(\H)$ for an $n$-dimensional Hilbert space, modeling quantum systems with finitely many degrees of freedom.
While all types may appear as subalgebras in the ground state sectors of many-body systems, see, e.g., \cite{matsui_split_2001, keyl_entanglement_2006, matsui_boundedness_2013, naaijkens_quantum_2017, ogata_type_2022, jones_local_2023, van_luijk_critical_2024}, local observable algebras in QFT are generically of type $\III$ \cite{buchholz_universal_1987, baumgaertel1995oam, yngvason_role_2005}. 
Our results show that bipartite systems of factors have very different entanglement properties relative to LOCC, depending on the type of the factors. 
We say that Haag duality holds if $\M_A=\M_B' = \{x\in \B(\H) : [x,b]=0, b\in \M_B\}$ (cp.\ \cite{keyl_entanglement_2006}). 
In the tensor product (type $\I$) setting $\H=\H_A\ox\H_B$, Haag duality simply means $\M_A=\B(\H_A)$, $\M_B=\B(\H_B)$.
Haag duality implies that $\M_A$ and $\M_B$ have the same type (but not necessarily the same subtype (for types I and II)).

\subsection*{Results}
Our first result concerns the structure of local operations. In the tensor product framework, a local operation by Alice is a quantum channel on Alice's Hilbert space $\H_A$, represented by Kraus operators in $\B(\H_A)$. In the general case, we have to distinguish between \emph{locality preserving} operations, which are quantum channels on the bipartite Hilbert space $\H$ (in the Heisenberg picture) such that
\begin{align}
    T(ab) = T(a)T(b),\quad a\in \M_A,b\in\M_B,
\end{align}
and \emph{local operations}. A local operation on Alice's side $T_A$ is a quantum channel on $\H$ such that
\begin{align}
    T_A(\M_A) \subset \M_A\quad \text{and}\quad T_A(b)=b,\ b\in \M_B.
\end{align}
A general local operation is a composition $T=T_A\circ T_B$ of Alice's and Bob's local operations.
\begin{introtheorem}[informal, see props.~\ref{prop:local_op} and \ref{prop:l=lp iff typeI}] 
    A local operation $T_A$ can be written as $T_A = \sum_x k_x(\placeholder) k_x^*$ with $k_x\in \M_A$ (similarly for Bob). In a bipartite system $(\M_A,\M_B)$ of approximately finite-dimensional\footnote{A factor is called approximately finite-dimensional or hyperfinite if it contains a weakly dense increasing net of matrix algebras $M_n(\CC)$. Essentially, von Neumann algebras appearing in physics always have this property.} factors, every locality preserving operation is a local operation if and only if $\H = \H_A\ox \H_B$ with $\M_A= \B(\H_A)\ox 1$ and $\M_B= 1\ox \B(\H_B)$ (i.e., $\M_A,\M_B$ have type $\I$). 
\end{introtheorem}
In \cref{sec:general LOCC}, we provide explicit examples demonstrating the distinction between locality preserving and local operations. That local operations should, in general, be modeled by those with local Kraus operators has previously also been observed in \cite{crann_state_2020}.
We define LOCC and stochastic LOCC (SLOCC) state transitions with respect to local operations and classical communication in \cref{Sec:locc-protocols}. For SLOCC, we show:
\begin{introtheorem}[informal, see thm.~\ref{thm:slocc}] Let $(\M_A,\M_B)$ be a bipartite system of factors in Haag duality on a Hilbert space $\H$. The following are equivalent:
\begin{enumerate}[(a)]
    \item $\Psi\in \H$ can be converted to $\Phi\in\H$ via SLOCC to arbitrary precision.
    \item $s_\phi \lesssim s_\psi$, where $s_\phi, s_\psi\in \M_A$ denote the support projections of the marginal states $\phi,\psi$ induced on $\M_A$ by $\Psi$ and $\Phi$, respectively.
\end{enumerate}
\end{introtheorem}

In the second item, $\lesssim$ denotes the ordering of projections in the sense of Murray and von Neumann \cite{murray_rings_1936}. 
Consequently, if $\M$ has type $\III$, all pure states are SLOCC equivalent up to arbitrarily small errors (since any two projections in von Neumann algebras of type III with separable predual are Murray-von Neumann equivalent).
If $\M_A$ and $\M_B$ are semifinite, the theorem implies that the generalized Schmidt rank $r(\Psi):=\Tr_{\M_A} s_{\psi_A} = \tr_{\M_B} s_{\psi_B}$ is a complete monotone for pure state SLOCC (see \cref{cor:schmidt-rank}).

Next, we generalize Nielsen's theorem to arbitrary factors:
\begin{introtheorem}[Nielsen's theorem, informal, see thm.~\ref{thm:nielsen}] Let $(\M_A,\M_B)$ be a bipartite system of factors in Haag duality on a Hilbert space $\H$. $\Psi\in \H$ can be transformed to $\Phi\in\H$ via LOCC with arbitrary precision if and only if 
\begin{align}\label{eq:intro maj}
\psi \in \overline{\mathrm{conv}}\{ u \phi u^*,\  u \in \U(\M_A)\} \quad\Leftrightarrow : \quad \psi \prec \phi,
\end{align}
where $\psi,\phi$ are the induced marginal states on $\M_A$, $\U(\M_A)$ denotes the unitary elements of $\M_A$ and the closure is taken with respect to the norm on the state space of $\M_A$. 
\end{introtheorem}
Since the assumption are symmetric under the exchange of $\M_A$ and $\M_B$, \eqref{eq:intro maj} holds equivalently with $\M_B$ instead of $\M_A$. Nielsen's theorem has recently been generalized to semifinite factors in \cite{crann_state_2020}. The type $\I_\oo$ case was obtained earlier in \cite{owari_convertibility_2008} (see also \cite{massri_locc_2024}). We provide a unified treatment covering all types, in particular including type III, which is particularly relevant in many-body physics and quantum field theory.

Using Nielsen's theorem, we show that on non-type $\I$ systems, all states contain an infinite amount of single-shot entanglement. The first two items of the following theorem generalize \cite[Prop.~4.9]{keyl_entanglement_2006} to the case that $\M_A,\M_B$ are not approximately finite-dimensional. It was already shown in \cite{keyl_infinitely_2003} that the type $\I$ case does not allow for infinite single-shot entanglement.
\begin{introtheorem}[informal, see thm.~\ref{thm:equivalence-bell}] Let $(\M_A,\M_B)$ be a bipartite system of factors in Haag duality on a Hilbert space $\H$. The following are equivalent:
\begin{enumerate}[(a)]
    \item For every $\Omega\in \H$, every $d\in\NN$ and every $\Psi\in \CC^d\ox \CC^d$ we have
    \begin{align}
        \Omega\ox (\ket 1 \ket 1) \xrightarrow{\LOCC} \Omega' \ox \Psi,\quad \Omega'\in\H.
    \end{align}
    \item $\M_A$ (hence $\M_B$) is not of type $\I$. 
\end{enumerate}
If $\M_A=\M_B'$ is approximately finite-dimensional, then both items are equivalent to:
\begin{enumerate}[resume*]
    \item Every density matrix $\rho$ on $\H$ \emph{maximally} violates the CHSH inequality.
\end{enumerate}
\end{introtheorem}
The theorem leaves open whether there is a distinction in terms of LOCC between type $\II$ and type $\III$. We show that in the type $\III$ case, bipartite pure state LOCC completely trivializes:
\begin{introtheorem}[informal, see thm.~\ref{thm:LOCC-trivial}] 
Let $(\M_A,\M_B)$ be a bipartite system of factors in Haag duality with $\M_A,\M_B\neq \CC$.
$\M_A$ has type $\III$ if and only if for all unit vectors $\Psi,\Phi \in \H$ we have that $\Psi \xrightarrow{\LOCC} \Phi$ to arbitrary accuracy.
$\M_A$ has type $\III_1$ if and only if the same is true with local operations and without classical communication. 
\end{introtheorem}
As a corollary of the theorem, we find that a bipartite system of factors in Haag duality is a \emph{universal LOCC embezzler} (cp.~\cite{short_paper,long_paper}, see also \cite{zanoni_complete_2023}) if and only if it is of type $\III$: For every $\Omega\in \H$ and $\Psi\in \CC^d\ox \CC^d$ we have
\begin{align}
    \Omega\ox (\ket 1 \ket 1) \xrightarrow{\LOCC} \Omega\ox\Psi
\end{align}
to arbitrary accuracy. 

To prove our results, we partly rely on majorization theory. In \cref{app:majorization}, we provide a self-contained treatment of majorization theory on $\sigma$-finite measure spaces and von Neumann algebras, which we believe to be of independent interest. 
Using majorization theory, we can also straightforwardly define entanglement monotones, such as R\'enyi entanglement entropies, in the case of semifinite factors \cite{crann_state_2020}. This is discussed in \cref{sec:monotones}.

Recently, we introduced a quantifier $\kappa_{max}$ that measures how well a bipartite system $(\M_A,\M_B)$ in Haag duality can be used as a resource to embezzle entanglement via local operations without classical communication \cite{van_dam2003universal,short_paper,long_paper}. For type $\III_\lambda$ factors $\M_A,\M_B$  with $0\leq \lambda\leq 1$, $\kappa_{max}$ allows to deduce the value of $\lambda$. Together with the results presented in this paper, we have now obtained a complete one-to-one correspondence between operational entanglement properties and the subtype of factors if $\M_A$ and $\M_B$ have the same subtype, see \cref{table:classification}. 
Note that Haag duality implies that the types of $\M_A,\M_B$ agree, but outside the type $\III$ case the subtype need not agree. 
As in finite dimension, entanglement based properties can only detect the type of the "smaller" of the two subsystems, see also \cref{sec:monotones}.

\begin{table*}[]\centering
\setlength{\tabcolsep}{7pt}
\renewcommand\arraystretch{1.2}
\begin{tabular}{@{} l cc cc ccc @{}}
\toprule
\multirow{2}{*}[-0.5\dimexpr \aboverulesep + \belowrulesep + \cmidrulewidth]{operational property} & \multicolumn{2}{c}{type I} & \multicolumn{2}{c}{type II} & \multicolumn{3}{c}{type III}\\
\cmidrule(lr){2-3} \cmidrule(lr){4-5} \cmidrule(l){6-8}
& I$_n$       & \ I$_\infty$   & \ II$_1$    & \ II$_\infty$ & \ III$_0$ & \ III$_\lambda$ & \ III$_1$ \\ 
 \midrule
one-shot entanglement           & \!\!\! $\le$$n$\!    & \!$<$$\infty$    & $\oo$&$\oo$ & $\oo$&$\oo$&$\oo$                                \\
maximally entangled state       & \yes &    \no          & \yes        & \no             & \yes&\yes&\yes          \\
all pure states LOCC equivalent & \no&\no     &\no&\no      & \yes&\yes&\yes         \\ 
some/all pure states LOCC-embezzling & \no&\no     &\no&\no      & \yes&\yes&\yes         \\ 
all pure states equally entangled& \no&\no&\no&\no & \yes & \yes & \yes\\
embezzling states             & \no&\no&\no&\no & (\,\yes\,) & \yes & \yes\\
worst embezzlement capability $\kappa_{\textit{max}}$ &  2 & 2    & 2&2       & 2 & \!\!$2\frac{1-\sqrt\lambda}{1+\sqrt\lambda}$\!\! & 0 \\ 
\bottomrule
\end{tabular}
\caption{Correspondence between operational entanglement properties and the type classification of factors with $n\in\NN$ and $0<\lambda<1$.
The inequalities for the single-shot entanglement in type $\I$ are bounds on the Schmidt rank of maximally entangled states that may be distilled via LOCC.
Maximally entangled states are discussed in \cref{sec:trivialization-III}.
We refer to \cite{short_paper,long_paper} for the formal definition of $\kappa_{\textit{max}}$ and the derivation of the formula for type $\III_\lambda$ factors. Since $\lambda\mapsto 2\frac{1-\sqrt\lambda}{1+\sqrt\lambda}$ is invertible, $\lambda$ is determined by the operational quantity $\kappa_{\textit{max}}$. 
Some $\III_0$ factors admit embezzling states while others do not.}
\label{table:classification}
\end{table*}

\paragraph{Acknowledgements.} We thank Niklas Galke and Albert H.\ Werner for useful discussions. We thank Fumio Hiai for helpful correspondence and providing a proof of \cref{lemma:fidelity-orbits}. We thank Jason Crann for making us aware of \cite{crann_state_2020} establishing Nielsen's theorem in the semifinite case after the first preprint of our work appeared. 

\paragraph{Funding.}
LvL and AS have been funded by the MWK Lower Saxony via the Stay Inspired Program (Grant ID: 15-76251-2-Stay-9/22-16583/2022).

\paragraph{Declarations.} 
The authors have no conflicts of interest to declare.
No data was generated or processed in this work.

\paragraph{Notation and standing conventions.}
We denote the set of unitary elements of a von Neumann algebra $\M$ by $\U(\M)$ and the set of partial isometries by $\piso(\M)$. All Hilbert spaces are assumed to be separable.
If $\M$ is a von Neumann algebra and $\Psi\in\H$, then $[\M\Psi]$ denotes both the closure of the subspace $\M\Psi\subset\H$ and the orthogonal projection onto it, which is a projection in $\M'$.
If $\M \subset \B(\H)$ is a von Neumann algebra and $\Psi\in\H$ a unit vector, we denote by $s_\psi= [\M'\Psi]$ the support projection of the induced state $\psi= \ip\Psi{(\placeholder)\Psi}$ on $\M$. Traces on semifinite von Neumann algebras are always normal, semifinite, and faithful. $\chi_A$ denotes the indicator function of a set $A$.

\section{Quantum systems described by von Neumann algebras}

\subsection{Basic setup}

A quantum system will be described by its observable algebra $\M$ acting on a separable Hilbert space $\H$.
The observables of the systems are described by self-adjoint elements $a=a^*$ in $\M$ and quantum states are given by the expectation value functionals $\omega:\M\to\CC$ induced by density operators $\rho$ on $\H$ via $\omega(a) = \tr a\rho$.
Throughout, we assume that the observable algebra $\M$ is \emph{von Neumann algebra}, which simply means that $\M$ contains the identity operator on $\H$ and that $\M$ is complete in the \emph{ultraweak topology}, i.e., the topology induced by the convergence of all expectation values induced by density operators on $\H$.
Both of these requirements are physically motivated: The identity operator corresponds to the binary measurement which always outputs 'yes' independent of the state, and the completion assumption is a basic requirement in general statistical theories \cite{ludwig_imprecision_1981,werner_uniformities_1983,ludwig_foundations_1983,lami_non-classical_2018}. 
The Hilbert space $\H$ plays the role of an environment in which the system described by $\M$ is embedded.
Importantly, the notion of state is independent of the description of the environment: Every faithful representation of $\M$ on a Hilbert space yields the same states on $\M$.
In fact, the state space admits an intrinsic description, which relies on the fact that a von Neumann algebra $\M$ has a unique predual $\M_*$: The states on $\M$ are precisely the elements $\omega$ of $\M_*$, viewed as functionals on $\M$, that are positive ($\omega(a^*a)\ge0$ for all $a\in\M$) and unital ($\omega(1)=1$).
These states are referred to as normal states on $\M$ to distinguish them from algebraic states on $\M$, which are general positive linear functionals $\omega:\M\to\CC$ with $\omega(1)=1$.%
\footnote{An algebraic state $\omega$ on a von Neumann algebra is normal if and only if it is $\sigma(\M,\M_*)$ continuous.}
To summarize, if the observable algebra of a quantum system is a von Neumann algebra $\M$, then the (physical) states of the system are described by normal states on $\M$. For this reason, all states appearing in the following will be normal, even if not explicitly stated.

It is not just quantum systems that can be described in this way.
The characteristic feature of classical systems is that all observables are jointly measurable and can be performed without perturbing the state of the system. 
Mathematically, this is reflected in the commutativity of the observables of classical systems.
Indeed, classical systems can be described by abelian von Neumann algebras, which are of the form $\M=L^\oo(X,\mu)$ for some measure space $(X,\mu)$.
The opposite case is that of a \emph{purely quantum system} in which any non-disturbing measurement is trivial.
Purely quantum systems are characterized as those systems whose observable algebra $\M$ is a factor $\M$, i.e., a von Neumann algebra with trivial center $Z(\M):=\CC1$ where the center $Z(\M)$ of a von Neumann algebra $\M$ is the (abelian) subalgebra of elements in $\M$ that commute with all other elements.%
\footnote{For instance, a system of two fermionic modes is not purely quantum since the parity can be measured without causing any perturbance. This is reflected by the fact that the observable algebra $\M$, given by the even part of the CAR algebra $\mathrm{CAR}(\CC^2)$, has center $Z(\M)= \langle (-1)^F\rangle\cong \CC^2$ where $(-1)^F$ is the parity operator.}
A general von Neumann algebra can decomposed as a direct integral of factors in an essentially unique way. This corresponds to the fact that every physical system is composed of classical and quantum degrees of freedom.\footnote{The existence of fundamentally classical degrees of freedom is debatable. However, von Neumann algebras can be used to describe effective/emergent systems which do contain classical degrees of freedom.}

If $\M$ is a von Neumann algebra on $\H$, then the same is true for the commutant
\begin{equation}
    \M' = \{b\in \B(\H) : [a,b]=0\ \forall a\in\M\}.
\end{equation}
The commutant sets up a basic duality for von Neumann algebras on a fixed Hilbert space. Indeed, we always have $\M''=\M$.%
\footnote{More generally, the bicommutant $\M''$ of a unital *-algebra $\M$ on $\H$ equals its weak closure \cite[Sec.~II.3]{takesaki1}.}
Since joint measurability of observables is equivalent to their commutativity, a bipartite quantum system is described by a commuting pair $(\M_1,\M_2)$ of von Neumann algebras on a Hilbert space $\H$.
For every von Neumann algebra, the pair $(\M,\M')$ is always a bipartite system, and the commutant $\M'$ is, by definition, the largest von Neumann algebra $\R$ on $\H$ such that $(\M,\R)$ is a bipartite system.
Therefore, we can regard the commutant $\M'$ as describing those physical degrees of freedom in the environment that are independent of the system.

We can directly generalize the above definition of a bipartite system to the case of multiple parties:

\begin{definition}[Multipartite systems]\label{def:multipartite system}
    Let $N$ be an integer. An \emph{$N$-partite system} on a Hilbert space $\H$ is a collection of $(\M_x)_{x=1}^N$ of $N$ pairwise commuting von Neumann algebras $\M_x$ on $\H$.
    An $N$-partite system is \emph{irreducible} if $\bigvee_x \M_x = \B(\H)$ and it satisfies \emph{Haag duality} if
    \begin{equation}
        \bigg(\bigvee_{x\in I}\M_x\bigg)' = \bigvee_{x\notin I} \M_x,\qquad \text{for all $I\subset [N]$}.
    \end{equation}
    A \emph{state} on the multipartite system $(\M_x)_{x=1}^N$ is a is a normal state $\omega$ on $\bigvee_x \M_x$.
    In the case of an irreducible system, we identify these with the corresponding density operator $\rho$ on $\H$ such that $\omega=\tr(\placeholder)\rho$, and we identify unit vectors $\Omega\in\H$ with the induced state $\omega = \bra\Omega\placeholder\ket\Omega$.
\end{definition}

\begin{lemma}[{\cite[Lem.~3]{van_luijk_multipartite_2024}}]
    Let $(\M_x)_{x\in[N]}$ be a multipartite system of factors on $\H$. Then Haag duality implies irreducibility.
    Moreover, every irreducible multipartite system consists of factors.
\end{lemma}

In case of bipartite systems, $(\M_A,\M_B)$ on a Hilbert space $\H$, irreducibility is equivalent to $\M_A\subset \M_B'$ being an irreducible subfactor inclusion.\footnote{A subfactor $\R\subset \N$ is irreducible if the relative commutant is trivial $\R'\cap\N=\CC1$.}
An important special class of bipartite systems are standard bipartite systems, which we study in the following:

\subsection{Standard bipartite systems}
\label{sec:standard}
Standard bipartite systems are a special class of bipartite systems $(\M_A,\M_B)$ where Alice's and Bob's subsystems are connected via an exchange symmetry.
A finite-dimensional bipartite system described by a product Hilbert space $\H=\H_A\ox\H_B$ and local algebras $\M_A = \B(\H_A)\ox 1, \M_B = 1\ox \B(\H_B)$ is standard precisely when the local Hilbert spaces of Alice and Bob have the same dimension, i.e., $\dim\H_A=\dim\H_B$.

To define standard bipartite systems, we briefly recall a few facts about von Neumann algebras (see \cite{haagerup_standard_1975} or \cite[Sec.~IV.1]{takesaki2} for details).
A \emph{standard representation} of a von Neumann algebra $\M$ is a faithful representation $\M\subset\B(\H)$ admitting a cyclic separating vector $\Omega\in \H$, i.e., a vector such that $\M\Omega$ is dense in $\H$ and such that $\omega(a)=\ip\Omega{a\Omega}$ defines a faithful state on $\M$.\footnote{Strictly speaking, this definition only applies to $\sigma$-finite von Neumann algebras. Since separable von Neumann algebras are $\sigma$-finite, this definition works in our case.}
It is well-known that all standard representations are unitarily equivalent.
Standard representations have a surprising amount of structure:
Given a cyclic separating vector $\Omega\in\H$, one can canonically construct a conjugation $J$ on $\H$ such that\footnote{A conjugation on a Hilbert space $\H$ is an anti-unitary operator $J:\H\to\H$ with $J^2=1$.}\textsuperscript{,}\footnote{$J$ is the anti-unitary arising from the polar decomposition of the closable operator $S_0(a\Omega)=a^*\Omega$ on $D(S_0)=\M\Omega$.}
\begin{equation}\label{eq:std form1}
    J\M J =\M',\qquad J\Omega=\Omega,\qquad JaJ=a^*, \ a\in Z(\M).
\end{equation}
Furthermore, by defining $\P = \overline{\{ aJa\Omega\ :\ a\in \M\}}$, one obtains a self-dual positive cone which satisfies 
\begin{equation}\label{eq:std form2}
    aJaJ \P\subset \P,\qquad J\Psi = \Psi, \ \Psi\in\P.
\end{equation}
The triple $(\H,\P,J)$ of a representation $\M\subset\B(\H)$, a self-dual cone $\P\subset\H$ and a conjugation $J$ is uniquely specified up to unitary isomorphism by \cref{eq:std form1,eq:std form2} \cite{haagerup_standard_1975}.
The standard form enjoys the property that, for each normal state $\varphi$ on $\M$ there is a unique $\Omega_\varphi\in\P$ such that
\begin{equation}
    \varphi(a) = \ip{\Omega_\varphi}{a\Omega_\varphi},\qquad a\in\M.
\end{equation}
In fact, the map $\omega\mapsto \Omega_\omega$ is a homeomorphism for the respective norm topologies as it satisfies \cite{haagerup_standard_1975}:
\begin{equation}\label{eq:purification_estimate}
    \norm{\Omega_\omega-\Omega_\varphi}^2 \le \norm{\omega-\varphi} \le \norm{\Omega_\omega-\Omega_\varphi}\cdot \norm{\Omega_\omega+\Omega_\varphi}, \qquad \omega,\varphi\in\M_*^+.
\end{equation}
Furthermore, there is a map $\alpha \mapsto u_\alpha$ which implements an automorphism $\alpha$ of $\M$ by a unitary $u_\alpha$ on $\H$, i.e., one has $\alpha(a)=u_\alpha^*au_\alpha$ for all $a\in\M$, and the implementation is such that
\begin{equation}\label{eq:unitary implementation}
    u_\alpha J = J u_\alpha,\qquad u_\alpha \Omega_\varphi = \Omega_{\varphi\circ\alpha}
\end{equation}
for all normal states $\varphi$ on $\M$. The map $\alpha\mapsto u_\alpha$ is a homeomorphism onto its range for the $u$-topology on $\Aut\M$ and the strong operator topology on the unitary group of $\H$.\footnote{The $u$-topology on $\Aut\M$ is the topology of pointwise norm-convergence on the state space, i.e., $\alpha_i \to \alpha$ in the $u$-topology if and only if $\norm{\omega\circ\alpha_i-\omega\circ\alpha}\to 0$ for all normal states $\omega$ on $\M$.}

\begin{proposition}\label{prop:std}
    Let $(\M_A,\M_B)$ be a bipartite system on a (separable) Hilbert space $\H$ such that Haag duality $\M_A=\M_B'$ holds.
    The following are equivalent:
    \begin{enumerate}[(a)]
        \item\label{it:std1}
            $\M_A$ (and hence $\M_B$) is in standard representation, i.e., there exists a cyclic separating vector.
        \item\label{it:std2}
            Existence of purifications: For every normal state $\omega$ on $\M_A$ there exists a vector $\Omega\in\H$ such that $\omega(a)=\ip\Omega{a\Omega}$, $a\in\M_A$, and $\M_B$ has the same property.
        \item\label{it:std3}
            Exchange symmetry: There is a conjugation $J$ on $\H$ such that $J\M_AJ=\M_B$ and $JaJ=a^*$ for all $a\in\M_A\cap\M_B$.
    \end{enumerate}
\end{proposition}
\begin{proof}
    The equivalence between \cref{it:std1,it:std2} is shown in \cite[Lem.~18]{long_paper}.
    \ref{it:std1} $\Rightarrow$ \ref{it:std3} is explained above.
    \ref{it:std3} $\Rightarrow$ \ref{it:std1}: 
    Set $\R=\M_A$. Let $\Psi$ be a nonzero $J$-invariant unit vector in $\H$ and let $\psi$ be the induced normal state on $\R$.
    Let $k_n\in \R$ be sequence of elements such that $\sum k_n^*k_n=1$ and such that the state $\omega=\sum k_n\psi k_n^*$ is faithful.%
    \footnote{To construct such a sequence, consider the support projection $p_1 = s_\psi$ and let $p_2,\ldots p_r$ be a collection of Murray-von Neumann-equivalent projections such that $p_1\vee\ldots p_r=1$. Pick partial isometries $v_n$ such that $v_n^*v_n=p_1$ and $v_vv_n^*=p_n$. Set $k_n = 2^{-n} v_n$ (where $v_1=p_1$) and $k_0= (1- \sum_{n=1}^r 4^{-n} p_n)^{1/2}$ such that $\sum_{n=0}^r k_n^*k_n=1$. By construction $\omega = \sum_n k_n\psi k_n^*$ is faithful. Indeed, $\omega \ge v_n \psi v_n^*$ for all $n=1,\ldots r$ implies $s_\omega\ge p_1\vee \ldots p_r=1$.}
    Set $\Omega = \sum_n k_n Jk_n \Psi$.
    Then $\omega(a)=\ip\Omega{a\Omega}$ for all $a\in\R$. Hence $\Omega$ is a separating vector and, since $J\Omega =\Omega$, we know that $\Omega$ is also separating for $J\R J=\R'$.
    Thus $\Omega$ is also cyclic for $\M_A$. Therefore $\R$ is in standard representation. 
    In particular, this shows that every $J$-invariant vector $\Psi\ne0$ induces a positive cone $\P$, such that $(\H,J,\P)$ is a standard from, via $\P=\overline{\lin}\,\{aJa\Psi\ :\ a\in\R\}$.
\end{proof}

\begin{definition}[{\cite{long_paper}}] \label{def:standard}
    A bipartite system $(\M_A,\M_B)$ on a Hilbert space $\H$ is \emph{standard} if the equivalent properties of \cref{prop:std} hold.
\end{definition}

Note at this point that a standard bipartite system need not necessarily be irreducible.
In fact, this is the case if and only if $\M_A$ (and hence $\M_B$) is a factor.
For a bipartite system $(\M_A,\M_B)$ of factors in Haag duality, being standard is essentially saying that $\M_A$ and $\M_B$ have the same size. For example, if both factors are infinite (types $\I_\oo$, $\II_\oo$ or $\III$), they are automatically standard \cite[Cor.~III.2.6.16]{blackadar_operator_2006}.
For finite factors, the relative size of $\M_A$ and $\M_B$ is measured by the coupling constant $c(\M_A,\M_B)$, and being standard is equivalent to $c(\M_A,\M_B)=1$ (see \cref{sec:monotones}).
If the bipartite system is not standard, one factor will be larger than the other one, and, by truncating it, one can obtain a standard bipartite system.

Since the algebras $\M_A$ and $\M_B$ in a standard bipartite system are (anti)-isomorphic, we will use the terminology for von Neumann algebras also for standard bipartite systems.
For example, an irreducible approximately finite-dimensional type $\II_\oo$ standard bipartite system is a standard bipartite system $(\M_A,\M_B)$ where $\M_A$ and $\M_B$ are approximately finite-dimensional type $\II_\oo$ factors. 
Since the latter is unique, this means that we consider the hyperfinite type $\II_\oo$ factor with its commutant in standard representation.

The following two results often allow for a reduction to the case of standard bipartite systems:

\begin{lemma}[{\cite[Prop.~23]{long_paper}}]\label{lem:minimal subspace}
    Let $(\M,\M')$ be a bipartite system of factors in Haag duality on $\H$.
    Let $\Omega\in\H$ be a unit vector and let $s_\omega=[\M\Omega]$ and $s_{\omega'}=[\M'\Omega]$ be the support projections of the reduced states.
    Set $e=s_\omega s_{\omega'}$, $\H_0=e\H$ and $\M_0$.
    Then $(\M_0,\M_0')$ is a standard bipartite system on $\H_0$ and $\Psi$ is a cyclic separating vector for $\M_0$ and $\M_0'$.
\end{lemma}

\begin{lemma}[{\cite[Cor.~III.2.6.16]{blackadar_operator_2006}}]\label{lem:tensor-standard}
    Let $(\M,\M')$ be a bipartite system of factors in Haag duality on $\H$. 
    Then the bipartite system of factors $(\tilde \M,\tilde \M')=(\M\ox\B(\K)\ox 1,\M'\ox1\ox \B(\K))$ on $\H\ox\K\ox\K$ is a standard bipartite system if $\dim\K=\oo$.
\end{lemma}

\subsection{Three examples}\label{sec:examples}

We give three examples of physical systems naturally formulated in these terms for readers who are not yet accustomed to the von Neumann algebraic description of infinite quantum systems.
The first example considers an idealized resource in quantum information theory, the second example concerns ground state sectors of quantum many-body systems in the thermodynamic limit, and the third example are the observable algebras in quantum field theory.

\subsubsection{Infinitely many entangled $N$-qudit systems}
\label{sec:oo-qudits}
We consider the situations where $N$ parties share a countably infinity of entangled $N$-qudit systems.
Let $\Phi_n \in (\CC^d)^{\ox N}$ be the state of the $n$th system.
The Hilbert space of the full system is the infinite Hilbert space tensor product 
\begin{equation}
    \H= \bigotimes_{n\in\NN} \big((\CC^d)^{\ox N};\Phi_n\big)
\end{equation}
with respect to the sequence of reference states $(\Phi_n)_{n\in\NN}$ \cite[Sec.~XIV.1]{takesaki3}.
We obtain $N$ commuting von Neumann algebras, corresponding to the $N$ parties, via
\begin{equation}
    \M_x = \overline{ \bigotimes_{n\in\NN} \big(1^{\ox x-1}\ox M_d(\CC)\ox 1^{N-x}\big)}^{w}.
\end{equation}
This defines an irreducible $N$-partite system of factors $(\M_x)_{x\in[N]}$ satisfying Haag duality \cite{van_luijk_multipartite_2024}. 
We will refer to this construction as an \emph{ITPFI multipartite system} (where ITPFI is short for "infinite tensor product of finite type $\I$") \cite{araki_classification_1968}.
The state of the full system is described by the vector $\Omega = \otimes_{n\in\NN}\Omega_n \in \H$.
This construction may be generalized by making the integer $d$ depend on $n$ (or even $x$).
Since the local factors are ITPFI factors
\begin{equation}
    \M_x\cong \bigotimes_{n\in\NN}\, (M_d(\CC); \phi_{x,n}), \qquad \phi_{x,n} = \tr_{[N]\setminus\{x\}} \kettbra{\Phi_n},
\end{equation}
their type can be computed from the asymptotic behavior of the spectrum of the reduced states $\phi_{x,n}$ \cite{araki_classification_1968}. 

In the case of two parties, the resulting bipartite system $(\M_A,\M_B)$ on $\H$ is guaranteed to be a standard bipartite system \cite{van_luijk_multipartite_2024}.
To be concrete, for $0\leq \lambda\leq 1$ consider a two-qubit pure state $\Psi_\lambda \in \CC^2\ox\CC^2$ of the form
\begin{align}
    \Psi_\lambda = \frac{1}{\sqrt{1+\lambda}}\big(\ket1\ox\ket 1 + \sqrt{\lambda} \ket 2\ox\ket 2\big).
\end{align}
Then, the ITPFI bipartite system obtained by choosing $\Phi_n = \Psi_\lambda$ for all $n$ is given by $(\M_A,\M_B) = (\R_\lambda,\R_\lambda')$, where $\R_\lambda$ is known as a \emph{Powers factor} \cite{powers_representations_1967}. 
It has type $\I_\infty$ for $\lambda=0$, type $\III_\lambda$ for $0<\lambda<1$ and type $\II_1$ for $\lambda=1$.
In particular, the latter case corresponds to an infinite supply of Bell states shared between Alice and Bob.

\subsubsection{Ground state sectors of quantum many-body systems}

We consider a quantum many-body system in the thermodynamic limit, which consists of infinitely many qudits positioned on the sites of a lattice $\Gamma$.
We denote by $\A_\Gamma= \bigotimes_{x\in\Gamma} M_d(\CC)$ the corresponding quasi-local algebra. 
To every region $X\subset\Gamma$, there is a naturally associated subalgebra $\A_X$ given by $\bigotimes_{x\in X}M_d(\CC)$ (by tensoring with the identity on all remaining sites).
For disjoint regions $X,Y$, the subalgebras $\A_X$ and $\A_Y$ commute and for $X\subset Y$ we have $\A_X\subset\A_Y$.
The $C^*$-algebra $\A_X$ is the norm-closure of the *-algebra of observables that are supported on finite subsets of $X$.
We can summarize this by saying that, for every region $X$ (even $X=\Gamma$), the elements of $\A_X$ have approximately finite support (or are `quasi-local').
For instance, if $\Gamma = \ZZ$ the algebra $\A_\NN$ will not contain elements like $\otimes_{x\in \NN} u$ where $1\neq u\in M_d(\CC)$ is some fixed unitary. It will, however, contain the unitary $\otimes_{x\in\NN} u_x$ if $(u_x)$ is a sequence of unitaries in $M_d(\CC)$ such that $u_x\to 1$ sufficiently fast as $x\to \oo$.\footnote{For example, $\sum_x \norm{1-u_x} <\infty$ is sufficient.}

If $\omega$ is a state on $\A_\Gamma$, we denote by  the GNS representation of $\omega$.
In the case where $\omega$ is a ground state of a local Hamiltonian $H$, we refer to $(\H,\pi,\omega)$ as a ground state sector of $H$.
Since $\A_\Gamma$ is a simple $C^*$-algebra (has no proper ideals), the representation $\pi$ is faithful, and we can suppress it by regarding $\A_\Gamma$ as a subalgebra of $\B(\H)$.
Taking weak closures, we obtain, for every region $X\subset \Gamma$, a von Neumann algebra 
\begin{equation}
    \M_X = \overline{\A_X}^w \subset \B(\H).
\end{equation}
Since the elements of $\A_X$ are supported in $X$, the same holds for the weak closure $\M_X$. 
Since the weak closure and the strong closure of $\A_X$ both give the same von Neumann algebra $\M_X$, the elements of $\M_X$ are also approximable by operators with finite support in the strong operator topology: For each $\Psi\in\H$ and $a\in \M_X$ we can approximate $a\Psi$ up to error $\eps>0$ by $a_\eps \Psi$ with $a_\eps\in \A_X$ having finite support in $X$.
Again, it follows that $\M_X$ and $\M_Y$ commute if $X$ and $Y$ are disjoint regions and that $\M_X\subset\M_Y$ whenever $X\subset Y$.
Moreover, we have $\M_X\vee\M_Y = \M_{X\cup Y}$.
For any pairwise disjoint collection of regions $X_1,\ldots X_N\subset \Gamma$, the algebras $(\M_{X_1},\ldots \M_{X_N})$ form an $N$-partite system on $\H$.

If the state $\omega$ is pure, e.g., if $\omega$ is the unique ground state of a local Hamiltonian $H$, we will have $\M_\Gamma=\B(\H)$.
For a partition of $\Gamma$ into $N$ subset $X_1,\ldots X_N$, this guarantees that the multipartite system $(\M_{X_1},\ldots \M_{X_N})$ is irreducible.
We remark that Haag duality is notoriously difficult to prove -- even in the case of a mere bipartition $\Gamma = X_1\cup X_2$ of the lattice.\footnote{A purported proof of Haag duality for half-chain algebras associated with translation invariant pure states in \cite{keyl2008haag} turned out to be flawed.}

\subsubsection{Observable algebras in quantum field theory}

In the algebraic approach, a quantum field theory on a spacetime $M$ is modeled by a net of von Neumann algebras
\begin{equation}
    \O \mapsto \A(\O), \qquad \O\subset M,
\end{equation}
on a common Hilbert space $\H$ \cite{haag_local_1996}. Here, $\O$ denotes an open spacetime region, and the algebras $\A(\O)$ describe the observables localized in $\O$.
The net is required to satisfy certain axioms -- most importantly, \emph{causality}, which means that causally separated regions $\O_1\bigtimes \O_2$ yield commuting von Neumann algebras $\A(\O_1) \subset \A(\O_2)'$. Another important requirement is that of \emph{isotony}, which means that $\A(\O_1)\subset\A(\O_2)$ whenever $O_1\subset \O_2$.
Typically, the Hilbert space $\H$ is assumed to contain reference vector $\Omega$, representing the "vacuum state".\footnote{On curved spacetime and in the absence of a time-translation symmetry (a timelike Killing vector field), there is, in general, no canonical choice of reference state  \cite{fewster_algebraic_2015}.}
There are a variety of additional axioms (or rather properties) depending on the specific setting to be addressed. Among the most common are \emph{causal completeness}, which states that the observable algebra of a region $\O$ coincides with the observable algebra of its causal completion $\O''$, i.e., $\A(\O) = \A(\O'')$, where $\O'$ denotes the causal complement (and hence $\O''$ is the causal completion of $\O$), \emph{additivity}, which means that $\A(\O_1)\vee\A(\O_2) = \A(\O_1\cup\O_2)$ holds for collections of open regions $\O_x$, and \emph{Haag duality}, which means that the algebra associated to the causal complement $\O'$ of an open subset $\O$ is the commutant
\begin{equation}
    \A(\O)' = \A(\O').
\end{equation}
Notice that Haag duality is a strengthened version of causal completeness.
In almost all cases, the von Neumann algebras $\A(\O)$ are approximately finite-dimensional type $\III$ (typically $\III_1$) algebras or even factors \cite{buchholz_universal_1987,baumgaertel1995oam}. It was recently argued, however, that local algebras may in fact be of type $\II$ if semiclassical gravity is taken into account, see for example \cite{witten2022crossed_product,chandrasekaran_algebra_2023,fewster_quantum_2024,de_vuyst_gravitational_2024,chen_clock_2024}.

Thus, a collection of pairwise causally separated regions $\O_1,\ldots,\O_N\subset M$ generically yields a multipartite quantum system $(\A(\O_x))_{x=1}^{N}$ on $\H$ with a natural multipartite state given by the vacuum vector $\Omega$.
This multipartite system is irreducible (see \cref{def:multipartite system}) whenever the sets $\O_x$ partition spacetime in the sense that $\big(\bigcup \O_x\big)'' = M$. It satisfies Haag duality (in the sense of $N$-partite system, see \cref{def:multipartite system}) if the net $\O \mapsto \A(\O)$ satisfies it.
For instance, in a quantum field theory on Minkowski spacetime, which satisfies Haag duality, the bipartition of $M$ into complementary wedges $W$ and $W'$ yields a standard bipartite system $(\M,\M')=(\A(W),\A(W'))$.

We caution that while the QFT setting fits the formal framework we develop here, some care must be taken in its interpretation: For once, if Alice and Bob act on spacelike separated regions, they clearly cannot communicate classically. Conversely, if they can communicate classically, they cannot act on spacelike separated regions, and their local observables will not commute (see also \cite{verch_distillability_2005}). As we will see, this problem is, in a sense, solved automatically: If the local algebras are of type $\III_1$, we will see that not only LOCC trivializes, but that all pure states may be mapped to each other via local operations \emph{without communication} (cp.~\cref{thm:LOCC-trivial}).
Second, it is not clear whether in relativistic QFT, indeed, all elements of the local operator algebras should be considered valid Kraus operators describing implementable operations (see, for example, \cite{fewster2020local_measurements} and reference therein for recent discussions of this point).

\subsection{Operations and their implementability}\label{sec:operations}

In this section, we discuss operations on a quantum system described by a von Neumann algebra $\M$ acting on a Hilbert space $\H$.
Basic statistical principles require that operations on the quantum system are described by unital completely positive (cp) maps
\begin{equation}
    T:\M\to \M.
\end{equation}
Further requiring that normal states are mapped to normal states, i.e., $\omega\circ T$ is normal for all normal states $\omega$ on $\M$, forces $T$ to be continuous with respect to the $\sigma(\M,\M_*)$ topology.
Such maps are called \emph{normal} unital cp maps.
The question arises whether all normal unital cp maps $T:\M\to\M$ correspond to operations on the quantum system described by $\M$.
This depends on the \emph{(local) implementability} of the operation, which requires that the operation can be extended to $\B(\H)$ in such a way that the extension is trivial on the commutant.
We will see that local implementability is equivalent to $T$ being an \emph{inner} cp map: 

\begin{definition}\label{def:inner}
 Let $\M$ be a von Neumann algebra. 
    A cp map $T:\M\to\M$ is called \emph{inner} if there exist operators $\{k_\alpha\}\subset\M$ such that
    \begin{equation}\label{eq:kraus}
        T(a) = \sum_{\alpha} k_\alpha^*ak_\alpha,\qquad a\in\M.
    \end{equation}
    The \emph{Kraus rank} $r(T)$ of an inner cp map $T$ is the minimal number of operators $k_\alpha$ that are necessary for \eqref{eq:kraus}.
    The set $\{k_\alpha\}\subset\M$ is called a collection of \emph{Kraus operators} if
    \begin{equation}
        \sum_{\alpha}k_\alpha^* k_\alpha=1.
    \end{equation}    
\end{definition}

Note that inner cp maps are automatically normal. They are unital if and only if the operators $k_\alpha$ in \eqref{eq:kraus} are a collection of Kraus operators.
The implementability of a normal ucp map $T:\M\to\M$ requires that this map can be extended to the environment in such a way that it acts trivially on the commutant (which we can think of as the complement of $\M$).

\begin{proposition}\label{prop:inner}
    Let $\M$ be a von Neumann algebra on $\H$.
    A normal ucp map $T:\M\to\M$ is inner if and only if it has a normal ucp extension $\tilde T:\B(\H)\to\B(\H)$ with $\tilde T \restriction \M' =\id_{\M'}$.
\end{proposition}

This follows from the following more general statement, which contains a description of the Stinespring dilation of inner maps and a generalization of \cref{prop:inner} to non-unital normal cp maps:

\begin{lemma}\label{lem:inner}
    Let $\M$ be a von Neumann algebra and let $T:\M\to\M$ be a normal cp map. The each of the following is equivalent to $T$ being inner:
    \begin{enumerate}[(i)]
        \item\label{it:inner2}
            There exists a faithful representation $\pi:\M\to\B(\H)$ and a normal cp map $\tilde T:\B(\H)\to\B(\H)$ satisfying
            \begin{equation}\label{eq:inner2}
                \tilde T\circ\pi= \pi\circ T \qandq \tilde T|_{\pi(\M)'} = T(1)\cdot \id_{\pi(\M)'}.
            \end{equation}
        \item\label{it:inner3}
            The statement in \cref{it:inner2} holds for all representations of $\M$.
        \item\label{it:inner4}
            Let $\M\subset\B(\H)$ be the standard representation. The minimal Stinespring dilation of $T:\M\to\B(\H)$ is of the form
            \begin{equation}\label{eq:inner3}
                T = v^*(\placeholder\ox1)v, \qquad v\in M_{r(T),1}(\M),
            \end{equation}
            where  $r(T)\in\NN\cup\{\oo\}$. The operator $v$ is an isometry if and only if $T$ is unital.
    \end{enumerate}
\end{lemma}

The number $r(T)\in\NN\cup\{\oo\}$ in item \ref{it:inner4} is precisely the Kraus rank of $T$.
We will refer to the dilation $T$ in \ref{it:inner4} as the \emph{minimal dilation} of $T:\M\to\M$.
The equation $v^*v=T(1)$ implies that $v$ is an isometry if and only if $T$ is unital, in which case we call $v$ the Stinespring isometry of $T$.

\begin{proof}
    The implications \ref{it:inner4} $\Rightarrow$ "$T$ is inner" $\Rightarrow$ \ref{it:inner3} $\Rightarrow$ \ref{it:inner2} are clear.

    \ref{it:inner2} $\Rightarrow$ "$T$ is inner":
    Let $\M\subset\B(\H)$ be a faithful representation such that $\tilde T = \sum_\alpha k_\alpha^*(\placeholder)k_\alpha$ with $k_\alpha\in\B(\H)$ satisfies \eqref{eq:inner2}.
    Let $y\in\M'$, then
    \begin{align*}
        0\le \sum_\alpha [k_\alpha,y]^* [k_\alpha,y] 
        &= \sum_\alpha ( y^*k_\alpha^*k_\alpha y+ k_\alpha^*y^*yk_\alpha- k_\alpha^* y^*k_\alpha y - y^*k_\alpha^*yk_\alpha ) \\
        &= y^* \tilde T(1) y + \tilde T(y^*y) - \tilde T(y^*)y-y^*\tilde T(y)=0.
    \end{align*}
    Therefore $[k_\alpha,y]^* [k_\alpha,y]$ must be zero for each $\alpha$ and all $y\in\M'$ and, hence, $k_\alpha\in \M''=\M$.

    "$T$ is inner" $\Rightarrow$ \ref{it:inner4}:
    Recall that the minimal Stinespring dilation can be obtained from an arbitrary dilation $(\K,\sigma,w)$ by reduction to the subspace $\hat\H=[\sigma(\M)w\H]$. 
    By assumption, we can find a $w\in M_{r,1}(\M)$ such that $T = w^*(\placeholder\ox 1)w$ for some $r\in\NN\cup\{\oo\}$.
    Let $\Omega \in\H$ be a cyclic separating vector for $\M$ and note that $wx' =(x'\ox1)w$, $x'\in\M'$.
    Thus:
    \begin{align*}
        [(\M\ox1)w\H] = [(\M\ox1)w\M'\Omega] = [(\M\M'\ox1)w\Omega]
        = [(\B(\H)\ox1)w\Omega]
    \end{align*}
    Schmidt decomposing $w\Omega$ reveals that the projection onto this subspace is $1\ox Q$ for some projection projection $Q$ on $\CC^{r}$.
    Therefore, the minimal Stinespring dilation is given by $v=(1\ox Q)w$, $\pi = \id \ox Q$ and $\hat\H=\H\ox Q\CC^r$. 
    Set $r(T)=\dim Q\CC^r$. Picking a basis of $Q\CC^r$, we can regard $v$ as an element of $M_{r(T),1}(\M)$, which shows the claim.
\end{proof}

A normal ucp map $T:\M\to \M$ has a normal ucp (left and right) inverse $T^{-1}:\M\to\M$ if and only if it is an automorphism (automorphism are automatically normal \cite[Cor.~III.3.10]{takesaki1}).
An automorphism $\alpha\in \Aut\M$ is called inner if it is implemented by a unitary in $u\in\M$ in the sense that $\alpha(a)=u^*au$, $a\in\M$.

\begin{lemma}[{\cite[Cor.~4.5]{mingo1989}}]\label{lem:inner auto}
    An automorphism is an inner cp map if and only if it is an inner automorphism.
\end{lemma}

Next, we explain how quantum measurements on a system with observable algebra $\M$ are defined. 
For simplicity and because of our application to LOCC, we only consider discrete outcome spaces.
An instrument implementing a measurement with outcome space $X$ is described by a \emph{quantum instrument}, which is a normal ucp map
\begin{equation}
    T: \M\ox\ell^\oo(X) \to \M.
\end{equation}
A quantum instrument is completely determined by the normal cp contractions $T_x = T(\placeholder\ox \delta_x)$. Conversely every collection $\{T_x\}$ of normal cp maps $T_x:\M\to\M$ such that $\bar T =\sum_x T_x$ is unital, defines an instrument $T$ with $T_x=T(\placeholder\ox\delta_x)$.
The maps $T_x$ describe the post-measurement state:
If the outcome $x$ was measured and the system was in the state $\omega$, then the post-measurement state is
\begin{equation}
    \omega_x = \frac1{p_x} \,\omega\circ T_x,\qquad p_x = \omega(T_x(1)),
\end{equation}
which is well-defined because $p_x$ is the probability of measuring the outcome $x$ (this probability is non-zero since, by assumption, the outcome $x$ has been measured).
We must again consider the implementability of these operations, and by \cref{prop:inner}, implementability is essentially equivalent to innerness.

\begin{lemma}\label{lem:coarse-graining}
    Let $\{T_x\}_{x\in X}$ be a quantum measurement on a von Neumann algebra $\M$.
    Then $\bar T$ is inner if and only if each $T_x$ is inner if and only if there exists a set $Y$, a partition $Y=\cup_{x\in X} Y_x$ and Kraus operators $k_{y}\in \M$, $y\in Y$, such that
    \begin{equation}
        T_x = \sum_{y\in Y_x} k_y^*(\placeholder) k_y.
    \end{equation}
    Thus, $\{T_x\}_{x\in X}$ is a coarse-graining of an instrument $\{\tilde T_y\}_{y\in Y}$ where every $\tilde T_x$ has Kraus rank $1$.
\end{lemma}

\begin{remark}[Inner operations are not closed w.r.t.\ pointwise convergence]
    The set of inner operations is not closed in the point-ultraweak topology, i.e., in the topology induced by the functionals $T\mapsto \omega(Tx)$, $\omega\in \M_*$, $x\in\M$.
    As an example, consider the approximately finite-dimensional type $\II_1$ factor $\M = \bigotimes_{n\in\NN} (M_2(\CC),\frac12\tr)$ acting on $\H=\bigotimes_{n\in\NN} (\CC^2\ox\CC^2,\Phi^+)$ where $\Phi^+= 2^{-1/2}(\ket0\ket0+\ket1\ket1)$, which may be constructed by taking the weak closure of $\A=\bigotimes_{n\in\NN}M_2(\CC)$ in the GNS representation of the unique tracial state.
    Let $u_0\in M_2(\CC)$ be a unitary and consider the automorphism $\alpha_0(a) = \ox_{n\in\NN} u_0(\placeholder)u_0^*$ on $\A$.
    Since $\alpha_0$ leaves the tracial state invariant, it is implemented by a unitary $u$ on the GNS space, which yields an automorphism $\alpha$ on $\M$ (namely the weak extension of $\alpha_0$).
    This automorphism is not inner \cite[Thm.~XIV.1.13]{takesaki3}. However, it lies in the point-ultraweak closure of inner operations since it can be approximated by the inner automorphisms that are implemented by the unitaries $u_0^{\ox k} \ox 1^{\ox\oo}\in \A$.
    While these approximating unitaries converge in the weak operator topology of $\H$, their limit is zero.
\end{remark}

\subsection{Local operations in multipartite systems}\label{sec:general LOCC}

\begin{definition}\label{def:locality-pres}
    A \emph{locality-preserving} transformation of a multipartite system $(\M_1,\ldots,\M_N)$ on a Hilbert space $\H$ is a normal ucp map $T:\B(\H)\to\B(\H)$ such that $T(\M_x)\subset \M_x$ for all $x\in[N]$ and
    \begin{equation}
        T( a_1\cdots a_N) = T(a_1)\cdots T(a_N), \qquad a_x\in\M_x, \ x\in [N].
    \end{equation}
    A unitary $u$ on $\H$ is locality preserving if the automorphism $T=u^*(\placeholder)u$ is a locality preserving operation.
\end{definition}

\begin{lemma}\label{lem:auts-lp}
    Let $\M$ be a von Neumann algebra with standard form $(\H,J,\P)$.
    Consider the standard bipartite system $(\M,\M')$.
    Then, for every automorphism $\alpha$ on $\M$, the unitary $u_\alpha$ from \cref{eq:unitary implementation} is a locality-preserving unitary $u_\alpha$ on $\H$ such that $u_\alpha J = J u_\alpha$. 
\end{lemma}
\begin{proof}
    The existence of $u_\alpha$ was discussed in \cref{sec:standard}. Clearly, $T(ab)=T(a)T(b)$ for $T=u_\alpha^*(\placeholder)u_\alpha$ and $a\in\M,b\in\M'$. Since $u_\alpha^* \M u_\alpha = \M$ we have 
    \begin{align}
    T(\M') = u_\alpha^* \M' u_\alpha = u_\alpha^* J \M J u_\alpha = J u_\alpha^* \M u_\alpha J = J \M J = \M'.
    \end{align}
\end{proof}

Locality-preserving operations cannot be regarded as local operations as they might not be implementable locally.

\begin{definition}\label{def:local operation}
    Let $(\M_1,\ldots,\M_N)$ be an $N$-partite system on $\H$.
    A \emph{local operation of the party $x\in[N]$} is a normal ucp map $T_x:\B(\H)\to \B(\H)$ such that
    \begin{equation}\label{eq:local_op}
        T_A(\M_x)\subset \M_x \qandq T_x\restriction{\M_x'} = \id_{\M_x'}.
    \end{equation}
    A local operation on the multipartite system $(\M_x)_{x\in[N]}$ is a product $T = T_1\circ \cdots \circ T_N$ of a local operations $T_x$ performed by each of the parties.
\end{definition}

It is evident that all local operations are locality-preserving, but the converse is false:

\begin{example}
    Consider a free field theory $\O \mapsto \A(\O)$ on Minkowski space $\RR^{1,3}$, let $W$ and $W'$ be complementary wedges, and let
    $\M=\A(W)$ and $\M'=\A(W')$ be the corresponding observable algebras (see \cref{sec:examples}).
    Let $g \in \text{SO}(1,3)$ be Lorentz boost  such that $gW=W$, $gW'=W'$,
    and let $u = u_g$ be the implementing unitary.
    Then $T = u^*(\placeholder)u$ is locality-preserving but not a product of local operations.
    This follows from the Bisognano-Wichmann theorem and the fact that $\M$ and $\M'$ are type $\III_1$ factors.%
    \footnote{By \cref{lem:inner auto} and \cref{prop:local_op}, the map $T=u^*(\placeholder)u$ being local would imply that $a\mapsto u^*au$ is an inner automorphism on $\M$.
    By the Bisognano-Wichmann theorem, the Lorentz boost is implemented as the modular flow $u=\Delta_\Omega^{it}$ relative to the vacuum $\Omega\in\H$ (at an appropriate time corresponding to the rapidity of the Lorentz boost) \cite{bisognano_duality_1975}.
    However, since $\M$ is a type $\III$ factor, the modular flow $\sigma_t(a) = \Delta_\Omega^{it}a\Delta_\Omega^{-it}$ is an outer automorphism for some $t\in\RR$ \cite[Thm.~III.4.6.6]{blackadar_operator_2006}.
    Thus, the Lorentz boost induces a non-inner automorphism on $\M$, which means that it cannot be a local operation for $(\M,\M')$. 
    }
\end{example}

\begin{example}\label{exa:locality-pres unitary}
    We consider the standard bipartite system corresponding to Alice and Bob sharing infinitely many bell pairs (see \cref{sec:examples}).
    The Hilbert space of this bipartite system is given as $\H=\bigotimes_{n\in\NN} (\CC^2\ox\CC^2; \Phi^+)$, where $\Phi^+=2^{-1/2}(\ket0\ket0+\ket1\ket1)$.
    Note that, for every unitary $u_0$ on $\CC^2$, we have $(u_0\ox \bar{u_0})\Phi^+=\Phi^+$.
    Therefore, the infinite tensor product $u = \otimes_{n\in\NN} (u_0\ox\bar u_0)$ yields a well-defined unitary on $\H$ (see \cite[Thm.~XIV.1.13]{takesaki3}).
    Clearly, this unitary leaves the factors $\M_A$ and $\M_B$ corresponding to Alice and Bob invariant and, hence, is locality preserving.
    However, the implemented automorphism $\alpha = \otimes_{n\in\NN} \Ad_{u_0}$ on $\M_A$ is outer unless $u_0$ is a scalar \cite[Thm.~XIV.1.13]{takesaki3}.
\end{example}

In analogy with \cref{def:locality-pres}, we could call a unitary $u$ on $\H$ local if the automorphism $T=u^*(\placeholder)u$ is a local operation.
The next result shows that such unitaries are simply given by products $u=\prod_x u_x$ of unitaries $u_x\in\M_x$.
Indeed, as a direct consequence of \cref{lem:inner}, we get that local operations are necessarily inner:

\begin{proposition}\label{prop:local_op}
    Let $(\M_x)_{x\in[N]}$ be a multipartite system on $\H$. 
    A map $T:\B(\H)\to\B(\H)$ is a local operation of the party $x$ if and only if it is of the form
    \begin{equation}\label{eq:local_op1}
        T = \sum_\alpha k_\alpha^*(\placeholder)k_\alpha, \qquad\sum_\alpha k_\alpha^*k_\alpha=1,\qquad \{k_\alpha\}\subset\M_x.
    \end{equation}
    Thus, $T$ is a local operation of the multipartite system if and only if it is of the form
    \begin{equation}\label{eq:local_op2}
        T= \sum_{\alpha_1,\ldots \alpha_N} (k_{\alpha_1}\cdots k_{\alpha_N})^*(\placeholder)(k_{\alpha_1}\cdots k_{\alpha_N}), \qquad \sum_{\alpha_x} k_{\alpha_x}^*k_{\alpha_x} =1,\qquad \{k_{\alpha_x}\}\subset \M_x,\,x\in[N].
    \end{equation}
\end{proposition}
\begin{proof}
    It is shown in \cref{lem:inner} that a normal unital completely positive map $T_x:\B(\H)\to\B(\H)$ us a local operation $T_x$ of $\M_x$ if and only if it is of the form in \eqref{eq:local_op1} (cp.\ \cref{def:local operation}).
    Since general local operations $T$ are, by definition, products of a local operation each of the single parties, the second claim follows.
\end{proof}

As a direct consequence of \cref{lem:inner auto}, we get:

\begin{corollary}
    Let $(\M_x)_{x\in[N]}$ be an $N$-partite system on $\H$ and let $T$ be a local operation.
    Then $T$ is invertible with local inverse, i.e., there exists a local operation $T^{-1}$ with $T\circ T^{-1}=T^{-1}\circ T = \id$, if and only if there exist unitaries $u_x\in\M_x$, $x\in[N]$, such that
    \begin{equation}
        T = u^*(\placeholder)u, \qquad u= \prod_{x\in[N]} u_x.
    \end{equation}
\end{corollary}

\subsubsection*{Locality preserving modulo local operations}

We wish to understand the discrepancy between locality-preserving and local operations.
In standard bipartite systems, we will see that this boils down to the difference between general and inner operations on a single von Neumann algebra.

If $(\M_x)_{x\in[N]}$ is an $N$-partite system on a Hilbert space $\H$, we define the unitary groups 
\begin{align}
    \U_{lp}(\M_1,\ldots,\M_N) &:= \{ u \in \U(\H)\quad:\quad \text{$u$ is locality-preserving}\}
\intertext{and}
    \U_{l}(\M_1,\ldots,\M_N) &:= \{ u \in \U(\H)\quad:\quad u=\prod u_x,\ u_x\in\U(\M_x) \}
\end{align}
It is evident that $\U_l(\M_1,\ldots,\M_N)$ is a normal subgroup of $\U_{lp}(\M_1,\ldots,\M_N)$, which suggests to study the induced quotient group.
If $\M$ is a von Neumann algebra, we denote by $\Inn(\M)$ the group of inner automorphisms, which is a normal subgroup of $\Aut(\M)$, and we denote the quotient group $\Aut(\M)/\Inn(\M)$ as $\Out(\M)$.

\begin{proposition}\label{lem:OutM}
    Let $(\M,\M')$ be a standard bipartite system.
    Then 
    \begin{equation}
        \U_{lp}(\M,\M')/\,\U_l(\M,\M') \cong \Out(\M).
    \end{equation}
    The isomorphism is induced by the homomorphism $\rho:\U_{lp}(\M_A,\M_B) \ni u \mapsto  u^*(\placeholder)u \in \Aut\M$ which maps $\ker\rho = \U_l(\M,\M')$ onto $\Inn(\M)$.
\end{proposition}

By a result of Falguières and Vaes \cite{falguieres_every_2008}, for every compact group $G$ there exists a (type $\II_1$) factor such that $\Out(\M) \cong G$.
Therefore, every compact group $G$ arises as locality-preserving mod local operations of a standard bipartite system.

\begin{proof}
We introduce the shorthands $\U_l$ and $\U_{lp}$ for $\U_l(\M,\M')$ and $\U_{lp}(\M,\M')$.
It is clear that $\rho(\U_l) = \Inn(\M)$.
Thus, $\rho$ induces a homomorphism $\dot \rho:\U_{lp}/\U_l \to \Out(\M)$.
Let $J$ be a conjugation and $\P\subset\H$ a self-dual positive cone such that $(\H,J,\P)$ is a standard form for $\M$.
Since every automorphism $\alpha\in\Aut(\M)$ is implemented by a locality-preserving unitary $u_\alpha$ (see \cref{lem:auts-lp}), the map $\rho$, and hence also $\dot\rho$, is surjective.
Next we show that $\dot\rho$ is injective by showing that $\ker\rho = \U_l$.
If $u\in\ker\rho$, then $u^*(\placeholder)u$ is inner on $\M$. We pick a unitary $v\in\M$ such that $u^*au = v^*av$ for all $a\in\M$ and set $v' = u^*v$.
Then, since $v'a = vu^*auu^* = vv^*avu^* =a v'$ for all $a\in\M$, we have $v'\in\M'$.
Therefore, $u = vv'$ is in $\U_l$.
\end{proof}

\begin{lemma}\label{lem:inner iff typeI}
    Let $\M$ be an approximately finite-dimensional factor. Then $\M$ is type $\I$ if and only if every automorphism on $\M$ is inner.
\end{lemma}
\begin{proof}
    For type $\III$ factors with separable predual, we can always find times $t\in\RR$ such that the modular flow $\sigma_t^\psi$ with respect to some a normal semifinite weight $\psi$ on $\M$ is outer.
    In fact, the Connes invariant $T(\M)$ (which consists of those times $t$ for which the modular flow is inner) of a type $\III$ factor with separable predual always has Lebesgue measure zero \cite[Prop.~27.2]{stratila_modular_2020}.
    Approximately finite-dimensional type $\II$ factors can be as written infinite tensor products $\M=\bigotimes_{n\in\NN}\,(M_2(\CC);\omega_n)$ where $\omega_n\in\{ \frac12\tr_2, \bra1\placeholder\ket1\}$ for all $n$ (one gets $\II_1$ if and only if $\omega_n=\frac12\tr_2$ for all but finitely many $n$ and $\II_\oo$ if both kinds appear infinitely often). Let $u_0=\mathrm{diag}(1,-1)\in M_2(\CC)$ such that $u_0\omega_n u_0^*=\omega_n$ for all $n$. However, since $\tr(u_0)=0$, the automorphism $\alpha = \otimes_{n\in\NN} u_0^*(\placeholder)u_0$ on $\M$ is outer \cite[Thm.~XIV.1.13]{takesaki3}.
\end{proof}

The assumption of approximate finite dimensionality is only essential in the type $\II$ case.
Indeed, any type $\III$ factor with separable predual admits outer automorphisms (see the proof of \cref{lem:inner iff typeI}).

\begin{proposition}\label{prop:l=lp iff typeI}
    Let $(\M_A,\M_B)$ be a factorial standard bipartite system on $\H$.
    If $\M_A$ and $\M_B$ are type $\I$, then every locality-preserving operation is local.
    The converse holds if $\M_A$ and $\M_B$ are approximately finite-dimensional.
\end{proposition}
\begin{proof}
    This follows from \cref{lem:OutM,lem:inner iff typeI}.
\end{proof}

As mentioned above, if it is known that $\M_A$ and $\M_B$ are not of type $\II$, then the assumption of approximate finite dimensionality may be dropped.

\subsection{LOCC protocols }
\label{Sec:locc-protocols}
A general LOCC protocol on an $N$-partite system $(\M_1,\ldots,\M_N)$ consists of a finite number of rounds of classical communication between agents interspersed by local quantum instruments applied by the $N$ agents (see \cite{chitambar_everything_2014} for details in the finite-dimensional setting). 
In each communication round, the outcomes of the instruments applied by each agent are communicated to the other agents, and the subsequent instruments applied by the agents will, in general, depend on the received messages, i.e., the outcomes of prior instruments. 
The overall outcome of the LOCC protocol applied to the input state $\psi$ is given by a classical outcome $x$ in some outcome space $X$ incorporating all outcomes of the steps of the protocol together with a quantum state 
\begin{align}
    \psi_x = \frac{1}{p_x} \psi\circ T_x,\quad p_x = \psi(T_x(1)),
\end{align}
where the instrument $\{T_x\}_{x\in X}$ takes the form
\begin{align}\label{eq:overall-LOCC-instrument}
    T_x = \prod_{j\in [N]}T_x^{(j)}
\end{align}
and each $\{T_x^{(j)}\}_{x\in X}$ determines an instrument on the $j$th subsystem. 
We call $\{T_x\}_{x\in X}$ the \emph{overall instrument} of a given LOCC protocol with local instruments $\{T\up j_x\}_{x\in X}$.  
Not every instrument of the form \cref{eq:overall-LOCC-instrument} appears as the overall instrument of a LOCC protocol.

\begin{definition}[LOCC transitions]
Let $(\M_1,\ldots,\M_N)$ be a multipartite system and let $\psi,\phi$ be normal states on $\bigvee_x\M_x$.
We say that a state $\psi$ \emph{can be transformed into $\phi$ with probability $p\in [0,1]$ via LOCC} if there exists an LOCC protocol with overall instrument $\{T_x\}_{x\in X}$ such that
\begin{align}
    \sum_{x:\ \psi_x = \phi} \psi(T_x(1)) \geq p. 
\end{align}
We write $\psi \locc \phi$ (resp.\ $\psi\slocc\phi$) if $\psi$ can be transformed into $\phi$ with $p=1$ (resp.\ $p>0$).
Finally, we write $\psi \barlocc \phi$ (resp.\ $\psi \barslocc \phi$) if for every $\eps>0$ there is a state $\phi'$ with $\norm{\phi-\phi'}<\eps$ and $\psi\xrightarrow{\LOCC} \phi$ (resp.\ $\psi\xrightarrow{\LOCC} \phi$).
\end{definition}

The following implications are obvious from the definition:
\begin{equation}
\begin{tikzcd}
    \quad\, \psi\locc\phi \quad\, \arrow[Rightarrow]{r} \arrow[Rightarrow]{d} &  \quad\, \psi\barlocc\phi \quad\, \arrow[Rightarrow]{d} \\
     \quad\psi\slocc\phi \quad \arrow[Rightarrow]{r} &  \quad\psi\barslocc\phi  \quad
    \end{tikzcd}
\end{equation}

\begin{remark}[LOCC in ground state sectors]
    Consider a ground state sector of a quantum many-body system of particles localized on the sites of a lattice $\Gamma$ (see \cref{sec:examples}).
    For every bipartition $\Gamma = A\cup B$, we can consider the bipartite system $(\M_A,\M_B)$.
    Since $\M_A$ is the $\sigma$-strong closure of the operators with finite support in $A$, the Kraus operators appearing in the local instruments of Alice realizing an LOCC state transition $\psi\xrightarrow{\LOCC}\phi$ may be approximated to arbitrary accuracy by Kraus operators with finite support within $A$ (and similarly for Bob).
    Thus, the state transition may be realized (up to arbitrarily small errors) via an LOCC protocol whose instruments all have finite support on the lattice. Therefore $\overline\LOCC$ transitions can equivalently be defined with finitely localized operations.
    We may view the von Neumann algebras $\M_A,\M_B$ as idealized objects that do not have direct physical significance but allow us to prove statements about the physically relevant case of finitely localized LOCC protocols up to any finite error.
\end{remark}

\begin{remark}
    LOCC transitions are also considered in the case that the observable algebras are general von Neumann algebras (including non-factors) in \cite{crann_state_2020}.
    In this case, the local operations are still described by inner operations.
    As a consequence, they do not allow the local parties to manipulate their respective local classical degrees of freedom (described by the center of the algebra). 
    It would be interesting to better understand LOCC transitions that allow for the manipulation of classical degrees freedom. 
\end{remark}

\section{Pure state LOCC}

In this section, we specialize in LOCC transformations between pure states of bipartite systems of factors, i.e., purely quantum systems, in Haag duality. 
We first discuss stochastic LOCC (SLOCC) transformations and then generalize Nielsen's theorem to bipartite systems of factors of arbitrary type. 
We then use Nielsen's theorem to discuss pure state LOCC transformations on different types of bipartite systems.
Since we consider pure state transformations on irreducible systems, in the following, we identify pure states with their vector representatives on $\H$. 
In particular, we use the notation $\Psi\xrightarrow{\LOCC}\Phi$ (and similarly for SLOCC)

As a direct application of \cref{lem:coarse-graining} we find:
\begin{lemma}\label{lem:pure state locc and kraus ops}
    Consider a bipartite system of factors $(\M_A,\M_B)$ on $\H$ and let $\Psi,\Phi\in\H$ be unit vectors such that $\Psi\xrightarrow{\LOCC}\Phi$.
    Then there exists a discrete outcome space $Y$, an LOCC protocol with overall instrument $\{T_y\}_{y\in Y}$ and local rank-1 instruments $T_{A,y} = k_{A,y}^*(\placeholder) k_{A,y}$ and $T_{B,y} = k_{B,y}^*(\placeholder) k_{B,y}$, such that 
    \begin{align}
        k_{A,y}k_{B,y}\Psi \ \propto\ \Phi \qquad \text{for all $y\in Y$}.\ \,
    \end{align}
    Similarly, if $\Psi \xrightarrow{\SLOCC} \Phi$ there exists a suitable LOCC protocol  such that
    \begin{align}
        \ \, k_{A,y}k_{B,y} \Psi \ \propto\ \Phi\qquad \text{for some $y\in Y$},
    \end{align}
    where $k_{A,y}\in \M_A$, $k_{B,y}\in \M_B$ are Kraus operators of the associated overall instrument.
\end{lemma}

\subsection{Majorization on von Neumann algebras}
\label{sec:majorization-mt}

The most important implication of Nielsen's theorem is that it transfers questions about pure state entanglement to questions in majorization theory, i.e., classical probability theory. For example, it allows us to construct a pure state entanglement monotone from any convex function on $\RR$, and, as we will see, the same holds true in systems with infinitely many degrees of freedom.

Most discussions of majorization theory restrict to finite-dimensional probability vectors or discrete probability distributions, which is sufficient to discuss pure state LOCC for finite-dimensional quantum systems. 
However, it is insufficient for quantum systems with infinitely many degrees of freedom, to which Nielsen's theorem also applies, as we will show. 
In the next sections, we will make use of majorization theory on general $\sigma$-finite measure spaces as well as on semifinite von Neumann algebras, much of which was developed in \cite{alberti_stochasticity_1982,kamei1983majorization,kamei1984double,petz1985scale,hiai_majorization_1987,hiai_majorizations_1987,fack1986generalized,hiai1989distance}. 
Since this material is treated differently by various authors, and since for the application to LOCC, we are only interested in majorization theory for positive functionals; we provide a self-contained and comprehensive, independent treatment of majorization theory on $\sigma$-finite measure spaces and semifinite von Neumann algebras in \cref{app:majorization}.
Our treatment is guided by reducing to the classical (commutative) case as quickly as possible. 
We believe this appendix to be of independent interest and now provide a brief summary  of those results that we will use in the application to LOCC in the following sections. 

To set the stage, we briefly recall the standard results of finite-dimensional majorization theory. For any probability distribution $p$ on $[d]$ we define the \emph{Lorenz curve} $L_p:[0,d]\to [0,1]$ as
\begin{align}
    L_p(t) := \sup\bigg\{ \sum_{s=1}^d p(s) a(s)\quad :\quad 0\leq a(s) \leq 1,\ \sum_{s=1}^d a(s) \leq t\bigg\}.
\end{align}
If $\rho\in M_d(\CC)$ is a $d$-dimensional density matrix we define $L_\rho(t) := L_{\lambda_\rho}(t)$, where $\lambda_\rho(t)$, $t\in[d]$, are the eigenvalues of $\rho$ in non-increasing order (repeated according to their multiplicity).
It follows that 
\begin{align}
   L_\rho(t) = \sup\Big\{\Tr(\rho a)\quad :\quad a\in M_d(\CC),\  0\leq a\leq 1, \ \Tr a\leq t\Big\}. 
\end{align}
Recall that a cp map $T:M_d(\CC)\to M_d(\CC)$ is called \emph{doubly stochastic} if $T(1)=1$ and $\Tr T(a) = \Tr a$ for all $a\in M_d(\CC)$.

\begin{lemma}\label{lem:majorization-finite-dimensional}
For any two density matrices $\rho,\sigma\in M_d(\CC)$ the following are equivalent:
\begin{enumerate}
    \item There exists a probability distribution $\{p_x\}$ and unitaries $u_x\in M_d(\CC)$ such that $\sigma = \sum_x p_x u_x \rho u_x^*$.
    \item There exists a doubly stochastic cp map $T:M_d(\CC)\to M_d(\CC)$ such that $T(\rho)= \sigma$. 
    \item For any convex function $f:\RR^+\rightarrow \RR^+$ we have
    \begin{align}
        \Tr f(\rho) \geq \Tr f(\sigma). 
    \end{align}
    
    \item For any $t\in [0,d]$ we have
    \begin{align}
        L_\rho(t) =L_{\lambda_\rho}(t) \geq L_{\lambda_\sigma}(t) = L_\sigma(t).
    \end{align}
    \end{enumerate}
    If any of these conditions are true, we say that $\rho$ majorizes $\sigma$, denoted by $\rho\succeq \sigma$. 
\end{lemma}
We use the first item to generalize the definition of majorization to states on arbitrary factors:
\begin{definition}
    Let $\M$ be a factor (of arbitrary type) and $\psi,\phi\in\M_*$ be states on $\M$. The state $\phi$ majorizes $\psi$, written $\phi\succeq \psi$, if
    \begin{align}
\psi \in \overline{\mathrm{conv}}\{u\phi u^*\ :\ u\in \U(\M)\},
    \end{align}
    where the closure is taken with respect to the norm topology on $\M_*$.
\end{definition}

Note that $\rho \mapsto \sum_x p_x u_x \rho u_x^*$ is a doubly stochastic cp map and corresponds precisely to the transformation that occurs on the level of the local marginals in a pure state LOCC transformation for finite-dimensional quantum systems.
In bipartite systems where the local factors are not finite, we will see that LOCC transformations are directly related to transformations of the type $\phi \mapsto \psi = \sum_x p_x v_x\phi v_x^*$, where $v_x$ are partial isometries instead of unitaries. This mapping is not necessarily doubly stochastic.
It is, however, \emph{doubly substochastic} and takes the relevant input state (i.e., $\phi$) to a properly normalized output state (i.e., $\psi$).

This difference is directly reflected in majorization theory and is the reason why we have to take a closure in our definition of majorization. We will see that majorization trivializes if $\M$ has type $\III$ (see \cref{sec:trivialization-III}). In the following we therefore restrict to semifinite von Neumann algebras. When dealing with majorization theory on general semifinite von Neumann algebras $\M$, there are three essential changes we have to make: 
\begin{enumerate}
\item Instead of mixtures of unitaries, we have to consider mixtures of partial isometries. 

\item Instead of considering doubly stochastic channels, we consider normal, completely positive maps $T:\M\to \M$, which are \emph{doubly substochastic}: $T(1_\M) \leq 1_\M$ and $\Tr_\M \circ T \leq \Tr_\M$. 
\end{enumerate}

We describe the definitions required in the general case and state the analog of \cref{lem:majorization-finite-dimensional}. 
In the following, $\M$ is a semifinite von Neumann algebra with trace $\Tr$. We denote by $L^p(\M):= L^p(\M,\Tr)$ the $L^p$ spaces of $\M$ with respect to its trace (see, for example, \cite[Sec. IX.2]{takesaki2} or \cite[Sec.~]{hiai_lecture_2021}). It is well-known that $L^1(\M)$ is isometrically isomorphic to $\M_*$. Hence,
every normal state $\psi$ on $\M$ can be identified with its density $\rho_\psi \in L^1(\M)^+$ defined by $\psi(x) = \Tr \rho_\psi x$. Therefore, we discuss majorization theory at the level of densities.

The definition of the Lorenz curve immediately generalizes to elements $\rho\in L^1(\M)^+$ by
\begin{align}\label{eq:lorenz curve main}
    L_\rho(t) := \sup\big\{\Tr(\rho a)\quad:\quad a\in \M,\ 0\leq a\leq 1,\ \Tr a\leq t\big\}. 
\end{align}
The non-increasingly eigenvalues from before are generalized using the notion of \emph{spectral scale} introduced by Petz \cite{petz1985scale}. 
The \emph{distribution function} of a density $\rho\in L^1(\M)^+$ is the right-continuous non-increasing function
\begin{align}
    D_\rho(t) := \Tr(\chi_{[t,\infty)}(\rho) ), \quad t > 0.
\end{align}
Here, $\chi_A$ denotes the indicator function of a set $A$. The spectral scale of $\rho$ is then defined by
\begin{align}
    \lambda_\rho(t) := \inf\{s >0 : D_\rho(s)\leq t\},\quad t > 0. 
\end{align}
If $f$ is a function on $\RR^+$ we have
\begin{align}
    \Tr f(\rho) = \int_{\RR^+} (f\circ \lambda_\rho)(t) dt.
\end{align}
In particular, the spectral scale of a normalized density $\rho$ is a probability density in $L^1(\RR^+)$. 
The spectral scale $\lambda_\rho$ of a density $\rho\in L^1(\M)^+$ coincides with the \emph{generalized s-numbers} of Fack and Kosaki \cite{fack1986generalized} (see also \cite[Sec.~4.2]{hiai_lecture_2021}).

\begin{definition}
    Let $\M$ and $\N$ be semifinite von Neumann algebras and let $\rho\in L^1(\M)^+$ and $\sigma \in L^1(\N)^+$. 
    We say that $\rho$ \emph{submajorizes} $\sigma$, denoted $\rho\succ_w \sigma$, if $L_\rho(t)\ge L_\sigma(t)$ for all $t\ge0$, and that $\rho$ \emph{majorizes} $\sigma$, denoted $\rho\succ \sigma$, if $\rho\succ_w \sigma$ and $\Tr \rho = \Tr \sigma$.
\end{definition}

This definition applies to general semifinite von Neumann algebras, not just factors. Furthermore, we allow to compare densities on different von Neumann algebras. 
In particular, we may compare $\rho$ with its spectral scale $\lambda_\rho$ and find that (see \cref{cor:marjorization is classical})
\begin{align}
\rho \succ \lambda_\rho \succ \rho,
\end{align}
where we view $\lambda_\rho$ as an element of $L^1(X,\mu)$, where $X=\Sp(\rho)\setminus\{0\}$ and $\mu$ is the measure induced by $\rho$ via $\mu(\Omega) = \Tr \chi_\Omega(\rho)$. 
This allows us to transfer results from the commutative case to the noncommutative case. We have:

\begin{theorem}[Submajorization, see \cref{thm:q submaj}]
    Let $\M$ and $\N$ be semifinite von Neumann algebras and let $\rho\in L^1(\M)^+$ and $\sigma\in L^1(\N)^+$. The following are equivalent:
    \begin{enumerate}[(a)]
        \item\label{it:q submaj1-mt} $\rho \succ_w \sigma$,
        \item\label{it:q submaj2-mt} $\lambda_\rho \succ_w \lambda_\sigma$, 
        \item\label{it:q submaj3-mt} $\Tr f(\rho) \ge \Tr f(\sigma)$ for all continuous convex functions $f:\RR^+\to\RR^+$ with $\phi(0)=0$,
        \item\label{it:q submaj4-mt} there exists a doubly substochastic cp map $T$ from $\M$ to $\N$ such that $T(\rho)=\sigma$.
    \end{enumerate}
    The operator $T$ can be chosen to satisfy $T(\supp \rho)\le \supp \sigma$ and $T^*(\supp \sigma) \le \supp \rho$.
\end{theorem}
Let us briefly discuss the distinction between doubly substochastic and doubly stochastic normal cp maps. 
One may be worried that doubly substochastic maps are not valid quantum channels. 
While this is true, it does not impede the physical relevance, since for every doubly substochastic, normal, completely positive map $T$ on $\M$ such that $\psi=\phi\circ T$ for two states $\phi,\psi$ on $\M$, we can 
find a unital (but not trace-preserving), normal, completely positive map $\hat T$ such that $\psi = \phi\circ \hat T$ by setting
\begin{align}
    \hat T = T + (1-T(1))\omega(\placeholder),
\end{align}
where $\omega$ is some normal state on $\M$. 

The relation between doubly substochastic maps and doubly stochastic maps is similar to the relation between (partial) isometries and unitaries. 
Consider a proper isometry $v$ on an infinite-dimensional Hilbert space $\H$ (proper meaning that $v\H$ is a proper subspace). Then, $v$ cannot be extended to  a unitary. 
The map $T_{v} = v (\placeholder) v^*$ is only doubly substochastic because $T_{v}(1) = vv^*\ne 1$ even though it is trace-preserving. 
It can be extended to a unital (but not trace-preserving) map $\hat T_{v}$ so that $\langle \Psi,T_v(\placeholder)\Psi\rangle=\langle \Psi,v(\placeholder) v^*\Psi\rangle = \langle \Psi,\hat T_{v}(\placeholder)\Psi\rangle$ for all $\Psi \in vv^* \H$. 
As a concrete example, consider the shift $v\ket{n} = \ket{n+1}$ on $\H=\ell^2(\NN)$. Then $v\H = \overline{\mathrm{span}}\{\ket{n}:\ n\geq 2\}$ and we may set
\begin{align}
    \hat T_v = v(\placeholder)v^* + \ketbra{1}{1}\, \omega(\placeholder)
\end{align}
for some state $\omega$, since $1-T(1) = \ketbra{1}{1}$. 

We now specialize to factors. Combining a result of Hiai \cite{hiai_majorization_1987} with our results, we find:
\begin{theorem}[Majorization in semifinite factors, see \cref{thm:factor majorization}]\label{thm:factor-majorization-mt}
    Let $\M$ be a semifinite factor and let $\rho,\sigma\in L^1(\M)^+$.
    The following are equivalent
    \begin{enumerate}[(a)]
        \item\label{it:factor majorization1-mt} $\rho\succ \sigma$, 
        \item\label{it:factor majorization2-mt} $\sigma \in \overline{\conv}\{u\rho u^*\ :\ u\in\U(\M)\}$,
        \item\label{it:factor majorization3-mt} $\tr \rho = \tr \sigma$ and $\tr f(\rho) \ge \tr f(\sigma)$ for all continuous convex functions $f:\RR^+\to\RR^+$, 
        \item\label{it:factor majorization4-mt} $\lambda_\rho \succ \lambda_\sigma$,
    \end{enumerate}
    If $\M$ is finite, i.e., $\tr 1<\oo$, these conditions are equivalent to the existence of a doubly stochastic map on $\M$ with $T(\rho)=\sigma$.
\end{theorem}
The closure in \cref{it:factor majorization2-mt} is taken in the norm topology on $L^1(\M)\cong \M_*$.

\subsection{Stochastic LOCC}\label{sec:slocc}
Before discussing LOCC transformations, we characterize SLOCC transformations of pure states.
\begin{theorem}\label{thm:slocc}
    Let $\M$ be a factor on $\H$ and let $\Psi,\Phi\in\H$ be unit vectors.
    Then SLOCC with respect to $(\M,\M')$ is characterized as
    \begin{equation}\label{eq:slocc}
        \Psi\slocc \Phi\quad \iff\quad \exists_{k\in\M,\,k'\in\M'}\ \ \text{s.t.}\ \ \Phi = kk'\Psi.
    \end{equation}
    Furthermore, the following are equivalent:
    \begin{enumerate}[(a)]
        \item\label{it:slocc1} $\Psi \barslocc \Phi$,
        \item\label{it:slocc2} $s_\phi \lesssim s_\psi$,
        \item\label{it:slocc3} $s_{\phi'}\lesssim s_{\psi'}$,
        \item\label{it:slocc4} There exist partial isometries $v\in\M$ and $v'\in\M'$ with $w^*w\Psi=\Psi$ and $\Phi \in [\M'w\Psi]\cap [\M w\Psi]$ where $w=vv'$,
        \item\label{it:slocc5} for each $\eps>0$ there exists a unitary $ u\in\M$ and $k'\in \M'$ such that $\norm{uk\Psi -\Phi} < \eps$.
    \end{enumerate}
\end{theorem}

We need the following Lemma, originally due to Murray-von Neumann (see \cite[Sec.~III.1]{murray_rings_1936}):

\begin{lemma}[{\cite[Cor.~2.7.10]{sakai_c-algebras_1998}}]\label{lem:nonstd schmidt rank}
    Let $\M$ be a von Neumann algebra on $\H$, let $\Psi,\Phi\in\H$ be unit vectors, let $\psi,\phi$ be the induced states on $\M$ and let $\psi',\phi'$ be the induced states on $\M'$.
    Then
    \begin{equation}\label{eq:nonstd schmidt rank}
        s_\psi=[\M'\Psi] \lesssim s_\phi =[\M'\Phi]\quad\iff\quad s_{\psi'}=[\M\Psi]\lesssim s_{\phi'}=[\M\Phi].
    \end{equation}
\end{lemma}

\begin{lemma}
    If $\Psi \barslocc\Phi$, then for each $\eps>0$, there exist $k\in\M$, $k'\in\M'$ such that $\norm{kk'\Psi-\Phi}<\eps$.
\end{lemma}
\begin{proof}
    This is a direct consequence of \cref{lem:pure state locc and kraus ops}.
\end{proof}

\begin{proof}[Proof of \cref{thm:slocc}]
    We first show \eqref{eq:slocc}: One direction has been shown in \cref{lem:pure state locc and kraus ops}. For the converse direction, if $kk'\Psi = \Phi$ we can extend the contractions $k/\norm{k}, k'/\norm{k'}$ to local instruments to obtain a suitable LOCC protocol preparing $\Phi$ with probability at least $1/(\norm{k}\norm{k'})$. 

    We now prove the equivalence of \cref{it:slocc1,it:slocc2,it:slocc3,it:slocc4,it:slocc5}.
    Note that \cref{lem:nonstd schmidt rank} already implies \ref{it:slocc2} $\Leftrightarrow$ \ref{it:slocc3}, and note that \ref{it:slocc5} $\Rightarrow$ \ref{it:slocc1} follows from \eqref{eq:slocc}.

    \ref{it:slocc1} $\Rightarrow$ \ref{it:slocc2}:
    Since all projections are equivalent in a type $\III$ factor, we may assume that $\M$ is semifinite with trace $\tr$.
    Let $k_n\in\M$, $k_n'\in\M'$ such that $\lim_n k_nk_n'\Psi=\Phi$.
    Using \cref{lem:rank lsc}, we get $\tr s_\phi = \tr [\M' \Phi] \le \liminf_n \tr[\M'k_nk_n'\Psi] \le \liminf_n \tr[k_n\M'\Psi] \le \tr [\M'\Psi] = \tr s_\psi$.
    Since for factors $p\lesssim q$ is equivalent to $\tr p\le \tr q$ ($p,q\in\proj\M$), this yields the claim.

    \ref{it:slocc2}\&\ref{it:slocc3} $\Rightarrow$ \ref{it:slocc4}:
    Pick partial isometries $v\in\M$ and $v'\in\M'$ such that $v^*v=[\M'\Psi]$, $v'^*v'=[\M\Psi]$, $vv^* \ge [\M'\Phi]$ and $v'v'^*\ge [\M\Phi]$.
    These properties ensure $[\M'vv'\Psi] =  [\M'v \Psi] \ni \Phi$ and $[\M vv'\Psi]=[\M v'\Psi]\ni\Phi$, which proves the claim.
    
    \ref{it:slocc4} $\Rightarrow$ \ref{it:slocc5}:
    Let $v,v'$ be partial isometries as in \ref{it:slocc4}.
    Let $\eps>0$. 
    Since $\Phi\in [\M'vv'\Psi]$, we can pick $a'\in\M'$ such that $\norm{\Phi- vk'\Psi}<\eps/2$ where $k'=a'v'$.
    Now pick a unitary $u\in\M$ such that $\norm{u\Psi-v\Psi}<\eps/2$.
    The claim follows from the triangle inequality.
\end{proof}

\begin{corollary} Let $\M$ be a factor of type $\III$ on $\H$ and $\Psi,\Phi\in \H$ be unit vectors. Then $\Psi \xrightarrow{\overline{\SLOCC}}\Phi$.
\end{corollary}
\begin{proof}
    In a type $\III$ factor, all projections are equivalent. Hence $s_\psi \lesssim s_\phi \lesssim s_\psi$. 
\end{proof}

Instead of considering (quasi-)exact SLOCC transformations, we now want to consider the maximum achievable fidelity in an SLOCC transformation.
Given a bipartite system of factors $(\M,\M')$ on $\H$ we define the SLOCC fidelity for unit vectors $\Psi,\Phi\in\H$ as
\begin{align}\label{eq:slocc-fidelity}
    F^2(\Psi \xrightarrow{\text{SLOCC}} \Phi) := \sup \{|\ip{\Phi}{\Omega}|^2\quad :\quad \Omega\in \H,\, \Psi \xrightarrow{\text{SLOCC}} \Omega\}.
\end{align}

\begin{lemma}\label{lem:slocc-std} Let $(\M,\M')$ be a bipartite system of factors on $\H$ and $\Psi,\Phi\in\H$. Denote by $(\M_0,\M_0')$ the standard bipartite system on $\H_0\subseteq \H$ induced by $\Phi$ according to \cref{lem:minimal subspace}.
    Then
    \begin{align}
        F^2(\Psi\slocc\Phi) = \sup \{|\langle \Phi,\Omega \rangle|^2\quad:\quad \Omega\in\H_0,\ s_{\omega} \lesssim s_{\psi}\}.
    \end{align}
\end{lemma}
\begin{proof}
   If $e=s_\phi s_{\phi'}$, we have $\H_0 = e\H$ and $e\Phi = \Phi$. Hence
    $|\ip{\Phi}{\Omega}|^2 = |\ip{\Phi}{e\Omega}|^2$ and $\Psi \slocc e\Omega$ by \cref{thm:slocc}. The claim now follows from using \cref{it:slocc2} in \cref{thm:slocc}.
\end{proof}

We briefly digress to discuss the Uhlmann fidelity \cite{uhlmann_transition_1976} on von Neumann algebras.
\begin{definition}
   Let $\psi,\phi$ be normal states on a von Neumann algebra $\M$. Then the \emph{Uhlmann fidelity} between $\psi$ and $\phi$ is
   \begin{align}
       F(\psi,\phi) = \sup_{s\in\M'}|\langle\Omega_\psi,  s \Omega_\phi\rangle|,
   \end{align}
   where the supremum is taken over contractions $s\in \M'$ where $\Omega_\psi$ and $\Omega_\phi$ are the representatives of $\psi$ and $\phi$ in the positive cone of the standard form of $\M$.
\end{definition}

If $(\M,\M')$ is a standard bipartite system on $\H$ and $\Psi \in \H$ is any purification of a state $\psi$ on $\M$, then $u \Omega_\psi = \Psi$ for some partial isometry $u\in \M'$. 
It follows that for any two purifications $\Psi,\Phi\in\H$ of $\psi$ and $\phi$, respectively, we have
\begin{align}
  |\langle\Psi,\Phi\rangle| = |\langle\Omega_\psi,u^*v\Omega_\phi\rangle| \leq F(\psi,\phi)
\end{align}
for appropriate partial isometries $u,v\in\M'$.
\begin{lemma}
    Let $(\M,\Tr)$ be a semifinite von Neumann algebra in standard form on $\H$. Then
    \begin{align}
        F(\psi,\phi) = \Tr|\rho_\psi^{1/2}\rho_\phi^{1/2}|,
    \end{align}
    where $\rho_\psi,\rho_\phi\in L^1(\M)$ are the densities representing $\psi$ and $\phi$, respectively.
\end{lemma}
\begin{proof}
    We identify $\H$ with $L^2(\M)$ in the usual way. 
    Then if $s\in \M'$ we have
    \begin{align}
        \langle \Omega_\psi, s \Omega_\phi \rangle = \Tr\rho_\psi^{1/2}\rho_\phi^{1/2} s.
    \end{align}
    Optimizing over $s$ we find
    \begin{align}
        \sup_s |\langle\Omega_\psi,s\Omega_\phi \rangle| = \sup_s \left|\Tr \rho_\psi^{1/2}\rho_\phi^{1/2} s\right| = \Tr|\rho_\psi^{1/2}\rho_\phi^{1/2}|.
    \end{align}
    The last equality follows from the fact that $\rho_\psi^{1/2}\rho_\phi^{1/2} = |\rho_\psi^{1/2}\rho_\phi^{1/2}| u$ for a partial isometry $u\in\M$. 
\end{proof}
Note that the proof shows that the optimization over $s$ in the definition of the fidelity can be restricted to partial isometries. 
\begin{remark}
    Even for type $\III$ algebras we can define the fidelity via the standard representation and obtain the same formula in terms of densities $\rho_\psi,\rho_\phi$, which are now elements of the Haagerup $L^1$ space. 
\end{remark}

We are grateful to F.~Hiai for communicating to us the following Lemma and its proof.
\begin{lemma}\label{lemma:fidelity-orbits}
    Let $\psi,\phi$ be states on a semifinite factor $\M$ with trace $\Tr$. Then
    \begin{align}
    \sup_{u\in\U(\M)}F(\psi, u\phi u^*) = \int_0^\infty \sqrt{\lambda_\rho(t)\lambda_\sigma(t)}\, dt,
    \end{align}
    where $\rho,\sigma\in L^1(\M)^+$ are the densities of $\psi,\phi$, respectively.
\end{lemma}
\begin{proof}
    It follows from \cite[Prop.~2.7]{fack1986generalized} and \cite[Thm.~4.2]{fack1986generalized} (see also \cite[Prop.~4.20, Prop.~4.42]{hiai_lecture_2021}) that
    \begin{align}
        F(\psi,\phi)=\Tr(|\rho^{1/2}\sigma^{1/2}|)= \int_0^\infty \lambda_{|\rho^{1/2}\sigma^{1/2}|}(t)\,dt \leq \int_0^\infty \lambda_{\rho^{1/2}}(t)\lambda_{\sigma^{1/2}}(t)\,dt = \int_0^\infty \sqrt{\lambda_\rho(t)\lambda_\sigma(t)}\,dt.
    \end{align}
    In the type $I$ case, the upper bound can clearly be achieved by choosing a unitary $u$ such that $[u\sigma u^*,\rho]=0$ and which reorders the eigenbasis appropriately. In the type $\II$ case, we choose an increasing family of projections $\{e_t : t\geq 0\}$ in $\M$ with $\Tr(e_t) =t$ for all $t\geq 0$ and define
    \begin{align}
        \tilde \rho = \int_0^\infty \lambda_\rho(t) de_t,\quad \tilde \sigma = \int_0^\infty \lambda_\sigma(t) de_t.
    \end{align}
    Since $\lambda_\rho(t)=\lambda_{\tilde \rho}(t)$ and $\lambda_\sigma(t)=\lambda_{\tilde \sigma}(t)$ for all $t>0$, by \cite[Lemma~4.1]{hiai1989distance} there exist sequences of unitaries $u_n, v_n\in \U(\M)$ such that
    \begin{align}
        \norm{\tilde \rho^{1/2} - u_n\rho^{1/2}u_n^*}_2 \to 0,\quad  \norm{\tilde \rho^{1/2} - v_n\rho^{1/2}v_n^*}_2 \to 0,
    \end{align}
    where $\norm{\placeholder}_2$ denotes the $L^2$ norm with respect to $\Tr$. 
    It follows from H\"older's inequality (see \cite[Thm.~4.2]{fack1986generalized} or \cite[Prop.~4.43]{hiai_lecture_2021}) that
    \begin{align}
        F(\psi,u_{n}^{*}v_{n}\phi v_{n}^{*}u_{n})=\Tr(|\rho^{1/2} u_n^* v_n \sigma^{1/2} v_n^* u_n|) &= \Tr(|u_n \rho^{1/2} u_n^* v_n \sigma^{1/2} v_n^*|) \to \Tr(|\tilde \rho^{1/2}\tilde\sigma^{1/2}|).
    \end{align}
    The latter is given by
    \begin{align}
        \Tr(|\tilde \rho^{1/2}\tilde\sigma^{1/2}|)&=\Tr\left(\int_0^\infty \lambda_\rho(t)^{1/2} \lambda_\sigma(t)^{1/2} de_t\right) = \int_0^\infty \sqrt{\lambda_\rho(t)\lambda_\sigma(t)}\,dt.
    \end{align}
    Thus, the upper bound can be achieved. 
\end{proof}

Using \cref{lemma:fidelity-orbits} we can compute the SLOCC fidelity from the Lorenz curve:
\begin{proposition} Let $(\M,\M')$ be a semifinite bipartite system of factors on $\H$ and let $\Psi,\Phi\in\H$ be unit vectors.
Then
\begin{align}
F^2(\Psi\xrightarrow{\SLOCC} \Phi)=  L_\phi(r(\Psi)),
\end{align}
where $r(\Psi) = \Tr(s_\psi)$ and $L_\phi$ is the Lorenz curve of the state $\phi$ (cp.~\eqref{eq:lorenz curve main}).
\end{proposition}
We will see in \cref{sec:monotones} that $r(\Psi)$ is a generalization of the \emph{Schmidt rank} for $\Psi$.
\begin{proof}
We only have to consider the case $\Tr s_\psi \leq \Tr s_\phi$, since otherwise $F^2(\Psi\slocc\Phi)=1=L_\phi(t)$ for $t\geq \Tr(s_\phi)$. Consider $e=s_\phi s_{\phi'}$ and set $\H_0 = e\H$. Since $\Tr s_\psi \leq \Tr s_\phi$ implies $s_\psi\lesssim s_\phi$ and $s_{\psi'}\lesssim s_{\phi'}$ by \cite[III.1.7.10]{blackadar_operator_2006} (or \cite{murray_rings_1937}), we can find partial isometries $v\in\M,v\in\M'$ such that $vv'\Psi\in \H_0$ and $v^*v = s_\psi, {v'}^*v'= s_{\psi'}$. Using \cref{lem:slocc-std} we can hence assume without loss of generality that $(\M,\M')$ is a standard bipartite system by restricting to $(\M_0,\M_0')$. 
By \cref{thm:slocc}, we can obtain any pure state $\Omega\in \H$  with $\Tr s_\omega =\Tr s_\psi$ to arbitrary precision. 
We, thus, want to maximize $F(\phi,\omega)$ under the constraint that $\Tr s_\omega = \Tr s_\psi$. Let $\rho_\phi,\rho_\omega \in L^1(\M)^+$ be the densities of $\phi$ and $\omega$. We find from \cref{lemma:fidelity-orbits} that
\begin{align}
\sup_{u\in\U(\M)} F(\phi, u\omega u^*) = \int_0^\infty \sqrt{\lambda_{\rho_\phi}(t) \lambda_{\rho_\omega}(t)}\,\chi_{(0,\Tr(s_\psi)]}(t) \, dt = \int_0^{\Tr(s_\psi)}\sqrt{\lambda_{\rho_\phi}(t) \lambda_{\rho_\omega}(t)}\, dt,  
\end{align}
and by H\"older's inequality
\begin{align}
    \sup_{u\in\U(\M)} F(\phi, u\omega u^*)  \leq \left(\int_0^{\Tr(s_\psi)}\lambda_{\rho_\phi}(t)\,dt\right)^{1/2} = \sqrt{L_\phi(\Tr(s_\psi))}.
\end{align}
For the converse direction, we can choose $\omega$ such that (see proof of \cref{lemma:fidelity-orbits}) 
\begin{align}
    \lambda_{\rho_\omega}(t) = \frac{1}{p} \lambda_{\rho_\phi}(t) \chi_{(0,\Tr(s_\psi)]}(t),\quad p = \int_0^{\Tr(s_\psi)} \lambda_{\rho_\phi}(t)\,dt, 
\end{align}
which yields
\begin{align}
    \sup_{u\in\U(\M)} F(\phi, u\omega u^*) = \sqrt{L_\phi(\Tr(s_\psi))}. 
\end{align}
\end{proof}

\subsection{Nielsen's theorem}\label{sec:nielsen}

In finite-dimensional quantum mechanics, Nielsen's theorem, first proven in \cite{nielsen_conditions_1999}, relates the possibility of a LOCC transformation mapping a pure state $\Psi$ on a bipartite system to another pure state $\Phi$ to a majorization relation on the marginals $\psi$ and $\phi$ on one of the subsystems.
In this section, we prove Nielsen's theorem for factors of arbitrary type.
Throughout this section, we consider a bipartite system of factors $(\M,\M')$ on $\H$.
We denote pure states by vectors $\Psi,\Phi\in \H$ and their marginals on $\M$ by $\psi$ and $\phi$.

\begin{theorem}[Nielsen's theorem]\label{thm:nielsen}
    Let $(\M,\M')$ be a bipartite system of factors in Haag duality on a Hilbert space $\H$.
    Let $\Psi,\Phi\in\H$ be unit vectors and denote by $\psi$ and $\phi$ the induced states on $\M$.
    Then 
    \begin{align}\label{eq:nielsen1}
        \Psi \xrightarrow{\LOCC} \Phi \quad &\iff \quad \psi = \sum_{x=1}^\oo p_x \,v_x\phi v_x^*
    \intertext{for a probability distribution $\{p_x\}_{x=1}^\oo$ and partial isometries $v_x\in\M$. Furthermore, it holds that}
        \Psi \xrightarrow{\overline{\LOCC}}\Phi \quad &\iff\quad\ \psi \prec\phi.
    \end{align}
\end{theorem}

In the case that $\M$ (and hence $\M'$) is of type $\I$ or type $\II$, the result was obtained in \cite{crann_state_2020} by different methods.
The type $\I_\oo$ case was obtained earlier in \cite{owari_convertibility_2008} (see also \cite{asakura_infinite_2017,massri_locc_2024}).

\begin{remark}
    Since we know from the outset that $\psi$ and $\phi$ are unit vectors, the partial isometries $v_x$ in \eqref{eq:nielsen1} necessarily satisfy $v^*v_x \ge s_\psi$. That is, the partial isometries act isometrically on $\phi$.
    In finite von Neumann algebras, every partial isometry can be extended to a unitary. Hence, for type $\II_1$ factors (the only finite factors that are not finite-dimensional), we obtain essentially the same statement as in finite-dimensional quantum mechanics. Nevertheless, there is a difference: In finite dimensions $\psi \prec \phi$ implies that one can convert $\Psi$ into $\Phi$ via a LOCC protocol with finitely many rounds, whereas in the type $\II_1$ case, one can only get arbitrarily close to $\Phi$.  
\end{remark}

The statement on approximate LOCC state transformations and majorization directly follows from the characterization of exact LOCC pure state transformations via partial isometries using the following Lemma, which relates majorization to convex mixtures arising from partial isometries (recall that $\piso(\M)$ denotes the set of partial isometries in $\M$).
\begin{lemma}
    Let $\M$ be a von Neumann algebra and let $\phi$ be a normal state on $\M$. Then 
    \begin{align}
        \overline{\conv} \{u\phi u^*\ :\ u\in \U(\M)\} = \overline{\conv} \{v\phi v^*\ :\ v\in \piso(\M),\, s_\phi \le v^*v\}
    \end{align}
    ($s_\phi \le v^*v$ ensures that $v\phi v^*$ is a state).
    If $\M$ is a finite factor, we have
    \begin{equation}
        \{v\phi v^*\ :\ v\in \piso(\M),\, s_\phi \le v^*v\} 
        =\{ u\phi u^*\ : \ u\in\U(\M)\}.
    \end{equation}
\end{lemma}
\begin{proof}
The first statement directly follows from \cite[Lem.~2.4]{haagerup1990equivalence}. The statement for finite factors follows because every partial isometry can be extended to a unitary.
\end{proof}
In the remainder of the section, we prove Nielsen's theorem, which we split into several Lemmas. 
We begin by proving it under the additional assumption that $(\M,\M')$ is a standard bipartite system.

\begin{lemma}[{\cite[Ex.~IX.1.2]{takesaki2}}]\label{lem:polar-decomposition} 
Let $\M$ be a von Neumann algebra with standard form $(\H,J,\P)$.
There exists a unique vector $|\Psi|\in \P$ and a unique partial isometry $u\in \M$ such that:
\begin{align}
    u\Psi = |\Psi|,\quad uu^* = [\M'|\Psi|] = s(|\psi|),\quad u^* u = [\M'\Psi] = s(\psi),
\end{align}
where $\psi$ is the marginal on $\M$ of $\Psi$ and $|\psi|$ that of $|\Psi|$.
\end{lemma}

\begin{lemma}\label{lem:twoway-to-oneway1}
    Let $\M$ be a von Neumann algebra with standard form $(\H,J,\P)$.
    Let $\Psi \in \H$, $m'\in\M'$ and $\Phi=m'\Psi \in \P$.
    Then there exists a unique partial isometry $u \in \M$ and $m\in\M$ such that 
    \begin{align}
        \Phi = j(u^*) m \Psi,\quad m = j(m')u, \quad u^* u = s(\psi) \quad uu^* =  s(|\psi|).
    \end{align}
    \begin{proof}
       By \cref{lem:polar-decomposition}, there is a partial isometry $u\in \M$  with $u\Psi = |\Psi|$. 
       We have $J\Psi = j(u^*)|\Psi| = j(u^*)u\Psi$ and $u$ is independent of $m'$. 
        Therefore 
        \begin{align}
          \Phi =   J\Phi = j(m')J\Psi = j(m') j(u^*) u\Psi.
        \end{align}
        Hence $\Phi = j(u^*) m\Psi$, where  $m= j(m')u$. Note that $m^* m = u^* j({m'}^* m')u$.
    \end{proof}
\end{lemma}

\begin{lemma}\label{lem:twoway-to-oneway2} 
    Let $\M$ be a factor with standard form $(\H,J,\P)$. Consider the standard bipartite system $(\M,\M')$ on $\H$.
    Let $\Psi \in \H$, $\Phi\in\P$ be unit vectors.
    If $\Phi$ can be prepared from $\Psi$ exactly via two rounds of LOCC, 
    then there also exists a finite set of Kraus operators $\{k_z\}_{z\in Z}\subset \M$ and a set of partial isometries $\{u'_z\}_{z\in Z}\subset \M'$ such that
    \begin{align}
        \sqrt{p_z} \Phi = u'_z k_z \Psi. 
    \end{align}
\begin{proof}
If $\Phi$ arises from two rounds of LOCC from $\Psi$, there exist Kraus operators $\{m_x\}_{x\in X_A}\subset\M$ and for each $x\in X_A$ Kraus operators $\{m'_{y|x}\}_{y\in Y_B}\subset\M'$ such that 
\begin{align}
    \sqrt{p_{y,x}} \Phi =   m'_{y|x}m_x \Psi, 
\end{align}
where $p_{y,x} \geq 0$ and $\sum_{y,x} p_{y,x} = 1$. 
Note that whenever $p_{y,x}>0$ we can assume that $m_x = s(\phi)m_x s(\psi)$ and $m'_{y|x} = s(\phi') m'_{y|x} s(\psi')$. 
Applying \cref{lem:twoway-to-oneway1} with $m' = m'_{y|x}$ to the vectors $\Psi_x = m_x\Psi$ and $\sqrt{p_{y,x}}\Phi$, we find that there exist partial isometries $u'_{x} = j(u_{x}^*)  \in \M'$ and operators $m_{y|x} = j(m'_{y|x}) u_{x} \in\M$ such that
\begin{align}
 \sqrt{p_{y,x}} \Phi = u_{x}' m_{y|x}m_x \Psi.
\end{align}
Defining $k_{y,x} = m_{y|x}m_x$ it follows from the proof of \cref{lem:twoway-to-oneway1} that
\begin{align}
0\leq \sum_{y,x: p_{y,x}>0} k_{y,x}^* k_{y,x} &= \sum_{y,x: p_{y,x}>0}m_x^* m_{y|x}^* m_{y|x} m_x = \sum_{y,x: p_{y,x}>0} m_x^* u_{x}^* j({m'}^*_{y|x}{m'}^*_{y|x})u_{x} m_x.
\end{align}
By \cref{lem:twoway-to-oneway1}, we have $u_x^* u_x = [\M' m_x\Psi] = [m_x s(\psi)\H] = [m_x\H]$, so that $u_x^* u_x m_x = m_x$.
Since $\sum_{y} {m'}^*_{y|x} m'_{y|x}=1$ for all $x$ we find
\begin{align}
    0\leq \sum_{y,x: p_{y,x}>0} k_{y,x}^* k_{y,x} &\leq 1.
\end{align}
We can therefore define an instrument $\{T_z\}_{z\in Z}$ with finite outcome space $Z$ using the  Kraus operators $k_z=k_{y,x}$ for those $(y,x)$ with $p_{y,x}>0$ and extending them to a full set of Kraus operators arbitrarily. 
The added Kraus operators occur with probability $0$ by construction.\footnote{Explicitly, $k_{y,x} = j(m'_{y,x})u_x \Psi_x = J m'_{y,x} J u_x \Psi_x = Ju_x m'_{y,x} m_x \Psi$, since $|\Psi_x| = u_x \Psi_x = u_x m_x \Psi$ by definition of $u_x$. Hence $\sum_{y,x:p_{y,x}>0} \langle \Psi, k^*_{y,x} k_{y,x}\Psi\rangle = \sum p_{x,y} = 1$.} We define $u'_z = u'_{y,x} = j(u^*_x)$ if $p_{y,x}>0$ and set $u'_z=1$ otherwise.
\end{proof}
\end{lemma}

\begin{corollary}\label{cor:oneway}
    Let $(\M,\M')$ be a standard bipartite system of factors on $\H$. Let $\Psi,\Phi\in\H$ be unit vectors.
    If $\Psi$ can be converted to $\Phi$ via LOCC, then it can be done by a single measurement on $\M$ followed by a partial isometry on $\M'$, depending on the measurement outcome.
    I.e., there exist  Kraus operators $\{k_x\}_{x\in X}$ in $\M$  and partial isometries $u_x'\in\M'$, such that
    \begin{equation}
        \sqrt{p_x} \Phi= k_xu_x' \Psi, \qquad p_x = \ip{\Psi}{k_x^* k_x\Psi}.
    \end{equation}
    \begin{proof}
        Iteratively apply the previous Lemma, in each step first absorbing the partial isometries into the Kraus operators on $\M'$ and afterward extending them to a valid set of Kraus operators again. 
    \end{proof}
\end{corollary}

\begin{proof}[Proof of "$\Rightarrow$" in \cref{thm:nielsen} for standard bipartite systems]
    We first show that if there exists a LOCC protocol that prepares $\Phi$ from $\Psi$ exactly, then $\psi \in \overline{\mathrm{conv}}\{u\phi u^* : u\in\U(\M)\}$. As before, we assume without loss of generality that $\Phi\in P$.
    By \cref{cor:oneway} (interchanging $\M$ and $\M'$) we have
    \begin{align}
        \sqrt{p_z}\Phi = u_z k'_z \Psi, 
    \end{align}
    for some finite set of probabilities $\{p_z\}_{z\in Z}$, partial isometries $\{u_z\}_{z\in Z}\subset\M$ and Kraus operators $\{k'_z\}_{z\in Z}\subset \M'$. 
    It follows from the fact that $\{k'_z\}$ are Kraus operators ($\sum_z {k'_z}^* k'_z = 1$) that ${u_z}^* u_z k'_z \Psi = k'_z \Psi$.\footnote{We have $r_z := \langle \Psi, {k'_z}^*k'_z \Psi\rangle \geq \sum_z \langle \Psi, {k'_z}^* u_z^* u_z k_z\Psi\rangle = p_z$ with $1 = \sum_z r_z = \sum_z p_z = 1$. Therefore $r_z=p_z$ and $u_z^* u_z k'_z \Psi = k'_z\Psi$.} Hence we have
    \begin{align}
        \sqrt{p_z}u_z^* \Phi = k'_z\Psi. 
    \end{align}
    Thus, for any $a\in\M$
    \begin{align}
        \psi(a) = \sum_z \langle\Psi, a {k'}_z^*k'_z \Psi\rangle = \sum_z p_z \langle\Phi, u_z a u_z^* \Phi\rangle = \sum_z p_z (u_z^* \phi u_z)(a).   
    \end{align}
    Thus $\psi \in \mathrm{conv}\{v \phi v^* : v\in\piso(\M)\}$.
\end{proof}

For the converse direction of Nielsen's theorem, we use (Tomita-Takesaki) modular theory to construct an LOCC protocol converting $\Omega_\psi$ to $\Omega_\phi$ provided that $\psi$ can be obtained as a convex combination of partial isometries applied to $\phi$:

\begin{lemma}\label{lem:cocycle}
    Let $\phi$ be a state on a von Neumann algebra $\M$, let $v_1,\ldots,v_n\in\M$ be a collection of partial isometries with $v_x^*v_x\ge s_\phi$, and let $p_x\ge0$ be such that $\sum_{x=1}^n p_x =1$.
    Set $\psi = \sum_x \psi_x$, where $\psi_x=p_x\,v_x\phi v_x^*$, and $k_x := v_x^*[D\psi_x:D\psi]_{-i/2} \in \M$.
    It holds that
    \begin{equation}
        k_x Jv_x^*J \Omega_\psi = \sqrt{p_x} \Omega_\phi,\qandq \sum_x k_x^*k_x=1.
    \end{equation}
\end{lemma}

To understand the appearance of the Connes cocycle, or rather its analytic continuation, in \cref{lem:cocycle}, let us consider the case where $(\M,\tr)$ is a semifinite von Neumann algebra.
In this case the analytic continuation of the Connes cocycle is simply given by $[D\omega:D\varphi]_{-i/2}= (\rho_\omega)^{\frac12}(\rho_\varphi)^{-\frac12}$, where $\rho_\omega$ and $\rho_\varphi$ are the density operators implementing faithful positive linear functionals $\omega,\varphi\in\M_*^+$ with respect to the trace.
Since $\rho_{v_x\phi v_x^*} = v_x\rho_\phi v_x^*$, the Kraus operators in \cref{lem:cocycle} may be written as 
\begin{equation}
    k_x = v_x^* (\rho_{\psi_x})^{\frac12} (\rho_\psi)^{-\frac12}= \sqrt{p_x}(\rho_\phi)^{\frac12} v_x^* (\rho_\psi)^{-\frac12},
\end{equation}
and this is exactly how they are defined in the proof of the finite-dimensional case (see, e.g., \cite[Sec.~12.5.1]{nielsen_chuang}).

\begin{proof}
    Since $\psi \ge \psi_x$ for all $x$, \cite[Lem.~A.24]{hiai2021quantum} (and the discussion after \cite[Lem.~A.59]{hiai2021quantum}), implies
    \begin{align}
        v_xk_x\Omega_\psi = v_x v_x^* [D\psi_x :D\psi]_{-i/2}\Omega_{\psi} = v_x v_x^* \Omega_{\psi_x} = \sqrt{p_x}v_x v_x^* v_x Jv_xJ \Omega_\phi = \sqrt{p_x} v_x Jv_xJ \Omega_\phi = \Omega_{\psi_x},
    \end{align} 
    since $v^*_x v_x v^*_x = v^*_x$ for any partial isometry. Similarly, we also have $v^*_x v_x k_x = k_x$.
    We thus find
    \begin{align}
        \sqrt{p_x} \Omega_\phi = Jv_x^*J v_x^* v_xk_x\Omega_\psi  = k_x Jv_x^* J \Omega_\psi.
    \end{align}
    This proves the first equation. For the second one, let $b,c\in\M'$ be arbitrary and note that 
    \begin{align*}
        \ip{b\Omega_\psi}{\sum_x k_x^*k_xc\Omega_\psi} & = \sum_x \ip{v_xk_x\Omega_\psi}{b^*c v_xk_x\Omega_\psi} = \sum_x \ip{\Omega_{\psi_x}}{b^*c\Omega_{\psi_x}} \\
        & = \sum_x (\psi_x)'(b^*c) = (\sum_x \psi_x)'(b^*c) = \psi' (b^*c) \\
        &=\ip{\Omega_\psi}{b^*c\Omega_\psi} = \ip{b\Omega_\psi}{c\Omega_\psi}.
    \end{align*}
\end{proof}

\begin{proof}[Proof of "$\Leftarrow$" in \cref{thm:nielsen} for standard bipartite systems] Without loss of generality, we can assume $\Psi = \Omega_\psi\in P$ and $\Phi = \Omega_\phi \in P$, since state vectors can be mapped to $P$ by local partial isometries.
By \cref{lem:cocycle} there exist Kraus operators $\{k_x\}$ such that
\begin{align}
    j(v_x^*) k_x\Omega_\psi = \sqrt{p_x}\Omega_\phi. 
\end{align}
Thus, the LOCC protocol which consists of first applying the instrument $\{k^*_x(\placeholder)k_x\}$ on $\M$ and then the instrument $\{T_{y|x}\}_{y=1}^2$ on $\M'$ defined by Kraus operators $\{m'_{1|x} = j(v_x^*), m'_{2|x} =  1 - j(v_xv_x^*)\}$ prepares $\Phi$ from $\Psi$.
Note that the outcome $y=2$ of the instrument on $\M'$ happens with probability $0$.
\end{proof}

It remains to prove the case where $(\M,\M')$ is not a standard bipartite system.
Let us fix a unit vector $\Psi\in\H$. Consider the projection $e = s_\psi s_\psi'$ (where $\psi$ and $\psi'$ are the induced states on $\M$ and $\M'$, respectively).
From \cref{lem:minimal subspace}, we know that the induced bipartite system $(\M_0,\M_0')= (e\M e,e\M'e)$ on
\begin{equation}\label{eq:H_0}
    \H_0 := [\M'\Psi]\cap  [\M\Psi] = s_\psi s_{\psi'}\H
\end{equation}
is a standard bipartite system (note that $(e\M e)'=e\M'e$ \cite[Prop.~5.5.6, Cor.~5.5.7]{KadisonRingrose1}).
Roughly speaking, our strategy to prove Nielsen's theorem is to use \cref{thm:slocc} to ensure that we may assume that the vector $\Phi$ is an element of $\H_0$, and to thereby reduce the general case to the case of standard bipartite systems.

\begin{lemma}\label{lem:wlog std}
    If $\Phi\in\H_0$ is a unit vector, then the following are equivalent:
    \begin{enumerate}[(a)]
        \item\label{it:wlog std1} $\Psi$ can be transformed to $\Phi$ in $\LOCC$ by the bipartite system $(\M,\M')$ on $\H$,
        \item\label{it:wlog std2} $\Psi$ can be transformed to $\Phi$ in $\LOCC$ by the standard bipartite system $(\M_0,\M_0')$ on $\H_0$
    \end{enumerate}
\end{lemma}
\begin{proof}
    \ref{it:wlog std1} $\Rightarrow$ \ref{it:wlog std2}: By truncating all Kraus operators of the overall instrument to the subspace $\H_0$, we get an LOCC protocol of the bipartite system $(\M_0,\M_0')$ that maps $\Psi$ to $\Phi$.
    \ref{it:wlog std2} $\Rightarrow$ \ref{it:wlog std1}: We can simply add additional Kraus operators to ensure unitality on $\H$. 
\end{proof}

\begin{proof}[Proof of \cref{thm:nielsen} in the general case:]
    We only have to show eq.~\eqref{eq:nielsen1}.
    "$\Rightarrow$":
    In particular, we have $\Psi\barslocc\Phi$ so that \cref{thm:slocc} implies that we can find partial isometries $v\in\M$, $v'\in\M'$ which act isometrically on $\Psi$ and satisfy $\Phi \in [\M' vv'\Psi]\cap [\M v'v\Psi]$.
    Thus, we may simply assume that $\Phi \in \H_0$ and that $s_\psi \ge s_\phi$ (see \eqref{eq:H_0}).
    By \cref{lem:wlog std}, it follows that the standard bipartite system $(\M_0,\M_0')$ on $\H_0$ admits an LOCC protocol sending $\Psi$ to $\Phi$.
    Using that Nielsen's theorem has been proved for standard bipartite systems, we get a probability distribution $\{p_x\}$ and partial isometries $v_x\in\M_0$ such that $\phi_0=\sum_x p_xv_x\psi_0 v_x^*$, where $\psi_0$ and $\phi_0$ are the states induced by $\Psi$ and $\Phi$ on $\M_0$.
    Since $\M_0 = s_\psi \M s_\psi$ and $s_\psi\ge s_\phi$, we get $\phi = \sum_x p_x v_x\psi v_x^*$.

    "$\Leftarrow$": 
    It follows from the assumption that $s_\phi \lesssim s_\psi$.
    By \cref{it:slocc4} of \cref{thm:slocc}, can use LOCC to map $\Psi$ to a vector $\tilde\Psi$ such that $\Phi \in [\M'\tilde\Psi]\cap[\M\tilde\Psi]$.
    The claim now follows from \cref{lem:wlog std}.
\end{proof}

\subsection{Pure state LOCC and different types}
\label{sec:trivialization-III}

In this section, we first discuss pure state LOCC transformations for bipartite systems $(\M,\M')$ on $\H$ where $\M$ (and hence $\M'$) is a  type $\III$ factor.
We then show that whenever $\M$ is not of type $\I$, every pure state $\Omega\in\H$ has the property that arbitrary finite-dimensional entangled states may be prepared from it via LOCC. 

\begin{theorem}\label{thm:LOCC-trivial}
    Let $(\M,\M')$ be a  bipartite system of factors on $\H$ with $\M,\M'\neq \CC$. Then the following are equivalent:
    \begin{enumerate}[(a)]
        \item\label{it:LOCC-trivial1} $\M$ has type $\III$ (equivalently, $\M'$ has type $\III$),
        \item\label{it:LOCC-trivial2} $\Psi \xrightarrow{\overline{\LOCC}}\Phi$ for any two unit vectors $\Psi,\Phi \in \H$.
    \end{enumerate}
    Moreover, $\M$ (hence $\M'$) has type $\III_1$ if and only if for any two unit vectors $\Psi,\Phi \in \H$ and any $\eps>0$ there exist unitaries $u\in \M$, $v\in\M'$ such that
    \begin{align}
        \norm{u v\Psi - \Phi} <\eps.
    \end{align}
\end{theorem}
Hence, $\M$ has type $\III$ if and only if every pure state can be transformed into any other pure state to arbitrary accuracy by an LOCC protocol. It has type $\III_1$ if and only if the same is true via local operations without using classical communication.  

The proof of \cref{thm:LOCC-trivial} relies on the following two results. 
The first was shown by Haagerup and St\o rmer in \cite{haagerup1990equivalence}:
\begin{lemma}[{\cite[Lem.~9.3]{haagerup1990equivalence}}]\label{lem:majorization_in_typeIII-mt}
    Let $\M$ be a $\sigma$-finite type $\III$ factor. Then $\psi\prec \phi$ for all pairs of normal states $\phi,\psi$ on $\M$.
\end{lemma}

The second result is a famous result by Connes and St\o rmer, known as the `homogeneity of the state space' of type $\III_1$ factors:
    \begin{theorem}[\cite{connes_homogeneity_1978}] A factor $\M\neq\CC$ with separable predual if of type $\III_1$ if and only if for every two normal states $\psi$ and $\phi$ on $\M$ and every $\eps>0$ there exists a unitary $u\in \M$ such that
    \begin{align}
        \norm{u\psi u^* - \phi} <\eps.
    \end{align}    
\end{theorem}

\begin{proof}[Proof of \cref{thm:LOCC-trivial}]
\ref{it:LOCC-trivial1} $\Rightarrow$ \ref{it:LOCC-trivial2} follows immediately by applying \cref{lem:majorization_in_typeIII-mt} to Nielsen's theorem. 
For the converse direction \ref{it:LOCC-trivial2} $\Rightarrow$ \ref{it:LOCC-trivial1},
we first observe that by \cite[Prop.~V.3.13, Ex.~V.1]{takesaki1} and \cref{lem:tensor-standard} (equivalently, from the discussion in \cref{sec:monotones}), either all states on $\M$ or all states on $\M'$ admit vector representations in $\H$. We may assume that this is the case for $\M$. If $\M$ is a semifinite factor with $\M\neq \CC$, there exist two states $\psi,\phi$ such that $\psi\not\prec\phi$ (this may easily be seen by considering the spectral scales).
By Nielsen's theorem, we cannot convert the vector representative $\Psi\in\H$ of $\psi$ into the vector representative $\Phi\in\H$ of $\phi$ via $\overline{\LOCC}$.

The statement about type $\III_1$ factors follows from the homogeneity of the state space: First, assume that $\M$ has type $\III_1$. By \cref{lem:polar-decomposition}, we can find  partial isometries $v,w\in\M'$ such that $\Psi = v\Omega_\psi$ and $\Phi = w\Omega_\phi$. By homogeneity of the state space, for every $\eps>0$ there exists a unitary $u\in \M$ such that $\norm{uj(u)\Omega_\psi - \Omega_\phi}<\eps$. Hence, there exist partial isometries as stated in the theorem mapping $\Psi$ to a vector $\eps$-close to $\Phi$. Since the unitary group for $\M$ is strongly dense in $\overline{\U}(\M)$, the result follows.
For the converse direction, we consider any two normal states $\psi,\phi$ on $\M$.
By the assumption, we know that \ref{it:LOCC-trivial2} holds, and $\M$ is of type $\III$. Since any type $\III$ von Neumann algebra acting on a separable Hilbert space is in standard form, we find vector representatives $\Omega_{\psi},\Omega_{\phi}\in\H$. Moreover, for every $\eps>0$ there exist unitaries $u\in \M$, $v\in\M'$ such that $\norm{uv\Omega_\psi - \Omega_\phi}<\tfrac{\eps}{2}$. Since $v\Omega_{\psi}$ has the same marginal on $\M$ as $\Omega_{\psi}$, we find
\begin{align}
\norm{u \psi u^* - \phi} & = \sup_{\substack{a\in\M \\ \norm{a}=1}}|\ip{v\Omega_{\psi}}{u^{*}auv\Omega_{\psi}}-\ip{\Omega_{\phi}}{a\Omega_{\phi}}| \leq 2\norm{uv\Omega_{\psi} - \Omega_\phi} < \eps.
\end{align}
This implies the homogeneity of the state space and, thus, that $\M$ is of type $\III_{1}$ \cite{connes_homogeneity_1978}.
\end{proof}

Given that pure state LOCC trivializes for type $\III$, it is interesting to ask what is possible in type $\II$. To this end, we now prove a theorem characterizing the entanglement properties of irreducible, standard bipartite systems that are not of type $\I$.
To state it we define the \emph{Bell coefficient} of a density matrix $\rho$ on $\H$ via
\begin{align}
    \beta(\rho,\M,\M') := \sup \Tr \rho \left(a_1(b_1+b_2) + a_2(b_1-b_2)\right) \leq 2\sqrt{2},
\end{align}
where the optimization is over $a_i\in \M$ and $b_i\in \M'$ such that $-1\leq a_i\leq 1$ and $-1\leq b_i\leq 1$. 
It measures the maximal value that can be achieved in the CHSH game \cite{clauser_proposed_1969,summers_tangent}. 

\begin{theorem}\label{thm:equivalence-bell}
     Let $(\M,\M')$ be bipartite system of factors on $\H$ in Haag duality. The following are equivalent:
    \begin{enumerate}[(a)]
        \item\label{item:LOCC-equiv1} For every $\Omega\in \H$, every $n\in\NN$, and every unit vector $\Psi\in \CC^n\otimes \CC^n$ we have
        \begin{align}
    \Omega\ox \ket{1}\ket{1} \xrightarrow{\overline{\LOCC}} \Omega' \ox \Psi,\quad \Omega'\in\H.
        \end{align}
    \item\label{item:notI} $\M$ is not of type $\I$. 
    \end{enumerate}
    If $\M$ is approximately finite-dimensional, both items are equivalent to:
    \begin{enumerate}[(a),resume]
  \item\label{item:CHSH1} For every density matrix $\rho$ on $\H$ we have $\beta(\rho,\M,\M') = 2\sqrt 2$.
    \end{enumerate}
\end{theorem}
\begin{proof}
For the implication \ref{item:LOCC-equiv1} $\Rightarrow$ \ref{item:notI}, suppose conversely that $\ref{item:LOCC-equiv1}$ holds true for some type $\I$ factor $\M$. In this case $\H = \H_A\ox \H_B$ and there exists a product state $\Phi_A\ox\Phi_B\in \H$. 
Considering $\Omega = \Phi_A\ox\Phi_B$ yields a contradiction since LOCC cannot prepare an  entangled state from a non-entangled state.
For the converse direction, we only have to consider the case that $\M$ has type $\II$ because of \cref{thm:LOCC-trivial} and since $\M\ox M_n(\CC)\cong \M$ for any properly infinite factor $\M$. 
We can furthermore restrict to the case $\Psi = \Omega_n:= \frac{1}{\sqrt{n}}\sum_{j=1}^n \ket j\ket j$, since every $\Psi \in \CC^n\ox \CC^n$ may be prepared from $\Omega_n$ via LOCC (the tracial state is majorized by all density matrices).
We thus have to show that for every $\rho\in L^1(\M)^+$ with $\Tr\rho=1$ on a type $\II$ factor $(\M,\Tr)$ we have 
\begin{align}
    \rho \ox |1\rangle\langle 1| \prec \rho'\ox \frac{1}{n}
\end{align}
for some $\rho'\in L^1(\M)^+$ with $\Tr\rho'=1$. 
The spectral scales fulfill \cite{long_paper}:
\begin{align}
    \lambda_{\rho\ox |1\rangle\langle 1|}(t) = \lambda_\rho(t),\quad \lambda_{\rho'\ox \frac{1}{n}} = \frac{1}{n}\lambda_{\rho'}\big(\frac{t}{n}\big),\quad t\in(0, n\Tr(1)),
\end{align}
where $\lambda_\rho$ is extended by zero outside the interval $(0,n\Tr(s_\rho))$ if $\Tr(s_\rho)<\Tr(1)$. 
We thus have to show that we can find $\rho'$ such that
\begin{align}
    \int_0^t \frac{1}{n}\lambda_{\rho'}(s/n)\,ds\geq \int_0^t \lambda_\rho(s)\,ds,\quad t\in(0,n\Tr(1)). 
\end{align}
This is achieved by choosing $\rho'$ via 
\begin{align}
    \lambda_{\rho'}(t) := n\, \lambda_\rho(t n)\,\chi_{(0,\Tr(s_\rho)/n)}.
\end{align}
For $t\leq \Tr(s_\rho)/n$) we get:
\begin{align}
    \int_0^t \lambda_{\rho'}(s)\,ds = n \int_0^{t} n \lambda_\rho(s n)\,ds = \int_0^{t n}\lambda_\rho(s)\,ds \geq \int_0^t \lambda_\rho(s)\,ds. 
\end{align}
For $t>\Tr(s_\rho)/n$ we have $\int_0^t \lambda_{\rho'}(s)\,ds=1 \geq \int_0^t \lambda_\rho(s)\,ds$. 

We now assume that $\M$ is approximately finite-dimensional.
Since product states exist in the case of type $\I$, the implication 
\ref{item:CHSH1} $\Rightarrow$ \ref{item:notI} is clear. 
For the converse direction: Every approximately finite-dimensional factor $\M$ that is not of type $\I$ is \emph{strongly stable} (or McDuff): $\M \cong \M\ox \mc R_1$, where $\mc R_1$ is the (unique) approximately finite-dimensional type $\II_1$ factor.\footnote{This follows from the fact that the flow of weights \cite{connes_flow_1977} is a complete invariant for approximately finite-dimensional factors \cite{connes_classification_1976,haagerup_connes_1987} and is invariant under semifinite amplifications (see for example \cite{long_paper} for a detailed discussion).} We now make use of \cite[Thm.~2.3]{summers_maximal_1987}, which shows that on a standard bipartite system $(\M,\M')$  with $\M$ strongly stable, every normal state, i.e., every density matrix on $\mc H$, maximally violates the CHSH inequality.
This shows the claim if $(\M,\M')$ is a standard bipartite system. 
If it is not standard, we consider the standard bipartite system $(\tilde \M,{\tilde \M}')=(\B(\H)\ox 1\ox \M,1\ox \B(\H)\ox \M')$ on $\tilde{\mathcal H} = \H\ox\H\ox\H$ with density matrix $\tilde \rho = \ketbra{1}{1}\ox\ketbra{1}{1}\ox\rho$ \cite[Cor.~III.2.6.16]{blackadar_operator_2006}. 
We then have $\beta(\tilde\Omega,\tilde \M,{\tilde \M}')=2\sqrt{2}$. For any $\eps>0$, let $\tilde a_i\in \tilde M,\tilde b_i\in {\tilde M}'$ be contractions such that 
\begin{align}
    \Tr\tilde \rho\left(\tilde a_1(\tilde  b_1+\tilde b_2)+\tilde a_2(\tilde b_1-\tilde b_2)\right) > 2\sqrt{2}(1-\eps).
\end{align}
If we denote by $v$ the isometry $\H\to \H^{\ox 3}$ defined by $\Psi\mapsto \ket1\ox\ket 1\ox \Psi$, the contractions $a_i = v^*\tilde a_i v \in \M$, $b_i = v^*\tilde b_i v\in \M'$ fulfill
\begin{align}
    \Tr \rho\left(a_1(b_1+b_2)+a_2(b_1-b_2)\right) =  \Tr\tilde \rho\left(\tilde a_1(\tilde  b_1+\tilde b_2)+\tilde a_2(\tilde b_1-\tilde b_2)\right) > 2\sqrt{2}(1-\eps), 
\end{align}
which finishes the proof.
\end{proof}

To distinguish the two subtypes of type $\II$, we need the notion of maximally entangled states:

\begin{definition}
    Let $(\M,\M')$ be a bipartite system of factors in Haag duality on a Hilbert space $\H$. 
    A state vector $\Psi\in\H$ is \emph{maximally entangled} if $\Psi\barlocc\Phi$ holds for every other state vector $\Phi\in\H$.
\end{definition}

By Nielsen's theorem, maximally entangled states are precisely characterized by having maximally mixed marginal states, i.e., states $\psi$ such that $\psi \prec \phi$ for all states $\phi$.

\begin{proposition}
    Let $(\M,\M')$ be a bipartite system of factors in Haag duality on a Hilbert space $\H$.
    Then maximally entangled states exist if and only if it is not true that $\M$ and $\M'$ are both of type $\I_\oo$ or of type $\II_\oo$.
\end{proposition}
\begin{proof}
    By virtue of Nielsen's theorem, this follows directly from the fact that maximally mixed states exist precisely in factors of type $\I_n$, $n<\oo$, $\II_1$ and $\III$.
\end{proof}


\subsection{Entanglement monotones for semifinite bipartite systems}\label{sec:monotones}

A central goal of entanglement theory is to \emph{quantify} entanglement.
Since entanglement cannot be measured by a single number, one usually measures entanglement via so-called \emph{entanglement monotones}: Functions on multipartite quantum states that can either only decrease or only increase in an LOCC protocol. We now consider entanglement monotones for pure state LOCC transformations on bipartite systems of factors in Haag duality. 
Since pure state LOCC transformations trivialize for type $\III$ systems we only consider semifinite systems. We note, however, that non-trivial entanglement measures still exist and are relevant in QFT if we consider pairs of subsystems that are properly spacelike separated \cite{hollands_entanglement_2018}.

For the remainder of this section, we fix a bipartite system $(\M,\M')$ of semifinite factors on $\H$.
To define entanglement monotones for a semifinite bipartite system $(\M,\M')$, we can use the trace on $\M$ or the trace on $\M'$.
It is only natural to ask whether the resulting monotones depend on this choice.
That this is not the case follows from the work of Murray and von Neumann \cite{murray_rings_1936,murray_rings_1937} (see also \cite[Sec.~V]{takesaki2}):
First, note that the trace on a factor is unique only up to scaling. 
Thus, it is uniquely determined by its value on a single finite projection.\footnote{E.g., the usual trace on a type $\I$ factor is determined by declaring $\tr P=1$ for a one-dimensional projection $P$.}
Thus, the question is whether the traces on $\M$ and $\M'$ can be chosen in a matching way.
To do this, we  pick a vector $\Psi\in\H$ such that $[\M'\Psi]$ is a finite projection in $\M$, which is equivalent to $[\M\Psi]$ being a finite projection in $\M'$ \cite[Sec.~X]{murray_rings_1936}, and fix the relative scaling of the traces through the equation
\begin{equation}\label{eq:trace coupling}
    \tr_{\M} [\M'\Psi] = \tr_{\M'} [\M\Psi].
\end{equation}
It then follows that \eqref{eq:trace coupling} holds for all vectors $\Psi\in\H$ \cite[Sec.~X]{murray_rings_1936}.
If $\M$ and $\M'$ are finite, the number
\begin{equation}\label{eq:coupling const}
    c(\M,\M') = \frac{\tr_{\M'} 1}{\tr_{\M}1} \in (0,\oo)
\end{equation}
is the \emph{coupling constant} of Murray-von Neumann, which measures the "size" of $\M'$ relative to $\M$ \cite[Def.~V.3.]{takesaki2}.
If we extend the coupling constant to infinite factors using \eqref{eq:coupling const} with the conventions $\frac t\oo=0$, $\frac\oo t=\oo$ and $\frac\oo\oo=1$ (for $t>0$), then one gets
$c(\M,\M')=c(\M',\M)^{-1}$ (with $0^{-1}=\oo$ and $\oo^{-1}=0$) and%
    \footnote{If $\M$ and $\M'$ are finite, this is proved in \cite[Prop.~V.3.13]{takesaki2}. If $c(\M,\M')$ is $0$ or $\oo$, i.e., if only one factor is finite and the other one is infinite, they cannot be in standard representation since then $\M$ and $\M'$ would be anti-isomorphic. If both algebras are infinite, they are in standard representation by \cite[Cor.~III.2.6.16]{blackadar_operator_2006}.}
\begin{equation}
    (\M,\M') \ \text{is standard}
    \quad \iff \quad c(\M,\M')=1.
\end{equation}
Next, we will show that, with the scaling of the traces on $\M$ and $\M'$ fixed by \eqref{eq:trace coupling}, the marginal states of all bipartite pure states have the same spectral scales.
As before, if $\Psi\in\H$ is a unit vector, we denote by $\psi$ and $\psi'$ the induced states on $\M$ and $\M'$, respectively.

\begin{proposition}\label{prop:schmidt spectrum}
    Let $\Psi\in \H$ be a unit vector and let $\rho\in L^1(\M)$ and $\rho'\in L^1(\M')$ be the densities of $\psi$ and $\psi'$, respectively.
    Then
    \begin{equation}
        \lambda_\rho(t) = \lambda_{\rho'}(t),\qquad t>0.
    \end{equation}
\end{proposition}

\begin{remark}
If $c(\M,\M') < 1$ then $\M'$ is finite and there exists a unit vector $\Psi\in \H$ that induces the faithful tracial state on $\M'$, but not on $\M$, which need not even have a tracial state. In this case, we say that (the subsystem described by) $\M'$ is smaller than $\M$. 
$\Psi$ then precisely corresponds to a maximally entangled state. 
In the type $\I$ case $\M=M_{n_A}(\CC), \M'=M_{n_B}(\CC)$ we simply have $n_B<n_A$. 
Since the Schmidt spectrum captures all pure state entanglement properties and is fully encoded on $\M'$, entanglement properties can only capture the smaller of the two subsystems.
\end{remark}

\begin{lemma}\label{lem:traces on std subsystem}
    Let $\Psi\in \H$ be a unit vector and let $(\M_0,\M_0')$ be the standard bipartite system obtained by reducing $(\M,\M')$ onto $\H_0= [\M'\Psi]\cap[\M\Psi]$ (cp.\ \cref{lem:minimal subspace}).
    Let $J_0$ be the modular conjugation defined by the cyclic separating vector $\Psi_0:=\Psi\in\H_0$.
    Then
    \begin{equation}
        \tr_{\M_0}(a) = \tr_{\M_0'}(JaJ),\qquad a\in L^1(\M_0).
    \end{equation}
\end{lemma}

In \cref{lem:traces on std subsystem}, the trace on the reduced algebra $\M_0$ is the reduced trace, i.e., if $e=s_\psi s_{\psi'}$ is the projection onto $\H_0$, then $\tr_{\M_0}(eae)=\tr_\M( s_\psi a s_\psi)$ for $a\in L^1(\M)$, and similarly for the commutant.

\begin{proof}
    Let $p_0\in \M_0$ be a finite nonzero projection that commutes with $\psi_0$ (e.g., a spectral projection of the density $\rho_0$ of $\psi_0$) and let $p_0'=J_0p_0J_0$. 
    Then we have $p_0\Psi_0=p_0p_0'\Psi_0 = p_0'\Psi_0$.
    Set $e=s_\psi s_{\psi'}$ and note that $e\H=\H_0$.
    Let $p\in\M$ be the unique projection such that $p\le s_\psi$ and $epe=p_0$, and let $p'\in\M'$ be defined analogously.
    Then $p\Psi= pp'\Psi = p'\Psi$ and, hence, $p = [p\M'\Psi]=[\M'pp'\Psi]$ and similarly for $p'$.
    Then
    \begin{align*}
        \tr_{\M_0} p_0 = \tr_{\M} p = \tr_\M [\M'pp'\Psi] = \tr_{\M'} [\M pp'\Psi] = \tr_{\M'} p' =\tr_{\M_0'} p_0' = \tr_{\M_0'} J_0p_0J_0,
    \end{align*}
    where we used \cref{eq:trace coupling}.
    Since the two traces $\tr_\M$ and $\tr_{\M'}(J(\placeholder)J)$ on $\M_0$ are unique up to scaling, it follows that they coincide.
\end{proof}

\begin{proof}[Proof of \cref{prop:schmidt spectrum}]
    We consider the standard bipartite system $(\H_0,\M_0,\M_0')$ obtained by reducing $(\H,\M,\M')$ with the projection $e= [\M'\Psi]\cap[\M\Psi]$ (cp.\ \cref{lem:wlog std}).
    Put $\Psi_0 = \Psi\in\H_0$, and denote by $\psi_0$, $\rho_0$, $\psi_0'$ and $\rho_0'$ the induced objects associated with $\M_0$ and $\M_0'$.
    Let $J_0,\P_0$ be the modular conjugation and positive cone determined by the cyclic separating vector $\Psi_0$ (now $(\H_0,J_0,\P_0)$ is the standard form of $\M_0$), and note that $\Psi_0=\Omega_{\psi_0}\in\P_0$.
    It follows that
    \begin{equation}
        \rho_0' =  J_0\rho_0 J_0,
    \end{equation}
    $\M_0$ is isomorphic with $s_\psi\M s_\psi$ in such a way that $\psi_0$ is identified with $\psi$ and similarly for the commutant. 
    Moreover, \cref{lem:traces on std subsystem} shows that the traces of $\M_0$ and $\M_0'$ are compatible. Thus, we have $\lambda_\rho = \lambda_{\rho_0} = \lambda_{J\rho_0J} = \lambda_{\rho_0'}=\lambda_{\rho'}$.
\end{proof}

Proposition~\ref{prop:schmidt spectrum} shows that the spectral scale of the marginals is the natural generalization of the \emph{Schmidt spectrum} for pure states on type $\I$ systems to the general semifinite case.\footnote{As discussed in \cite{long_paper}, the spectral scale of a marginal state $\psi$ is in one-to-one correspondence with the spectral state $\hat{\psi}$ introduced by Haagerup and St{\o}rmer \cite{haagerup1990equivalence}. In this sense, the spectral state, which is a normal state on the flow of weights, can be interpreted as a further generalization of the Schmidt spectrum beyond the semifinite case.} It ensures that the entanglement monotones defined in the following remain the same if $\M$ and $\M'$ are interchanged. 

From Nielsen's theorem in conjunction with results from noncommutative majorization theory (specifically, \cref{it:factor majorization2-mt} in \cref{thm:factor-majorization-mt}), we immediately find (see also \cite{crann_state_2020}):
\begin{lemma}\label{lem:f monotone}
    Let $f:\RR^+\to \RR^+$ be a continuous convex function. Then the function
    \begin{align}\label{eq:f monotone}
        M_f(\Psi) := \Tr f(\rho) = \int_0^\oo f(\lambda_\rho(t))\,dt , 
    \end{align}
    defined on unit vectors of $\H$ (using the notation from above), is a pure state entanglement monotone: 
    \begin{equation}
        \Psi \xrightarrow{\overline\LOCC} \Phi \quad\implies\quad M_f(\Psi)\le M_f(\Phi).
    \end{equation}
\end{lemma}

Proposition \ref{prop:schmidt spectrum} guarantees that \eqref{eq:f monotone} is independent of exchanging $\M$ for $\M'$.
The numerical value of $M_f(\rho)$ depends on the scaling of the traces only up to a constant independent of $\Psi$. As a special case of these monotones, we obtain entanglement entropies.

For a state $\psi$ on $\M$ we define the \emph{R\'enyi entropy} of order $\alpha$    as
\begin{align}
        S_\alpha(\psi) := \frac{1}{1-\alpha}\log \Tr \rho^\alpha = \frac{1}{1-\alpha} \log \int_0^\oo \lambda_{\rho}(t)^\alpha\, dt,\quad \alpha \in(0,1)\cup(1,\infty).
\end{align}
The limit $\alpha \to 1$ of the R\'enyi entropy is the \emph{von Neumann entropy}
\begin{equation}
    S(\psi) \equiv S_1(\psi) = \tr \eta(\rho), \qquad \eta(t) = -t\log t.
\end{equation}
For $\alpha=0$, the R\'enyi entropy is defined as the logarithm of the rank of $\psi$: $S_0(\psi) = \log \tr s_\psi$.
The induced states $\psi$ and $\psi'$ of a unit vector $\Psi$ have the same entropies $S_\alpha(\psi)=S_\alpha(\psi')$, and this allows us to define the entanglement entropies for bipartite pure states:

\begin{definition}
    Let $\Psi\in\H$ be a unit vector.
    The \emph{R\'enyi entanglement entropy} of order $\alpha$ of $\Psi$ is 
    \begin{equation}
        S_\alpha(\Psi) := S_\alpha(\psi) = S_\alpha(\psi'),
    \end{equation} 
    where $\psi$ is the induced state on $\M$, and the \emph{generalized Schmidt rank} of $\Psi$ is
    \begin{equation}
        r(\Psi) := \Tr s_\psi = \Tr s_{\psi'}.
    \end{equation}
\end{definition}

As a special case of \cref{lem:f monotone}, we find that the R\'enyi entanglement entropies are entanglement monotones.
Moreover, \cref{thm:slocc}, together with the fact that a trace on a semifinite factor defines a dimension function on projections (see \cite[III.1.7.10]{blackadar_operator_2006}), implies that the Schmidt rank is even a monotone with respect to SLOCC transformations:

\begin{corollary}\label{cor:schmidt-rank}
    The R\'enyi entanglement entropy is an entanglement monotone:
    If $\Psi\xrightarrow{\overline{\LOCC}}\Phi$, then $S_\alpha(\M)_\Psi \geq S_\alpha(\M)_\Phi$, $\alpha \in [0,\oo)$.
    The Schmidt rank is a complete monotone with respect to SLOCC: 
    \begin{equation}
    \Psi\xrightarrow{\overline{\SLOCC}}\Phi \quad\iff\quad r(\Psi)\ge r(\Phi).
    \end{equation}
\end{corollary}

In the type $\I$ case, where $(\M,\M') =(\B(\H_A)\ox1,1\ox\B(\H_B))$ and $\Psi\in\H_A\ox\H_B$, the generalized Schmidt rank $r(\Psi)$ is the usual Schmidt rank of $\Psi$: 
It is the number of nonzero Schmidt coefficients $\lambda_\alpha\ge0$ in the Schmidt decomposition
\begin{equation}
    \Psi = \sum_\alpha \lambda_\alpha^{1/2} \ket{\alpha}_A\ket{\alpha}_B
\end{equation}
where $\{\ket\alpha_j\}\subset\H_j$, $j=A,B$ are suitable bases.
Therefore, \cref{cor:schmidt-rank} answers an open problem from \cite{van_luijk_schmidt_2024} asking for a generalization of the Schmidt rank that works for bipartite pure states outside of the type $\I$ case.

We conclude that entanglement monotones can be defined straightforwardly in the semifinite case, just as in finite dimensions. 
However, one has to be careful with the interpretation of the numerical values of entanglement monotones (see \cite{segal1960entropy} for a general discussion of entropy). 
As an example, suppose $\M$ is a type $\II_1$ factor and let $\tau$ be the tracial state (which equals the trace $\tau=\tr$ if the latter is normalized).
Then, as one would expect, for any normal state $\psi$ on $\M$, we have $\psi\succ\tau$ and hence
\begin{align}
    S_\alpha(\tau) \geq S_\alpha(\psi). 
\end{align}
However, the spectral scale $\lambda_\tau$ is simply the constant function $\lambda_\tau(t)=1$ on the interval $[0,\Tr(1)]=[0,1]$, which implies
\begin{align}
    S_\alpha(\tau) = 0 
\end{align}
and hence $S_\alpha(\psi) \leq 0$ for all normal states $\psi$ on $\M$. 
On the other hand, if $\rho$ is a density matrix on $\H$, we evidently have
\begin{align}
    S_\alpha(\rho) \geq 0.
\end{align}
How should we interpret the negativity of entanglement entropies? The answer is straightforward if we instead consider relative R\'enyi entropies. For concreteness, let us consider, for $\psi,\phi\in \M_*^+$, the \emph{Petz-R\'enyi relative entropies} defined as \cite{hiai2021quantum}
\begin{align}
    S_\alpha(\psi,\phi) := \frac{1}{\alpha-1} \log \Tr(\rho_\psi^\alpha \rho_\phi^{1-\alpha}),\qquad \alpha\in(0,1)\cup(1,\infty)
\end{align}
when $\alpha < 1$ or $\psi$ is absolutely continuous with respect to $\phi$ (otherwise we set $S_\alpha(\psi,\phi)=\infty$).
For a $d$-dimensional quantum systems modelled by the factor $M_d(\CC)$ we have
\begin{align}
    S_\alpha(\psi, \Tr_d) = - S_\alpha(\psi),\quad S_\alpha(\psi,\tau_d) = \log(d) - S_\alpha(\psi), 
\end{align}
where $\Tr_d$ denotes the standard trace on $M_d(\CC)$ with $\Tr_d(1)=d$ and where $\tau_d =\Tr_d(\placeholder)/d$. Thus
\begin{align}
    S_\alpha(\psi) = \log(d) - S_\alpha(\psi,\tau_d). 
\end{align}
In the type $\II_1$ case we have
\begin{align}
    S_\alpha(\psi) = - S_\alpha(\psi,\tau)
\end{align}
and $\log(\Tr(1))=0$ by convention. 
On the other hand, \cref{thm:equivalence-bell} shows that from any pure state on a standard bipartite system of type $\II_1$ factors, we can extract arbitrary amounts of entanglement. 
This matches \cref{sec:oo-qudits}, which shows that the purification $\Omega$ of $\tau$ corresponds to infinitely many Bell pairs. 
The \emph{maximally} entangled state $\Omega$ has R\'enyi entanglement entropy $S_\alpha(\Omega) = 0$, which should be considered as \emph{renormalized} version of the formally infinite \emph{physical} R\'enyi entanglement entropy $S\up{\text{ph.}}_\alpha(\Omega)=+\infty$.
The \emph{physical} R\'enyi entanglement entropy $S\up{\text{ph.}}_\alpha$ of \emph{any} pure state can be considered to be
\begin{align}\label{eq:physical entE}
    S\up{\text{ph.}}_\alpha(\Psi) = \infty + S_\alpha(\Psi) \leq S_\alpha\up{\text{ph.}}(\Omega).
\end{align}
Thus, the finite R\'enyi entanglement entropy $S_\alpha(\Psi)$ on the type $\II_1$ factor formally results from the subtraction of an infinite additive constant (cp.~\cite{longo2023continuous_entropy} for a similar discussion of entropy in the type $\II$ setting).

\clearpage
\appendix

\section{Commutative and noncommutative majorization theory}
\label{app:majorization}

In this appendix, we review commutative and noncommutative majorization theory in the $\sigma$-finite setting.
Majorization theory was initiated by Hardy-Littlewood-Polya in \cite{hardy1929} and has a wide range of applications (see, e.g.,  \cite{marshall_inequalities_2011} for an overview).
It was generalized to the non-commutative setting of semifinite von Neumann algebras in a series of works in the 1980s, see \cite{alberti_stochasticity_1982,kamei1983majorization,kamei1984double,petz1985scale,fack1986generalized,hiai_majorization_1987,hiai_majorizations_1987,hiai1989distance} and references therein. Our goal is to present a complete, but streamlined treatment of those aspects of majorization theory most significant for entanglement theory. 

We restrict our attention to majorization of \emph{positive} $L^1$ elements but allow for majorization $a\succ b$ of elements $a,b$ of different von Neumann algebras.
The important role played by doubly substochastic maps is emphasized.
Our guiding principle is that majorization theory for von Neumann algebras is essentially a classical, i.e., commutative, theory.
Indeed, almost all proofs in the noncommutative case are given by a reduction to the commutative case.
For this reason, we begin by reviewing majorization theory on $\sigma$-finite measure spaces in section \ref{sec:classical majorization}.
In \cref{sec:q majorization}, we consider majorization theory in the noncommutative setting by paralleling (and reducing to) the commutative case.

\subsection{Majorization theory on $\sigma$-finite measure spaces.}
\label{sec:classical majorization}

All measure spaces will be assumed to be $\sigma$-finite.
Whenever possible, we omit the $\sigma$-algebras, i.e., we denote a measure space $(X,\Sigma,\mu)$ simply by $(X,\mu)$.
Unless explicitly said otherwise, intervals $I\subset \RR$ will be equipped with the Lebesgue measure $dt$ and we simply write $I$ for $(I,dt)$.

We begin by recalling the concept of decreasing rearrangements.
For measurable functions $f:X\to \RR^+$, the distribution function is defined as
\begin{equation}
    D_f(t) := \mu( [f>t] )
\end{equation}
where $[f>t]:= f^{-1}((t,\oo))$.
The \emph{decreasing rearrangement} $f^\downarrow$ of $f$ is the right-continuous decreasing function $f^\downarrow : \RR^+\to \RR^+$ defined by
\begin{equation}
    f^\downarrow(r) = \inf\{t>0 : D_f(t)>r\} = D_{D_f}(r).
\end{equation}
By construction, $f^\downarrow$ is supported on $[0,\mu(\supp f))$ and equimeasurable with $f$, i.e., the level sets have the same measure $\mu([s\le f <t]) = |[ s\le f^\downarrow<t]|$ (where $|\placeholder|$ denotes the Lebesgue measure). This implies the identity
\begin{equation}\label{eq:equimeasurability}
    \int_X \phi\circ f\,d\mu = \int_0^\oo \phi(f^\downarrow(r))\,dr,
\end{equation}
valid for all measurable functions $\phi$ on $\RR^+$ so that the integrals exist.
In the following, we only consider integrable functions $f\in L^1(X,\mu)^+$.
The \emph{Lorenz curve} of $f$ is defined as the anti-derivative of the decreasing rearrangement:
\begin{equation}\label{eq:lorentz curve}
    L_f(t) := \int_0^t f^\downarrow(s)\,ds,\qquad t>0.
\end{equation}
It can be characterized directly via $f$ with the following variational expression
\begin{equation}\label{eq:variational formula}
    S_f(t) = \sup \bigg\{ \int_X e f\,d\mu : e\in L^1(X,\mu),\ 0\le e \le 1, \ \int_X e\,d\mu \le t \bigg\}.
\end{equation}
The majorization ordering is defined directly in terms of the Lorenz curves $L_f$:

\begin{definition}
    Let $(X,\mu)$ and $(Y,\nu)$ be measure spaces and let $f\in L^1(X,\mu)^+$ and $g\in L^1(Y,\nu)^+$.
    Then $f$ \emph{submajorizes} $g$, denoted $f\succ_w g$, if 
    \begin{equation}
        L_f(t)\ge L_g(t),\qquad t>0,
    \end{equation}
    and $f$ \emph{majorizes} $g$, denoted $f\succ g$, if $f$ submajorized $g$ and if $\int f\,d\mu=\int g\,d\nu$ holds.
\end{definition}

Majorization is intimately connected with the concept of stochastic maps.
Let $(X,\mu)$ and $(Y,\nu)$ be measure spaces.
A \emph{doubly substochastic} (DSS) operator from $(X,\mu)$ to $(Y,\nu)$ is a positive linear operator $T:L^\oo(X,\nu)\to L^\oo(Y,\nu)$ that is continuous for the respective weak* topologies, subunital $T(1)\le1$ and integral-non-increasing 
\begin{equation}\label{eq:int_pres}
    \int_Y T(f)\,d\nu \le \int_X f\,d\mu,\qquad f\in L^1(X,\mu)^+\cap L^\oo(X,\mu).
\end{equation}
The set of DSS operators from $X$ to $Y$ will be denoted $\DSS((X,\mu)\to (Y,\nu))$.
Eq.~\eqref{eq:int_pres} ensures that a DSS operator extends from $L^1\cap L^\oo$ to a linear contraction $L^1(X,\mu)\to L^1(Y,\nu)$.\footnote{\label{footnote}This follows from $\norm{f}_{L^1(X,\mu)} = \inf\{ \int (f_++f_-)\,d\mu : 0\le f_\pm, \, f=f_+-f_-\}$, which implies that the operator norm of a densely-defined linear operator $S:D(S)\subset L^1(X,\mu)\to L^1(Y,\nu)$ is $\norm S = \sup\{\norm{S(f)}_{L^1(Y,\nu)} : 0\le f\in D(S),\, \int f\,d\mu=1 \}$. This holds in general for linear maps $S:D(S)\subset B\to E$, if $B$ is a "base norm space", e.g., a noncommutative $L^1$ space, and $E$ is a normed space \cite{nagel1974order,lami2018non,jameson1970ordered}.}
Abusing notation, we denote this extension also by $T$.
For every DSS operator $T$ from $(X,\mu)$ to $(Y,\nu)$, there is a \emph{dual} DSS operator $T^*$ from $(Y,\nu)$ to $(X,\mu)$.
It is defined as the dual of the operator $T$ considered as an operator between the $L^1$ spaces, i.e., $T^*$ is the weak* continuous operator $L^\oo(Y,\nu)\to L^\oo(X,\mu)$ defined by
\begin{equation}
    \int_Y g\cdot T(f)\,d\nu = \int_X T^*(g)\cdot f\,d\mu
\end{equation}
for $f\in L^1(X,\mu)$ and $g\in L^\oo(Y,\nu)$.
It is clear that $T^*$ is DSS if $T$ is. 
Combining this fact with the variational formula \eqref{eq:variational formula} immediately shows that
\begin{equation}
    f\succ_w T(f), \qquad f\in L^1(X,\mu)^+,
\end{equation}
for all $T\in \DSS((X,\mu)\to(Y,\nu))$.
In fact, the converse is also true: $f\succ_w g$ if and only if $g=T(f)$ for some DSS operator as we will see below (see \cref{thm:maj}).

A doubly substochastic operator $T$ is \emph{doubly stochastic} (DS) if it unital, i.e., $T1=1$, and integral-preserving, i.e., eq.~\eqref{eq:int_pres} holds with equality.
The set of DS operators will be denoted $\DS((X,\mu)\to(Y,\nu))$.

\begin{lemma}\label{lem:IP on Omega}
    Let $T\in \DSS((X,\mu) \to(Y,\nu))$ and let $\Omega\subset X$.
    The following are equivalent:
    \begin{enumerate}[(a)]
        \item\label{it:IP on Omega1} $T$ is integral-preserving on $\Omega$, i.e., 
        \begin{equation}\label{eq:IP on Omega}
            \int_Y T(f)\,d\nu=\int_X f\,d\mu
        \end{equation} 
        for all $f\in L^1(X,\mu)$ with support in $\Omega$,
        \item\label{it:IP on Omega2} $T$ is integral-preserving on a single function $f\in L^1(X,\mu)^+$ with $\supp f = \Omega$.
        \item\label{it:IP on Omega3} $T^*$ is "unital onto $\Omega$": $T^*(1) |_{\Omega} \equiv 1$.
    \end{enumerate}
    Thus, $T$ is integral-preserving (resp.\ DS) if and only if $T^*$ is unital (resp.\ DS).
\end{lemma}
\begin{proof}
    \ref{it:IP on Omega1} $\Rightarrow$ \ref{it:IP on Omega2} is clear. \ref{it:IP on Omega3} $\Rightarrow$ \ref{it:IP on Omega1} follows from $\int T(f)\,d\nu = \int T^*(1) f\,d\mu$.
    \ref{it:IP on Omega2} $\Rightarrow$ \ref{it:IP on Omega3}:
    Let $f$ be as in \ref{it:IP on Omega1}.
    Since $1-T^*(1)\in L^\oo(X,\mu)^+$, we have
    \begin{align*}
        0\le \int (1-T^*(1)) f\,d\mu = \int f \,d\mu - \int T^*(1)f\,d\mu = \int f\,d\mu-\int T(f)\,d\nu =0,
    \end{align*}
    which implies that $T^*(1)\equiv 1$ on the support of $f$, which is $\Omega$.
\end{proof}

\begin{lemma}\label{lem:interconversion}
    Let $(X,\mu)$ be a  measure space and let $f\in L^1(X,\mu)^+$.
    There exists a DS operator $T$ from $(\supp f,\mu)$ to the interval $[0,\mu(\supp f))$ (with the Lebesgue measure) such that
    \begin{equation}\label{eq:interconversion}
        T(f) = f^\downarrow,\quad T^*(f^\downarrow)=f.
    \end{equation}
    Thus, there is a DSS operator $T$ from $(X,\mu)$ to $\RR^+$ such that \eqref{eq:interconversion} holds. 
\end{lemma}
\begin{proof}
    We may assume $f$ to be faithful, i.e., $\supp f=X$. Let $\Sigma$ be the $\sigma$-algebra on $X$ and let $\Sigma_f$ be the $\sigma$-subalgebra generated by the level sets of $f$.
    The conditional expectation (or coarse graining) is a DS operator $E$ from $(X,\Sigma,\mu)$ to $(X,\Sigma_f,\mu|_{\Sigma_f})$ such that $E(f)=f$ and $E^*(f)=f$.
    We can consider $[0,\mu(X))$ with the Borel $\sigma$-algebra $\mathrm B$ or with the $\sigma$-subalgebra $\mathrm B_{f^\downarrow}$ generated by $f^\downarrow$.
    By identifying the level sets of $f$ and $f^\downarrow$, we get a $\sigma$-algebra isomorphism $\Sigma_f\cong \mathrm B_{f^\downarrow}$ which maps the restricted measure $\mu|_{\Sigma_f}$ to the restricted the Lebesgue measure.
    This induces an isomorphism $\psi:L^\oo(X,\Sigma_f,\mu_0)\to L^\oo([0,\mu(X)),\mathrm B_{f^\downarrow},dt)$ such that $\psi(f)=f^\downarrow$ which satisfies $\psi^* = \psi^{-1}$.
    Finally, let $F$ be the conditional expectation of $\mathrm B_{f^\downarrow}$. Then $T=F^*\circ\psi\circ E$ satisfies the claim.
\end{proof}

\begin{theorem}[Submajorization]\label{thm:submaj}
    Let $(X,\mu)$ and $(Y,\nu)$ be  measure spaces and let $f\in L^1(X,\mu)^+$ and $g\in L^1(Y,\nu)^+$.
    The following are equivalent:
    \begin{enumerate}[(a)]
        \item\label{it:submaj1} $f\succ_w g$,
        \item\label{it:submaj2} $\int_X \phi\circ f\,d\mu \ge \int_Y \phi\circ g\,d\nu$ for all convex functions $\phi:\RR^+\to\RR^+$ with $\phi(0)=0$,
        \item\label{it:submaj3} $\int_X (f-t)_+\,d\mu \ge \int_Y (g-t)_+\,d\nu$ for all $t\ge0$,
        \item\label{it:submaj4} there exists a DSS operator $T$ from $(X,\mu)$ to $(Y,\nu)$ such that $T(f)=g$.
    \end{enumerate}
    The operator $T$ can be chosen to satisfy $T(\chi_{\supp(f)})\le \chi_{\supp(g)}$ and $T^*(\chi_{\supp(g)})\le \chi_{\supp(f)}$.
\end{theorem}

For the proof, we need the following result, which follows from the Banach-Alaoglu and Tychonov theorems (cp.\ \cref{lem:compactness}):
\begin{lemma}\label{lem:dss_are_compact_classical}
    Let $(X,\mu)$ and $(Y,\nu)$ be  measure spaces.
    Then $\DSS((X,\mu)\to (Y,\nu))$ is compact in the point-weak* topology, i.e., in the topology induced by the functionals 
    \begin{equation}
        T\mapsto \int_Y g\cdot T(f)\,d\nu,\qquad (f,g)\in L^\oo(X,\mu)\times L^1(Y,\nu).
    \end{equation}
    The subset of unital DSS operators is closed in this topology. 
    If both measure spaces are finite, the set of DS operators is closed.
\end{lemma}

\begin{proof}[Proof of \cref{thm:submaj}]
    Let us begin by noting the following identity
    \begin{equation}\label{eq:ac_identity}
        \int_X (f-t)_+ d\mu =\int_0^\oo (f^\downarrow(r)-t)_+\,dr = \int_0^s f^\downarrow(r)\,dr - s t, \qquad s\in [D_f(t^-),D_f(t)],
    \end{equation}
    where $D_f(t^-)$ denotes the left-limit of the (right-continuous) function $D_f$ at $t$.
    Furthermore, we note that each item implies 
    \begin{equation*}\label{eq:weight_ineq}
        \int_X f\,d\mu \ge \int_Yg\,d\nu.
    \end{equation*}
    We claim that we may assume $(X,\mu)=(Y,\nu)=(\RR^+,dx)$ and $f=f^\downarrow$ as well as $g=g^\downarrow$.
    Indeed, items \ref{it:submaj1} -- \ref{it:submaj3} are satisfied by $f^\downarrow$ and $g^\downarrow$ if and only if they are satisfied by $f$ and $g$ as is evident from equimeasurability \eqref{eq:equimeasurability}.
    Moreover, \cref{lem:interconversion} shows that \ref{it:submaj4} holds for $f$ and $g$ if and only if it holds for $f^\downarrow$ and $g^\downarrow$.
    From now on, we assume both measure spaces to be $\RR^+$ and $f,g$ to be decreasing right-continuous functions $L^1$ functions.

    \ref{it:submaj2} $\Rightarrow$ \ref{it:submaj3} is clear since, for each $t>0$, $\phi_t(x) = (x-t)_+$ is a convex function with $\phi_t(0)=0$.

    \ref{it:submaj3} $\Rightarrow$ \ref{it:submaj2}: Let $A=\lin_{\RR^+}\{\phi_t : t>0\}$.
    Note that item \ref{it:submaj2} holds for all $\phi\in A$ and that $A$ is closed under finite pointwise suprema.
    Since every convex function $\phi$ with $\phi(0)=0$ is a monotone limit of functions in $A$, item \ref{it:submaj2} follows from the monotone convergence theorem.

    \ref{it:submaj1} $\Rightarrow$ \ref{it:submaj3}:
    Let $t>0$. Then \eqref{eq:ac_identity} gives:
    \begin{multline*}
        \int_0^\oo (g(r)-t)_+\,dr = \int_0^{D_g(t)} g(r)\,dr - tD_g(t) \le \int_0^{D_g(t)}f(r)\,dr - t D_g(t)\\
        = \int_0^{D_g(t)} (f(r)-t)\,dr \le \int_0^{D_g(t)} (f(r)-t)_+\,dr \le \int_0^\oo (f(r)-t)_+\,dr.
    \end{multline*}

    \ref{it:submaj3} $\Rightarrow$ \ref{it:submaj1}:
    Let $t>0$ and pick $s>0$ such that $s\in [D_f(t^-),D_f(t)]$. Then  \eqref{eq:ac_identity} gives:
    \begin{multline*}
        \int_0^s g(r)\,dr
        =\int_0^s (g(r)-t)\,dr + ts
        \le \int_0^\oo(g(r)-t)_+\,dr + ts \\
        \le \int_0^\oo(f(r)-t)_+\,dr +ts  
        =\int_0^s f(r)\,dr.
    \end{multline*}

    \ref{it:submaj4} $\Rightarrow$ \ref{it:submaj3}:
    Let $t>0$, then $g-t = Tf-t \le T(f-t) \le T(f-t)_+$ implies $(g-t)_+ \le T(f-t)_+$ and hence
    \begin{align*}
        \int_0^\oo (g(r)-t)_+\,dr \le \int_0^\oo (T(f-t)_+)(r)\,dr \le \int_0^\oo (f(r)-t)_+\,dr.
    \end{align*}
    
    \ref{it:submaj1} $\Rightarrow$ \ref{it:submaj4}:
    We set $I\up n_m = [\frac{m-1}n,\frac mn)$, $m,n\in\NN$, to get a partition $I\up n = \{I\up n_m:m\in\NN\}$ of $\RR^+$.
    Set $f\up n_m= \int_{I_m\up n}f(s)\,ds$ and $g\up n_m=\int_{I_m\up n}g(s)\,ds$.
    By construction, $(f\up n_m)_m$ and $(g\up n_m)_m$ are in $\ell^1(\NN)$ with positive decreasing entries that satisfy $\sum_{m=1}^k f\up n_m \ge \sum_{m=1}^k g\up n_m$ for all $k,n\in\NN$ .
    Thus, there exists a doubly substochastic $n^2\times n^2$ matrix $[T\up n_{ml}]_{ml}$ such that $g\up n_m=\sum_{l=1}^{n^2} T\up n_{ml}f\up n_l$ for $m\le n^2$.
    Define a doubly substochastic operator $T\up n $ on $L^\oo(\RR^+)$ by
    \begin{equation}\label{eq:Tupn}
        T\up n h(r) =   n\sum_{m,l=1}^{n^2} \chi_{I_m\up n}(r) \, T\up n_{ml} \int_{I_l\up n} h(s)\,ds + \chi_{[n,\oo)}(r)h(r)
    \end{equation}
    (the prefactor of $n$ is the inverse of the Lebesgue measure of $I_m\up n$).
    By construction, we have $T\up n f\up n = g \up n$ where $f\up n = \sum_{m=1}^n \chi_{I_m\up n}(t) f\up n_m$, and similarly for $g$.
    By \cref{lem:dss_are_compact_classical}, the sequence $T\up n$ of doubly substochastic operators has a cluster point $\tilde T$ (w.r.t.\ the topology specified in \cref{lem:dss_are_compact_classical}).
    Since $(f,g) = \lim_n (f\up n,g\up n)$ in $L^1\times L^1$, we have $Tf=g$.

    Finally, we show that we may assume $T(\chi_{\supp(f)})\le \chi_{\supp(g)}$ and $T^*(\chi_{\supp(g)})=\chi_{\supp(f)}$.
    To this end, let $T'$ be the DSS operator from \ref{it:submaj4} and define a DSS operator $T$ by $T(h) = \chi_{\supp(g)}\,T'(\chi_{\supp(f)}h)$.
    Then $T(\chi_{\supp(f)}) = \chi_{\supp(g)}\cdot T'(1)\le \chi_{\supp(g)}$. Furthermore, since $T^*(h) =\chi_{\supp(f)}\,(T')^*(\chi_{\supp(g)}h)$, the same argument shows the claim for $T^*$.
\end{proof}

\begin{theorem}[Majorization]\label{thm:maj}
    Let $(X,\mu)$ and $(Y,\nu)$ be  measure spaces and let $f\in L^1(X,\mu)^+$ and $g\in L^1(Y,\nu)^+$ with $\int_X f\,d\mu=\int_Yg\,d\nu$.
    The following are equivalent:
    \begin{enumerate}[(a)]
        \item\label{it:maj1} $f\succ g$,
        \item\label{it:maj2} $\int_X \phi\circ f\,d\mu \ge \int_Y \phi\circ g\,d\nu$ for all convex functions $\phi:\RR^+\to\RR^+$ with $\phi(0)=0$,
        \item\label{it:maj3} $\int_X (f-t)_+\,d\mu \ge \int_Y (g-t)_+\,d\nu$ for all $t\ge0$,
        \item\label{it:maj4} there exists a DS operator $S\in\DS(\RR^+)$ such that $S(f^\downarrow)=g^\downarrow$.
        \item\label{it:maj5} there exists a DSS operator $T$ from $(X,\mu)$ to $(Y,\nu)$ such that $T(f)=g$ which satisfies $T(\chi_{\supp(f)})\le \chi_{\supp(g)}$ and $T^*(\chi_{\supp(g)})=\chi_{\supp(f)}$.
    \end{enumerate}
    If these equivalent properties hold, then $\mu(\supp f)\le \nu(\supp g)$. If $\mu(X)=\nu(Y)<\oo$, items \ref{it:maj1} to \ref{it:maj5} are further equivalent to
    \begin{enumerate}[resume*]
        \item\label{it:maj6} there exists a DS operator $T$ from $(X,\mu)$ to $(Y,\nu)$ such that $T(f)=g$.
    \end{enumerate}
\end{theorem}

For the proof, we need the following result on DS extensions of DSS maps:

\begin{lemma}\label{lem:DS extension}
    Let $T$ be a DSS operator from $(X,\mu)$ to $(Y,\nu)$ that is integral-preserving on $\Omega\subset X$ (cp.\ \cref{lem:IP on Omega}).
    The following are equivalent
    \begin{enumerate}[(a)]
        \item\label{it:DS extension1} There exists a $\tilde T\in \DSS((X,\mu)\to (Y,\nu))$ such that $\tilde T(f)=T(f)$ for all functions supported on $\Omega$
        \item\label{it:DS extension2} $\mu(X\setminus\Omega)=\int_Y (1-T(\chi_\Omega))\,d\nu$, where both sides may be infinite.
    \end{enumerate}
\end{lemma}

Formally, item \ref{it:DS extension2} is equivalent to $\mu(X)-\mu(\Omega) = \nu(Y)- \mu(\Omega)$, which is formally equivalent to $\mu(X)=\nu(Y)$.
While $\mu(X)=\nu(Y)$ is necessary for \ref{it:DS extension2} to hold, it is not sufficient if the measure spaces are not finite.
Sufficient conditions for \ref{it:DS extension2} are: (i) $\mu(X)=\nu(Y)<\oo$, and (ii) $\mu(X)=\nu(Y)=\oo$ and $\mu(\Omega)<\oo$.

\begin{proof}
    \ref{it:DS extension1} $\Rightarrow$ \ref{it:DS extension2} is clear. 
    \ref{it:DS extension1} $\Rightarrow$ \ref{it:DS extension2}: Let $(\Omega_n^c)$ be a partition of $X\setminus\Omega$ into sets of finite measure and let $(w_n)$ be a family functions in $L^1(X,\mu)^+$ with $\int w_n\,d\mu = \mu(\Omega_n^c)$ and $\sum_n w_n =1-T(\chi_\Omega)$
    (such a family of functions always exists; in the case, $\mu(X\setminus\Omega)<\oo$, these can simply be chosen as $X\setminus\Omega$ and $w = 1-T(\chi_\Omega)$.)
    A doubly stochastic map with the desired property is obtained as follows
    \begin{equation}
        \tilde T(h) := T(\chi_\Omega h ) + \sum_n  \,w_n \int_{\Omega_n^c}h \,\frac{d\mu}{\mu(\Omega_n^c)}.
    \end{equation}
    Indeed, $\tilde T(1) = T(\chi_\Omega) + \sum_n w_n =1$ and $\int T(h)\,d\nu = \int T(h\chi_\Omega)\,d\nu + \sum_n \mu(\Omega_n^c)^{-1} \int w_n\,d\nu \int_{\Omega_n^c} h\,d\mu = \int_\Omega h\,d\mu + \sum_n \int_{\Omega_n^c}h\,d\mu = \int h\,d\mu$.
\end{proof}

\begin{proof}
    By \cref{thm:submaj}, we have \ref{it:maj4} $\Rightarrow$ \ref{it:maj1}.
    For the proof, we introduce a weaker version of item \ref{it:maj5}:
    \begin{enumerate}[(e')]\it
        \item\label{it:maj5'} there exists a $T\in \DSS((X,\mu)\to (Y,\nu))$ such that $T(f)=g$ with $T(\chi_{\supp(f)})\le \chi_{\supp(g)}$ and $T^*(\chi_{\supp(g)})\le\chi_{\supp(f)}$.
    \end{enumerate}
    The equivalence of items \ref{it:maj1}, \ref{it:maj2}, \ref{it:maj3} and \ref{it:maj5'} follows from \cref{thm:submaj}.
    Equivalence of \ref{it:maj5'} and \ref{it:maj5} is seen as follows:
    Let $T'$ be the DSS operator from \ref{it:maj5'} and define a DSS operator $T$ by $T(h) = \chi_{\supp(g)}\,T'(\chi_{\supp(f)}h)$.
    Since $f$ is faithful on its support, the claim in \ref{it:maj5} follows from
    \begin{equation*}
        0\le \int_{\supp(f)} f\cdot (\chi_{\supp(f)}-T^*(\chi_{\supp(g)})\,d\mu = \int_{\supp(f)}f\,d\mu - \int_{\supp(g)} g\,d\nu=0.
    \end{equation*}

    \ref{it:maj1} $\Rightarrow$ \ref{it:maj4}:
    By \cref{thm:submaj}, there exists a DSS operator $T$ on $\RR^+$ such that $Tf^\downarrow=g^\downarrow$ and $T^*(\chi_{\supp g^\downarrow}) = \chi_{\supp f^\downarrow}$ and $T(\chi_{\supp f^\downarrow})\le\chi_{\supp g^\downarrow}$.
    By \cref{lem:IP on Omega}, $T$ is integral preserving on $\supp f^\downarrow$.
    If $|(\supp f^\downarrow)|=\mu(\supp f)<\oo$, the claim follows from \cref{lem:DS extension}.
    This leaves us with the case where $\supp f^\downarrow= \RR^+$.
    For this, we adapt the argument for "\ref{it:submaj1} $\Rightarrow$ \ref{it:submaj4}" in the proof of \cref{thm:submaj} to show that we can pick a \emph{unital} DSS map $S$ such that $S(f)=g$.
    We shall also use the notation from that proof.
    Since $\sum_{m=1}^{k} f\up n_m \ge \sum_{m=1}^{k} g\up n_m$, $k\in\NN$, there is, for every $n\in\NN$, a doubly stochastic $n^2\times n^2$ matrix $[S\up n_{ml}]$ such that $g_m\up n \le \sum_{l=1}^{n^2}S_{ml}\up n f_l\up n$ for all $m=1,\ldots n^2$ \cite[Thms.~II.1.10 \& II.2.8]{bhatia}.
    The operators $S\up n$ that defined in terms of these matrices as in \eqref{eq:Tupn} are doubly stochastic operators on $\RR^+$ and satisfy 
    \begin{align*}
        S\up n(f\up n) &= n\sum_{m,l=1}^{n^2}  S_{ml}\up n f\up n_l\, \chi_{I\up n_m} +  \chi_{[n,\oo)}  f\up n  =:\tilde g\up n\ge g\up n + \chi_{[n,\oo)} (f\up n -g\up n)
    \end{align*}
    Note that, since $\norm{\tilde g\up n-g\up n}_{L^1} \to 0$ as $n\to\oo$, $\tilde g\up n \to g$ as $n\to\oo$.
    By using \cref{lem:dss_are_compact_classical}, we can pick a cluster point $S$, which is guaranteed to be a unital DSS operator such that $S(f^\downarrow) =g^\downarrow$.
    By \cref{lem:IP on Omega}, the facts that $f^\downarrow$ is faithful and that $S(f^\downarrow)$ has the same integral as $f^\downarrow$ imply that $S$ is also integral-preserving and, hence, DS.

    If we assume $\mu(X)=\nu(Y)<\oo$, the equivalence \ref{it:maj5} $\Leftrightarrow$ \ref{it:maj6} is a direct consequence of \cref{lem:DS extension}.
\end{proof}

Even in the case $(X,\mu)=(Y,\nu)$, the equivalent properties of \cref{thm:maj} are insufficient to guarantee the existence of a DS operator from $(X,\mu)$ to $(Y,\nu)$ such that $T(f)=g$.

\begin{example}
    Consider $X=Y=\NN$ with the counting measure $\mu$.
    Let $g=(g_n)\in \ell^1$ be a strictly decreasing sequence of positive numbers $g_n>0$, e.g., $g_n=2^{-n}$, and put $f = (0,g_1,g_2,g_3,\ldots)$.
    Then $f^\downarrow=g^\downarrow=g$ and, hence, $f\prec g\prec f$.
    It is clear that there exists no DS operator $T$ such that $T(g)=f$ because the measure of the cosupport of $f$ is larger than the measure of the cosupport of $g$ (indeed, $\mu(\NN\setminus\supp(f))=1 > \mu(\NN\setminus\supp(g))=0$).
    Moreover, we claim that there exists no DS operator $T\in\DS(\NN)$ such that $T(f) = g$.
    To see this, note that a doubly (sub)stochastic map $T$ on $\ell^1$ is simply an infinite matrix $T\equiv [t_{nm}]_{n,m=1}^\oo$ with row sums and column sums (less than or) equal to one.
    If $T$ is a DSS matrix with $T(f)=g$, we have 
\begin{equation*}
    g_1 = T(f)_1=\sum_{m=2}^\oo t_{1m}g_{m-1} \le g_1 \sum_{m=2}^\oo t_{1m} \le g_1 \sum_{m=1}^\oo t_{1m} \le g_1.
\end{equation*}
Hence, each inequality is an equality, which implies that the first row of $T$ is $[t_{1m}]_{m=1}^\oo = [0,1,0,0,\ldots]$. Hence, the second column is $[t_{n2}]_{n=1}^\oo=[1,0,0,\ldots]^\top$. 
With this information, we repeat the same argument for the second row:  Since $g_2 = \sum_{m=3}^\oo t_{2m} g_{m-1} \le g_2 \sum_{m=3}^\oo t_{2m} \le g_2$, equality forces the second row of $T$ to be $[0,0,1,0,\ldots]$ and the third column to be $[0,1,0,\ldots]$. Iterating this, shows that $T$ equals the right-shift, i.e., $t_{nm}=\delta_{n+1,n}$, which is not doubly stochastic.
\end{example}

As far as we know, no general criterion deciding the existence of such a DS map is known.
A much easier problem is to decide whether a given DSS map $T$ with $T(f)=g$ admits a DS extension with the same property (cp.\ \cref{lem:DS extension}). 
Turning this into a criterion for the existence of a DS map $T$, however, requires a variation over \emph{all} DSS maps $T$ with $T(f)=g$ (see \cref{cor:index} below).

By \cref{thm:maj}, $f\succ g$ implies the existence of a DSS operator such that $T(f)=g$.
Since such an operator is integral-preserving on $\supp f$, \cref{lem:DS extension} implies that $T$ has a DS extension if and only if $\mu(X\setminus\supp f) = \int_Y (1-T(\chi_{\supp f}))\,d\nu = \nu(Y\setminus\supp g)+ \int_{\supp g} \int (1-T(\chi_{\supp f}))\,d\nu$.
As a consequence, we get:

\begin{corollary}\label{cor:index}
    Let $(X,\mu)$ and $(Y,\nu)$ be  measure spaces and let $f\in L^1(X,\mu)^+$ and $g\in L^1(Y,\nu)^+$ with $f\succ g$.
    The following are equivalent:
    \begin{enumerate}[(a)]
        \item\label{it:index1} There exists a DS operator $T$ from $(X,\mu)$ to $(Y,\nu)$ such that $T(f)=g$,
        \item\label{it:index2} There exists a DS map $T$ from $(X,\Sigma_f,\mu|_{\Sigma_f})$ to $(Y,\Sigma_g,\nu|_{\Sigma_g})$ such that $T(f)=g$, where $\Sigma_{f/g}$ is the $\sigma$-algebra generated by the level sets of $f/g$.
        \item\label{it:index3} There exists a doubly substochastic operator $T$ from $(X,\mu)$ to $(Y,\nu)$ such that $T(f)=g$ and 
        \begin{equation}\label{eq:index}
            \mu(X\setminus\supp f) =  \nu(Y\setminus \supp g) + \int_{\supp g} (1-T(\chi_{\supp f}))\,d\nu
        \end{equation}
    \end{enumerate}
    If these hold, then $\mu(X)=\nu(Y)$, $\mu(X\setminus\supp f) \ge \nu(Y\setminus \supp g)$ and $\mu(\supp f)\le \nu(\supp g)$.
\end{corollary}
\begin{proof}
    Equivalence of \ref{it:index1} and \ref{it:index2} is clear from the double stochasticity of the conditional expectation.
    \ref{it:index3} $\Leftrightarrow$ \ref{it:index1} follows from \cref{lem:DS extension} because the right-hand side of \eqref{eq:index} equals $\int_Y (1-T(\chi_{\supp f}))\,d\nu$ (since $T(\chi_{\supp f})$ is supported in $\supp g$).
\end{proof}

The example mentioned above shows that the three properties $\mu(X)=\nu(Y)$, $\mu(X\setminus\supp f)\ge \nu(Y\setminus \supp g)$ and $\mu(\supp f)\le \nu(\supp g)$ are \emph{not} sufficient to guarantee the existence of a doubly stochastic map $T$ such that $T(f)=g$.

\subsection{From commutative to noncommutative majorization theory}\label{sec:q majorization}

We review majorization theory on von Neumann algebras following the philosophy of regarding von Neumann algebras as noncommutative measure space.
Recall that in the commutative case, majorization depends on the specific choice of a faithful measure and not just on the measure equivalence class.
The noncommutative analog of this is given by the choice of a trace, typically denoted $\Tr$, on the von Neumann algebra.
To simplify things, we will, in the following, mean by "semifinite von Neumann algebra" $\M$ a tuple $(\M,\tr)$ of a $\sigma$-finite semifinite von Neumann algebra $\M$ together with choice of a normal semifinite faithful trace $\tr$.
To emphasize the corresponding algebra we sometimes write $\Tr_\M$.
We denote by $L^p(\M):=L^p(\M,\Tr)$ the $L^p$ spaces of $\M$ with respect to its trace.
In particular, $L^1(\M)$ is isomorphic to the predual $\M_*$ via the Radon-Nikodym map.

The notions of distribution function and decreasing rearrangements from commutative majorization theory have immediate noncommutative generalizations:
The distribution function of a positive element $\rho \in L^1(\M)^+$ is the right-continuous non-decreasing function 
\begin{equation}
    D_\rho(t):= \Tr( \chi_{[t,\oo)}(\rho) ), \qquad t>0,
\end{equation}
Following Petz \cite{petz1985scale}, the noncommutative analog of the decreasing rearrangement is called the \emph{spectral scale} and is defined by
\begin{equation}
    \lambda_\rho(t) := \inf \{s>0 : D_\rho(s)\le t\}.
\end{equation}
The spectral scale of $\rho \in L^1(\M)^+$ coincides with the \emph{generalized s-numbers} of Fack and Kosaki \cite{fack1986generalized}, which are also defined for non-positive operators (see also \cite[Sec.~4.2]{hiai_lecture_2021}). Since we only consider positive elements, we follow Petz's nomenclature.
If $\M = L^\oo(X,\mu)$ is an abelian von Neumann algebra with $\tr=\int\placeholder\,d\mu$, these definitions reduce to the distribution function and decreasing rearrangement of functions on a $\sigma$-finite measure space.
In particular, we have
\begin{equation}
    \tr \phi(\rho) = \int_{\RR^+} (\phi\circ \lambda_\rho)(t)\,dt
\end{equation}
We will often consider the spectral scale as an element of $L^1(\RR^+)$.
Based on the variational expression \eqref{eq:variational formula}, the \emph{Lorenz curve} of an element $\rho\in L^1(\M)^+$ is defined as the function
\begin{equation}\label{eq:nc lorentz curve}
    L_\rho(t) := \sup \big\{ \tr(\rho x) : x=x^*\in \M,\ 0\le x\le 1,\ \tr x\le t\}.
\end{equation}
We will see below that this is indeed simply the anti-derivative of the spectral scale (cp.\ \cref{cor:marjorization is classical}).
Equipped with a definition of Lorenz curves, we define majorization and submajorization as in the commutative case:

\begin{definition}
    Let $\M$ and $\N$ be semifinite von Neumann algebras and let $\rho\in L^1(\M)^+$ and $\sigma \in L^1(\N)^+$. 
    We say that $\rho$ \emph{submajorizes} $\sigma$, denoted $\rho\succ_w \sigma$, if $L_\rho(t)\ge L_\sigma(t)$ for all $t\ge0$, and that $\rho$ \emph{majorizes} $\sigma$, denoted $\rho\succ \sigma$, if $\rho\succ_w \sigma$ and $\Tr \rho = \Tr \sigma$.
\end{definition}

Next, we consider doubly substochastic operators. The definition from the commutative case immediately generalizes:

\begin{definition}
Let $\M$ and $\N$ be semifinite von Neumann algebras. A normal, completely positive map $T:\M\to\N$ is called \emph{doubly substochastic} (DSS) if $T(1_\M) \leq 1_\N$ and 
\begin{align}
    \Tr_\N\circ T(\rho) \leq \Tr_\M(\rho),\qquad \rho\in L^1(\M)^+\cap \M.
\end{align} 
We denote the set of such maps by $\DSS(\M\to\N)$. A DSS map is \emph{doubly stochastic} (DS) if it is unital and trace-preserving.
\end{definition}

As in the classical case, $T$ extends to a linear contraction $L^1(\M) \to L^1(\N)$ (see \cref{footnote}), which we also denote by $T$.
The normal linear map $T^*:\N \to \M$ such that
\begin{equation}
    \Tr_\N T(\rho)y  = \Tr_\M \rho T^*(y), \qquad \rho\in L^1(\M)^+,\ y\in\N^+,
\end{equation}is DSS (resp.\ DS) if $T$ is and will be called the DS map dual to $T$.
Since $T^*$ is DSS, the variational expression \eqref{eq:nc lorentz curve} immediately gives
\begin{equation}
    \rho \succ_w T(\rho), \qquad \rho\in L^1(\M)^+,
\end{equation}
for every $T\in \DSS(\M\to\N)$.

If $p$ is a projection in a semifinite von Neumann algebra, then the corner $\M_p:=p\M p$ is a semifinite von Neumann algebra, naturally equipped with the restricted trace $\tr_{\M_p} = \tr_\M\upharpoonright \M_p$.
A subalgebra $\N\subset\M$ of a semifinite von Neumann algebra $\M$ is called \emph{semifinite} if $\tr_\N:=\tr_\M \upharpoonright \N$ is a semifinite trace on $\N$.\footnote{Not every subalgebra of semifinite von Neumann algebra is semifinite, e.g., $\B(\H_A)\ox 1\subset\B(\H_A\ox\H_B)$ is semifinite if and only if $\H_B$ is finite-dimensional.}

\begin{lemma}\label{lem:sf_subalg}
    \begin{enumerate}[(1)]
        \item\label{it:sf_subalg1}
            The inclusion $j:\N\hookrightarrow\M$ of a semifinite subalgebra is DS and the dual DS map $j^*:\M\to\N$ is the (unique) trace-preserving conditional expectation.
        \item\label{it:sf_subalg2}
            If $p$ is a projection in $\M$, the natural inclusion $j:\M_p\hookrightarrow\M$ is DSS, and the dual DSS map $j^*:M\to \M_p$ is the cut-down given by $j^*(x)= pxp$.
        \item\label{it:sf_subalg3}
            Let $\rho=\rho^*\in L^1(\M)$ and let $\A\subset \M_{\supp \rho}$ be the abelian subalgebra generated by $\rho$.
            Then $\A$ is a semifinite subalgebra of $\M_{\supp \rho}$.\footnote{Note that the (unital) abelian algebra that is generated by $\rho$ in $\M$ is not semifinite in general (it is semifinite if and only if $\rho$ has finite co-support, i.e., if $\tr(1-\supp \rho)<\oo$.}
    \end{enumerate}
\end{lemma}
\begin{proof}
    \ref{it:sf_subalg1}:
    Since $j$ is clearly doubly stochastic the same holds for $E$.
    Let $x\in L^1(\N)$ then $\tr_\N[y E(x)] = \tr_\M[j(y)x] = \tr_\M[yx]=\tr_\N[yx]$ for all $y\in\N$ shows $E(x)=x$. This proves $E|_\N =\id_\N$.
    Item \ref{it:sf_subalg2} follows similarly.

    \ref{it:sf_subalg3}:
    Without loss of generality, assume $p=\supp(\rho)=1$.
    Elements of $\A$ are of the form $f(\rho)$ for bounded Borel functions $f:\Sp(\rho)\to\RR$. Let $f\ge0$ and set $f_\eps = f\cdot \chi_{\Sp(\rho)\setminus(-\eps,\eps)}$ for $\eps>0$.
    The function $f_\eps$ converges pointwise from below to the function $f_0 = f \cdot \chi_{\Sp(\rho)\setminus\{0\}}$. Thus, the functional calculi converge ultraweakly to $f_0(\rho)$ from below. 
    However, $f_0(\rho)=f(\rho)$ since $\{0\}$ is a null set of the spectral measure of $\rho$ and since $\rho\in L^1(\M)$, we know that $\tr(f_\eps(\rho))<\infty$. Thus, every element in $\A^+$ can be approximated in the ultraweak topology by elements with finite trace, i.e., $ \tr\upharpoonright\A$ is a semifinite trace.
\end{proof}

Equipped with this result, we can show that every individual $L^1$ element of a semifinite von Neumann algebra is equivalent to a function a measure space from the point of view of majorization theory.
If $(X,\mu)$ is a measure space, we implicitly understand that $L^\oo(X,\mu)$ is equipped with the trace $\int\placeholder\,d\mu$ and we write $\DSS((X,\mu)\to\N)$ and $\DSS(\N\to(X,\mu))$ for the corresponding sets of DSS maps.

\begin{corollary}\label{cor:marjorization is classical}
    Let $\M$ be a semifinite von Neumann algebra and let $\rho=\rho^*\in L^1(\M)$. Then:
    \begin{enumerate}[(i)]
        \item\label{it:majorization is classical1}
            There exists a measure space $(X,\mu)$, a function $f\in L^1(X,\mu;\RR)$ and $T\in \DSS((X,\mu)\to\M)$, $S\in \DSS(\M\to (X,\mu))$ such that $S(\rho)=f$ and $T(f)=\rho$.
        \item\label{it:majorization is classical2}
            Assume $\rho\ge0$ and let $(X,\mu)$ and $f$ be as in \ref{it:majorization is classical1}. Then $\rho \prec f \prec \rho$ and 
            \begin{equation}
                \lambda_\rho(t) = f^\downarrow(t) \qandq L_\rho(t) = \int_0^t \lambda_\rho(s)\,ds = L_f(t).
            \end{equation}
    \end{enumerate}
    \end{corollary}
\begin{proof}
    \ref{it:majorization is classical2} follows from     \ref{it:majorization is classical1}.
    \ref{it:majorization is classical1}:
    We set $X=\Sp(\rho)\setminus\{0\}$, $\mu(\Omega)=\tr \chi_\Omega(\rho)$ and $f(t)=t$.
    The functional calculus gives an isomorphism between $L^\oo(X,\mu)$ and the abelian subalgebra of $\A\subset \M_{\supp \rho}$ generated by $\rho$ which identifies the integral with respect to $\mu$ with the trace on $\A$. The result follows from \cref{lem:sf_subalg}.
\end{proof}

As a direct consequence of the classical result \cref{thm:submaj}, we then obtain the main result of this section:

\begin{theorem}[Submajorization]\label{thm:q submaj}
    Let $\M$ and $\N$ be semifinite von Neumann algebras and let $\rho\in L^1(\M)^+$ and $\sigma\in L^1(\N)^+$. The following are equivalent:
    \begin{enumerate}[(a)]
        \item\label{it:q submaj1} $\rho \succ_w \sigma$,
        \item\label{it:q submaj2} $\lambda_\rho \succ_w \lambda_\sigma$, 
        \item\label{it:q submaj3} $\Tr \phi(\rho) \ge \Tr \phi(\sigma)$ for all convex functions $\phi:\RR^+\to\RR^+$ with $\phi(0)=0$,
        \item\label{it:q submaj4} there exists a DSS map $T$ from $\M$ to $\N$ such that $T(\rho)=\sigma$.
    \end{enumerate}
    The operator $T$ can be chosen to satisfy $T(\supp \rho)\le \supp \sigma$ and $T^*(\supp \sigma) \le \supp \rho$.
\end{theorem}

If one additionally assumes that $\tr \rho =\tr \sigma$, one gets a majorization theorem for elements of von Neumann algebras stating that $\rho$ majorizes $\sigma$ if and only if majorization holds on the level of spectral scales, i.e., 
\begin{equation}
    \rho\succ \sigma \quad\iff\quad \lambda_\rho \succ \lambda_\sigma.
\end{equation}
This is, in turn, equivalent to the existence of a DSS operator $\M\to\N$ such that $T(\rho)=\sigma$.
As discussed in \cref{sec:classical majorization}, in general, $\rho\succ \sigma$ does not imply the existence of a DS map that takes $\rho$ to $\sigma$.
In the case that $\M=\N$ is a semifinite \emph{factor}, Hiai \cite[Thm.~2.5]{hiai_majorization_1987} showed that majorization can be characterized through mixtures of unitaries (just as in finite-dimensions).
This will be considered in \cref{sec:majorization in factors} below.

\subsection{Properties of doubly substochastic maps between von Neumann algebras}

\begin{lemma}
    Let $T\in\DSS(\M\to \N)$ and let $\phi:\RR\to\RR^+$ be a convex function with $\phi(0)=0$.
    Then
    \begin{equation}\label{eq:berezin_lieb}
        \tr \phi(Ta) \le \tr \phi(a),\qquad a=a^*\in\M,
    \end{equation}
    where both sides may be infinite. 
    If $T$ is doubly stochastic, \eqref{eq:berezin_lieb} holds for all convex functions $\phi:\RR\to\RR^+$.
\end{lemma}
\begin{proof}
    Let $T\in \DSS(\M\to\N)$, $a=a^*$, and consider $\phi(t) = (t-\lambda)_+$ with $\lambda>0$.
    Then $Ta-\lambda1_\N \le T(a-\lambda1_\M) \le  T(a-\lambda1_\M)_+ = T\phi(a)$ and the fact that $x\mapsto x_+$ is a monotone function on $\RR$ imply
    \begin{equation*}
        \tr_\N \phi(Ta) = \tr_\N (Ta-\lambda1_\N)_+ \le \tr_\N (T\phi(a))_+ = \tr_\N T\phi(a) \le \tr_\M\phi(a),
    \end{equation*}
    where we used that $\phi(a)$ is positive.
    The same argument applies to $\phi(t)=( (-t)-\lambda)_- = (t+\lambda)_-$.
    Clearly \eqref{eq:berezin_lieb} extends to finite sums $\phi(t)=\sum\phi_i(t)$ of functions of the form $\phi_i(t)=(t\mp \lambda_i)_\pm$ with $\lambda_i\ge0$.
    Normality further extends
    \eqref{eq:berezin_lieb} to monotone limits.
    Since all convex functions $\RR\to\RR^+$ with $\phi(0)=0$ can be obtained in this way, the first claim follows.
    Now assume $T\in \DS(\M\to \N)$.
    We establish the claim for $\phi(t) = (\alpha a+\beta)_+$, $\alpha,\beta\in\RR$:
    \begin{align*}
        \tr_\N \phi(Ta) = \tr_\N (\alpha Ta+ \beta1))_+ = \Tr_\N (T(b))_+ \le \tr_\M b_+ = \tr_\M \phi(a)
    \end{align*}
    where $b=\alpha a +b=b^*\in\M$.
    As before, \eqref{eq:berezin_lieb} extends to all monotone limits of finite sums of such functions.
    In this way, we obtain all convex functions $\RR\to\RR^+$.
\end{proof}

In particular, since $\phi(t)=\abs t^p$ is a convex function $\RR\to\RR^+$ for $p\ge1$, we obtain 
\begin{equation}\label{eq:Lp_contractive}
    \tr |Ta|^p \le \tr |a|^p
\end{equation}
for every $T\in\DSS(\M\to\N)$ and all self-adjoint $a\in\M$.
Since completely positive contractions satisfy $\abs{Ta}\le T\abs a$, $a=a^*\in\M$, \eqref{eq:Lp_contractive} extends to non-hermitian elements.
Thus, every DSS map $T$ yields a bounded operator $T:L^p(\M)\to L^p(\N)$ for all $p\ge1$.

\begin{lemma}\label{lem:compactness}
    Let $\M$ and $\N$ be semifinite von Neumann algebras.
    Then the convex set $\DSS(\M\to\N)$ is compact in the point ultraweak topology, i.e., the initial topology determined by the maps $T\mapsto \tr b T(a)$ with $b\in \M$ and $a\in L^1(\N)$.
    The convex subset of unital DSS maps is closed in this topology.
\end{lemma}
\begin{proof}
    Since the space contractions $\M\to\N$ is compact in the point-ultraweak topology, we only have to show that $\DSS(\M\to\N)$ is a closed subset.%
    \footnote{This follows from the Tychonov theorem: The space of linear contractions $\M\to\N$ with the point-ultraweak topology can be regarded as a closed subset of the product space $\prod_{x\in B(\M)} B(\N)$, where $B$ denotes the unit ball with the ultraweak topology. By the Banach-Alaoglu theorem, $B(\N)$ is ultraweakly compact and by Tychonov's theorem, the product space is compact. Thus, the linear contractions are point-ultraweakly compact.}
    Let $T_\alpha\in \DSS(\M\to\N)$, $\alpha\in I$, converge in the point-ultraweak topology to some contraction $T$.
    As pointwise limits of contractive cp maps are cp, $T$ is cp and, hence, subunital $T1_\M \le 1_\N$.
    Since the trace is ultraweakly lower semicontinuous on the positive cone $\N^+$, we further have
    $\tr[Tx] \le \liminf_\alpha \tr[T_\alpha x] \le \lim_\alpha \tr x =\tr x$ for all $x\in\M^+$.
    Therefore, $T\in \DSS(\M\to\N)$.
    If each $T_\alpha$ is unital, then $T$ is unital: $T1_\M = \lim_\alpha T_\alpha 1_\M=\lim_\alpha 1_\N =1_\N$.
    If the traces are finite, they are elements of the predual and, hence, ultraweakly continuous so that $\tr[Tx] = \lim_\alpha \tr[T_\alpha x] = \lim_\alpha \tr[x]=\tr[x]$, $x\in\M$.
\end{proof}

We need the following Lemma:

\begin{lemma}
    Let $a\in \M^+$ be such that $\tr a=\oo$. Then, for every $\alpha>0$ there exist $a_n\in \M^+$, $n\in \NN$, such that $\tr a_i = \alpha$ and $\sum_i a_i = a$.
\end{lemma}

While doubly substochastic maps exist between arbitrary semifinite von Neumann algebras (take, for instance, $T=0$), this is quite different for doubly stochastic maps:
\begin{lemma}\label{lem:existence_of_ds_maps}
    $\DS(\M\to \N)\ne \emptyset$ if and only if $\tr1_\M = \tr 1_\N$.
\end{lemma}

\begin{proof}
    Clearly $T\in \DS(\M\to \N)$ implies $\tr1_\N = \tr T(1_\M) = \tr 1_\M$.
    For the converse, let us first consider the case where both traces are finite. Then, we can define a doubly stochastic map by $T = (\tr1_\N)^{-1} \tr(\placeholder) 1_\N$.
    If the traces are infinite, we can pick $p_i\in \M^+$ and $q_i\in \N^+$ such that $\sum p_i=1_\M$ and $\sum q_i=1_\N$ and such that $\tr p_i = \tr q_i=\alpha>0$.
    Set $T = \sum_i \alpha^{-1} q_i \tr p_i(\placeholder)$.
    We have $T1 = \sum_i q_i \alpha\tr p_i =\sum q_i =1$ and $\tr T a = \sum_i \alpha^{-1} \tr q_i \tr p_i a = \sum_i \tr p_i a = \tr a$.
\end{proof}

We note the following noncommutative generalization of \cref{lem:DS extension}:

\begin{lemma}\label{lem:DS extension2}
    Let $T\in \DSS(\M\to \N)$ be trace-preserving on $\M_p$ for some projection $p\in\M$ (i.e., $\tr Ta = \tr a$ for all $a\in L^1(\M)$ with $\supp a \le p$).
    There exists a $\tilde T\in \DS(\M\to \N)$ with $\tilde T\upharpoonright\M_p = T\upharpoonright\M$ if and only if
    \begin{equation}
        \tr (1-p) = \tr(1 - T(p))
    \end{equation}
    where both sides may be infinite.
\end{lemma}
\begin{proof}
    The "only if" part is trivial. 
    For the converse, we adapt the proof of \cref{lem:DS extension}.
    If both sides are finite, we may set $\tilde T(a) = T(pa) (1-T(p)) + (\tr(1-p))^{-1}\tr((1-p)a)$.
    If both sides are infinite, we pick finite projections $p_n^c$ such that $\sum_n p_n^c = (1-p)$ and functions $w_n\in L^1(\N)^+$ such that $\Tr w_n =\Tr p_n^c$ and $\sum_n w_n = (1-T(p))$ and set $\tilde T(a) = T(pa) + \sum_n w_n \tr(p_n^c a)$.
\end{proof}

In the noncommutative setting, we defined doubly (sub)stochastic maps to be completely positive.
The reason for this are the applications to quantum mechanics (where complete positivity is enforced by the statistical interpretation) considered in the main text of this article.
However, in the literature on majorization theory on von Neumann algebras, doubly substochastic maps are usually only assumed to be positive maps.
In the following, we refer to the two classes as "positive DSS maps" and "completely positive DSS maps".
While the class of completely positive DSS maps is strictly smaller than the class of positive DSS maps (e.g., the transposition $a\mapsto a^\top$ on $M_n(\CC)$ is a positive DS map but not a completely positive DS map), they implement the same set of transitions:

\begin{lemma}\label{lem:WLOG cp}
    Let $a=a^*\in L^1(\M)$ and $b=b^*\in L^1(\N)$. The following are equivalent:
    \begin{enumerate}[(i)]
        \item\label{it:WLOG cp1} there exists a completely positive DSS (resp.\ DS) map $T$ with $T(a)=b$,
        \item\label{it:WLOG cp2} there exists a positive DSS (resp.\ DS) map $T$ with $T(a)=b$.
    \end{enumerate}
\end{lemma}
\begin{proof}
    \ref{it:WLOG cp1} $\Rightarrow$ \ref{it:WLOG cp2} is clear.
    \ref{it:WLOG cp2} $\Rightarrow$ \ref{it:WLOG cp1}:
    Let $\A\subset \M_{\supp a}$ be the abelian subalgebra generated by $a$ and let $E:\M_{\supp a} \to \A$ be the trace-preserving conditional expectation (cp.\ \cref{lem:sf_subalg}).
    Note that $T_0 := T\upharpoonright\A$ is completely positive because $\A$ is abelian.
    Since the conditional expectation is completely positive as well, $S = (T_0\circ E)((\supp a)\placeholder(\supp a))$ is a completely positive DSS map $\M_{\supp a}\to \N$ with $S(a) = b$.
    This shows the case where $T$ is DSS.
    If $T$ is DS, the map $S$ is by construction trace-preserving on $\M_{\supp a}$.
    We use \cref{lem:DS extension2} to show that $S$ admits a completely positive DS extension:
    \begin{equation*}
        \tr(1-S(\supp a))= \tr(1-T(\supp a)) = \tr T(1-\supp a) = \tr(1-\supp a).
    \end{equation*}
\end{proof}

\subsection{Distance of unitary orbits and majorization in factors}\label{sec:majorization in factors}

Hiai \cite[Thm.~2.5]{hiai_majorization_1987} showed that majorization in semifinite factors can be characterized through mixtures of unitaries (just as in finite-dimension).
Combining his result with \cref{thm:q submaj}, we get:

\begin{theorem}[Majorization in semifinite factors]\label{thm:factor majorization}
    Let $\M$ be a semifinite factor and let $\rho,\sigma\in L^1(\M)^+$.
    The following are equivalent
    \begin{enumerate}[(a)]
        \item\label{it:factor majorization1} $\rho\succ \sigma$, i.e., $L_\rho(t)\ge L_\sigma(t)$ for all $t\ge0$ and $\tr\rho = \tr\sigma$,
        \item\label{it:factor majorization3} $\sigma \in \overline{\conv}\{u\rho u^* : u\in\U(\M)\}$,
        \item\label{it:factor majorization4} $\tr \rho = \tr \sigma$ and $\sigma \in \overline{\conv}\{v\rho v^* : \norm v\le 1\}$,
        \item\label{it:factor majorization5} $\tr \rho = \tr \sigma$ and $\tr \phi(\rho) \ge \tr \phi(\sigma)$ for all convex functions $\phi:\RR^+\to\RR^+$, 
        \item\label{it:factor majorization2} $\lambda_\rho \succ \lambda_\sigma$,
        \item\label{it:factor majorization6} there exists a DS operator $S$ on $\RR^+$ such that $S(\lambda_\rho) = \lambda_\sigma$.
    \end{enumerate}
    If $\M$ is finite, i.e., $\tr 1<\oo$, this is equivalent to the existence of a doubly stochastic operator $T\in \DS(\M)$ with $T(\rho)=\sigma$.
\end{theorem}

In \cref{it:factor majorization3,it:factor majorization4}, the closure is taken in the norm topology of $L^1(\M)$.

Next, we combine independent results of Hiai-Nakamura \cite{hiai1989distance} and Haagerup-St\o rmer \cite{haagerup1990equivalence} on the $L^1$-distance of unitary orbits.
We begin with the following:

\begin{proposition}\label{prop:lower bound}
    Let $\M$ be a semifinite von Neumann algebra and let $\rho,\sigma \in L^1(\M)^+$.
    Then
    \begin{equation}\label{eq:lower bound}
        \norm{\rho - \sigma}_{L^1(\M)} \ge \norm{\lambda_\rho - \lambda_\sigma}_{L^1(\RR^+)} = \norm{D_\rho - D_\sigma}_{L^1(\RR^+)}.
    \end{equation}
\end{proposition}

To show this, we adapt an argument of Haagerup-St\o rmer from \cite{haagerup1990equivalence} which yields the following general fact:

\begin{lemma}\label{lem:trick}
    Let $\M$ and $\N$ be von Neumann algebras and let $Q:\M_*^+ \to \N_*^+$ be a map (not necessarily linear) which is monotone and normalization-preserving, i.e., $\omega\le\varphi$ implies $Q(\omega)\le Q(\varphi)$ and $Q(\omega)(1)=\omega(1)$.
    Then
    \begin{equation}
        \norm{Q(\omega)-Q(\varphi)}\le \norm{\omega-\varphi}\qquad\omega,\varphi\in\M_*.
    \end{equation}
\end{lemma}
\begin{proof}
    We follow the proof of \cite[Lem.~4.2]{haagerup1990equivalence}.
    To express the norms $\norm{\omega-\varphi}$ and $\norm{Q(\omega)-Q(\varphi)}$, we use the maximum of two normal positive linear functionals.
    For self-adjoint $\psi\in\M_*$, we denote by $\psi^\pm\in\M_*^+$ its positive and negative parts (satisfying $\psi=\psi^+-\psi^-$ and $\norm\psi = \psi^+(1)+\psi^-(1)$).
    The maximum of $\omega,\varphi\in\M_*^+$ is the functional $\omega\vee\varphi:=\omega+(\omega-\varphi)^-=\varphi+(\omega-\varphi)^+ \in \M_*$.
    It follows that $\omega\vee\varphi\ge\omega,\varphi$ and by \cite[Lem.~2.7]{haagerup1990equivalence}, it holds
    \begin{equation}\label{eq:norm_vee}
        \norm{\omega-\varphi} = 2(\omega\vee\varphi)(1)-\omega(1)-\varphi(1) = 2 \inf_{\substack{\rho\in\M_*^+\\\rho\ge \varphi,\omega}} \rho(1) -\omega(1)-\varphi(1) .
    \end{equation}
    Since $Q$ is monotone, we have $Q(\omega\vee\varphi)\ge Q(\omega),Q(\varphi)$.
    Applying \eqref{eq:norm_vee}, we get 
    \begin{multline}
        \norm{Q(\omega) - Q(\varphi)} 
        \le 2\,Q(\omega\vee\varphi)(1) - Q(\omega)(1)-Q(\varphi)(1) \nonumber= 2\,(\omega\vee\varphi)(1)-\varphi(1)-\omega(1) = \norm{\omega-\varphi}.
    \end{multline}
\end{proof}

\begin{proof}[Proof of \cref{prop:lower bound}]
    Up to the identifications $L^1(\M)\cong \M_*$ and $L^1(\RR^+)\cong L^\oo(\RR^+)_*$, the map $Q:L^1(\M)^+\to L^1(\RR^+)$, $Q(\rho) = D_\rho$ satisfies the assumption of \cref{lem:trick}, which proves the inequality $\norm{\rho-\sigma}_{L^1}\ge \norm{D_\rho-D_\sigma}_{L^1}$.
    Since $D_{D_\rho} = \lambda_\rho$ and $D_{\lambda_\rho}=D_\rho$, the same argument proves $\norm{\lambda_\rho-\lambda_\rho}_{L^1} \le \norm{D_\rho-D_\sigma}_{L^1}\le \norm{\lambda_\rho-\lambda_\sigma}_{L^1}$, which finishes the proof.
\end{proof}

We can use \cref{prop:lower bound} to show that the `rank' is a lower-semicontinuous function on $L^1(\M)$:

\begin{lemma}\label{lem:rank lsc}
    Let $\M$ be a semifinite von Neumann algebra.
    If $(\rho_n)$ is a Cauchy sequence in $L^1(\M)$, then
    \begin{equation}
        \tr \big(\supp(\lim_n \rho_n)\big) \le \liminf_n \tr\big(\supp(\rho_n)\big).
    \end{equation}
\end{lemma}
\begin{proof}
    By \cref{prop:lower bound}, the assumption implies that the spectral scales $\lambda_{\rho_n}$ form a Cauchy sequence in $L^1(\RR^+)$.
    Since $\tr(\supp(\rho))= \tr(\chi_{(0,\oo)}(\rho))$, we have $\tr(\supp(\rho))= |\supp(\lambda_\rho)|$.
    Therefore, the result follows from the lower-semicontinuity of the Lebesgue measure of the support on $\RR^+$.
\end{proof}

Since both $\rho\mapsto \lambda_\rho$ and $\rho\mapsto D_\rho$ are unitarily invariant maps, the estimate \eqref{eq:lower bound} yields a lower bound on the distance of the unitary orbits of $\rho$ and $\sigma$.
In the case where $\M$ is a factor, Hiai-Nakamura and Haagerup-Størmer showed that this estimate is, in fact, an equality \cite{hiai1989distance,haagerup1990equivalence}:

\begin{theorem}
    Let $\M$ be a semifinite factor and let $\rho,\sigma\in L^1(\M)^+$.
    Then 
    \begin{equation}\label{eq:unitary orbits L1 dist}
        \inf_{u\in \U(\M)} \norm{\rho - u\sigma u^*}_{L^1(\M)} = \norm{\lambda_\rho-\lambda_\sigma}_{L^1(\RR^+)} = \norm{D_\rho - D_\sigma}_{L^1(\RR^+)}.
    \end{equation}
\end{theorem}

The identity \eqref{eq:unitary orbits L1 dist} is generalized in \cite{hiai1989distance} also for the $L^p$ distance, yielding $\inf_{u\in \U(\M)} \norm{\rho - u\sigma u^*}_{L^p(\M)} = \norm{\lambda_\rho-\lambda_\sigma}_{L^p(\RR^+)}$.

\addcontentsline{toc}{section}{References}
\printbibliography

\end{document}